\documentclass[onecolumn,journal, twocolumn]{IEEEtran}
\usepackage[latin9]{inputenc}
\usepackage{xcolor}
\usepackage{float}
\usepackage{mathtools}
\usepackage{algorithm2e}
\usepackage{amsmath}
\usepackage{amsthm}
\usepackage{amssymb}
\usepackage{graphicx}

\makeatletter

\providecommand{\tabularnewline}{\\}

\let\oldforeign@language\foreign@language
\DeclareRobustCommand{\foreign@language}[1]{%
  \lowercase{\oldforeign@language{#1}}}
\theoremstyle{definition}
 \newtheorem{example}{\protect\examplename}
\theoremstyle{remark}
\newtheorem{rem}{\protect\remarkname}
\theoremstyle{plain}
\newtheorem{lem}{\protect\lemmaname}

\usepackage{hyperref}
\usepackage{cite}

\newcommand{\cprod}{\ensuremath{\prod}\,\,\,\,\!\!\!\!\kern-0.915em{{\scriptscriptstyle \times}}\kern+0.15em}
\newcommand{\cprods}{\ensuremath{\prod}\,\,\,\,\!\!\!\!\kern-0.957em{{\scriptscriptstyle \times}}\kern+0.1em}

\ifdefined\showcaptionsetup
 \PassOptionsToPackage{caption=false}{subfig}
\fi
\usepackage{subfig}
\makeatother

\providecommand{\examplename}{Example}
\providecommand{\lemmaname}{Lemma}
\providecommand{\remarkname}{Remark}

\begin{document}
\markboth{Preprint: IEEE Transactions on Signal Processing, VOL.72, 2024, PP.4888-4917, DOI: 10.1109/TSP.2024.3472068}{An Overview of Multi-Object Estimation via Labeled Random Finite Set}  
\title{An Overview of Multi-Object Estimation\\
 via Labeled Random Finite Set}
\author{Ba-Ngu~Vo,~Ba-Tuong~Vo,~Tran~Thien~Dat~Nguyen, and Changbeom~Shim\thanks{Acknowledgment: This work was supported by the Australian Research
Council under grants LP200301507 and FT210100506.} \thanks{The authors are with the School of Electrical Engineering, Computing
and Mathematical Sciences, Curtin University, Bentley, WA 6102, Australia
(email: \{ba-ngu.vo, ba-tuong.vo, t.nguyen1, changbeom.shim\}@curtin.edu.au).}}
\maketitle
\begin{abstract}

This article presents the Labeled Random Finite Set (LRFS) framework
for multi-object systems\textendash systems in which the number of
objects and their states are unknown and vary randomly with time.
In particular, we focus on state and trajectory estimation via a
multi-object State Space Model (SSM) that admits principled tractable
multi-object tracking filters/smoothers. Unlike the single-object
counterpart, a time sequence of states does not necessarily represent
the trajectory of a multi-object system. The LRFS formulation enables
a time sequence of multi-object states to represent the multi-object
trajectory that accommodates trajectory crossings and fragmentations.
We present the basics of LRFS, covering a suite of commonly used models
and mathematical apparatus (including the latest results not published
elsewhere). Building on this, we outline the fundamentals of multi-object
state space modeling and estimation using LRFS, which formally address
object identities/trajectories, ancestries for spawning objects, and
characterization of the uncertainty on the ensemble of objects (and
their trajectories). Numerical solutions to multi-object SSM problems
are inherently far more challenging than those in standard SSM. To
bridge the gap between theory and practice, we discuss state-of-the-art
implementations that address key computational bottlenecks in the
number of objects, measurements, sensors, and scans. 
\end{abstract}


\section{Introduction\protect\label{s:Introduction}}

\subsection{State Space Model}

State Space Model (SSM), also known as Hidden Markov Model (HMM),
is a fundamental concept in dynamical systems theory. The state of
the hidden object at sampling instant $k$ is characterized by the
state vector $x_{k}$, in some finite dimensional state space $\mathbb{X}$.
This state generates an observation vector $z_{k}$ in an (finite
dimensional) observation space $\mathcal{\mathbb{Z}}$, as depicted
in Fig. \ref{fig:single-object-sys}. \textit{State estimation}, \textit{system
identification}, and \textit{control} are three interrelated fundamental
problems in SSM.\vspace{-0cm}
\begin{figure}[H]
\begin{centering}
\resizebox{88mm}{!}{\includegraphics[clip]{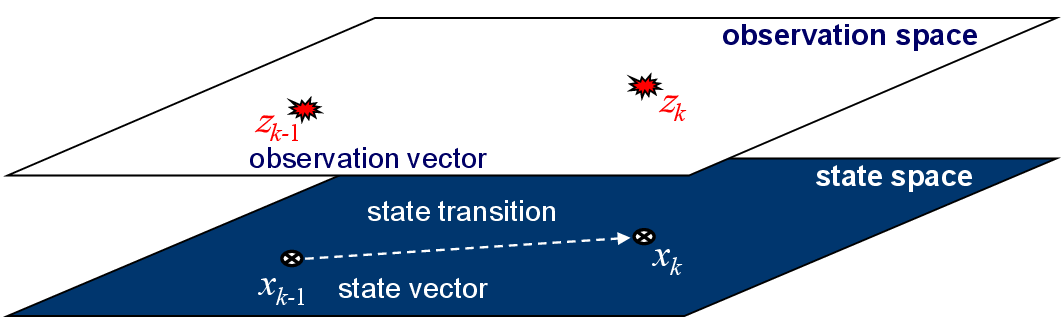}}
\par\end{centering}
\caption{State Space Model of a single-object system.}
\label{fig:single-object-sys}
\end{figure}

\textit{State estimation} entails \textit{estimating the state trajectory}
$x_{0:k}\triangleq(x_{0},...,x_{k})$, and is the most basic problem
upon which the other fundamental SSM problems are formulated. The
trajectory can be estimated via \textit{smoothing}, i.e., jointly
estimating the states in batches, or via \textit{filtering}, i.e.,
sequentially estimating the state at each time \cite{AndersonMoore79,Sarkkabook13},
see Fig. \ref{fig:track-single-obj}. In practice, jointly estimating
$x_{0:k}$ is intractable because the dimension of the variables (and
computational load) per time step increases with $k$ \cite{Sarkkabook13}.
Hence, it is imperative to smooth over short windows to ensure the
computational complexity per time step does not grow with time \cite{Sarkkabook13}.
Filtering\textendash the special case with a window length of one\textendash is
the most widely used \cite{AndersonMoore79,Sarkkabook13}. State estimation
is an active research area popularized by the Kalman filter \cite{AndersonMoore79,Jazwinski70}
and particle filter \cite{Gordon93,Doucet00}, which has far-reaching
impact in many fields of study \cite{AndersonMoore79,Jazwinski70,Robert_Bayesian_Choice,Sarkkabook13}.\vspace{-0cm}
\begin{figure}[h]
\begin{centering}
\resizebox{88mm}{!}{\includegraphics[clip]{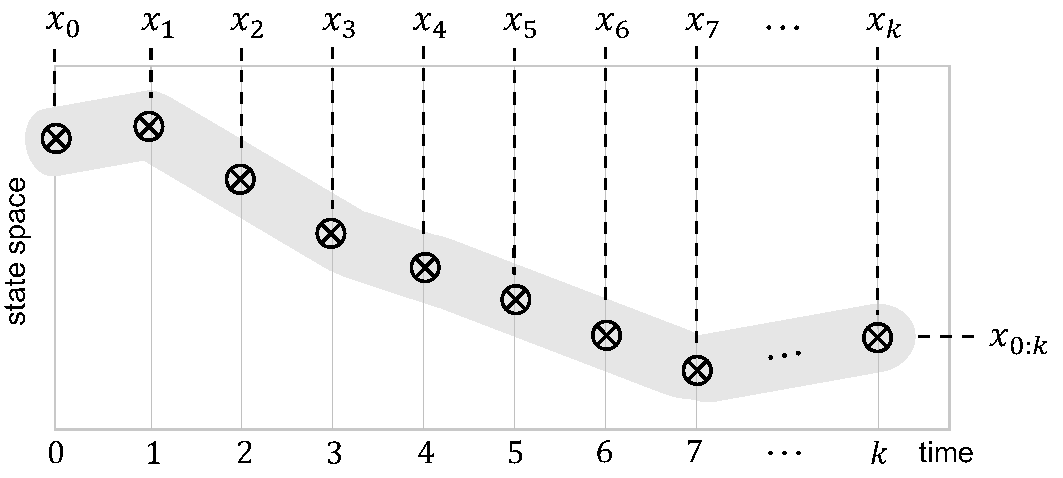}}
\par\end{centering}
\begin{centering}
\caption{\protect\label{fig:track-single-obj}States and trajectory of a 1-D
system. The state history $x_{0:k}$ defines the trajectory. Thus,
a trajectory estimate can be obtained from a history of individual
state estimates.}
\par\end{centering}
\vspace{-0cm}
\end{figure}

\subsection{Multi-Object System}

A \textit{multi-object system} is a generalized dynamical system arising
from a host of applications where, instead of a single object, the
number of objects and their states are unknown and vary randomly with
time, as depicted in Fig. \ref{fig:multi-object-sys}. The \textit{multi-object
state} is the \textit{set} \textit{of states of individual objects},
while the set of their trajectories is the \textit{multi-object trajectory}.
The multi-object state $X_{k}\subset\mathbb{X}$ at time $k$ generates
an observation set $Z_{k}\subset\mathbb{Z}$. Each existing object
may or may not generate an observation, while there could be false
observations not generated by any object. This is compounded by \textit{data
association uncertainty}, i.e., it is not known which observations
originated from which objects \cite{BlackmanBook99,BarShalomWillettTian11}.
\begin{figure}[h]
\begin{centering}
\resizebox{88mm}{!}{\includegraphics[clip]{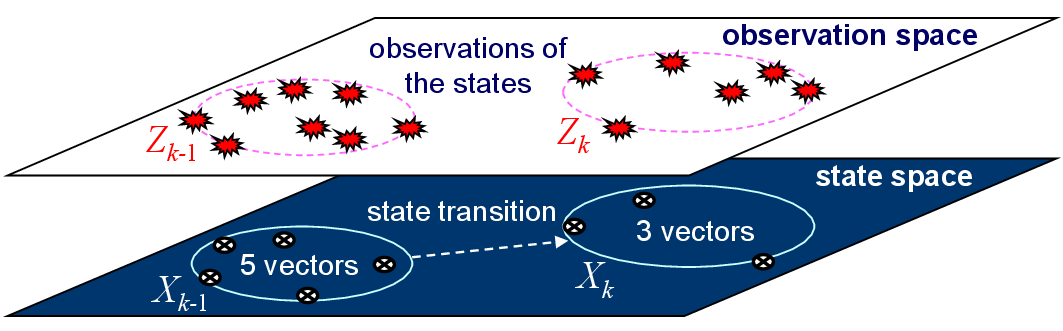}}
\par\end{centering}
\caption{\protect\label{fig:multi-object-sys}Multi-object system. The number
of states and observations vary with time. Existing objects may not
be detected, false observations may occur, and it is not known which
observations originated from which objects.}
\vspace{-0cm}
\end{figure}

In traditional SSM, the terminologies ``state estimation'' and ``trajectory
estimation'' are used interchangeably because they both entail estimation
of the trajectory (for the only object). In line with this terminology,
``\textit{multi-object state estimation}'' and ``\textit{multi-object
trajectory estimation}'' should both refer to estimation of the multi-object
trajectory, and hence are abbreviated as ``\textit{multi-object estimation}''
herein. To differentiate the task of estimating only the multi-object
states (without trajectory information), we use the term ``multi-object
localization''. It is important to note from Fig. \ref{fig:unlabed-multi-obj}
that unlike single-object systems, the multi-object state history
$X_{0:k}$ does not necessarily represent the multi-object trajectory.
Consequently, \textit{multi-object estimation} may not be possible
via filtering nor smoothing over moving windows.\vspace{-0cm}
\begin{figure}[H]
\begin{centering}
\vspace{-1mm}
\par\end{centering}
\begin{centering}
\subfloat[\label{fig:unlabeled-hist}Multi-object state history $X_{0:k}=$
($X_{0}$,...,$X_{k}$).]{\begin{centering}
\resizebox{88mm}{!}{\includegraphics[clip]{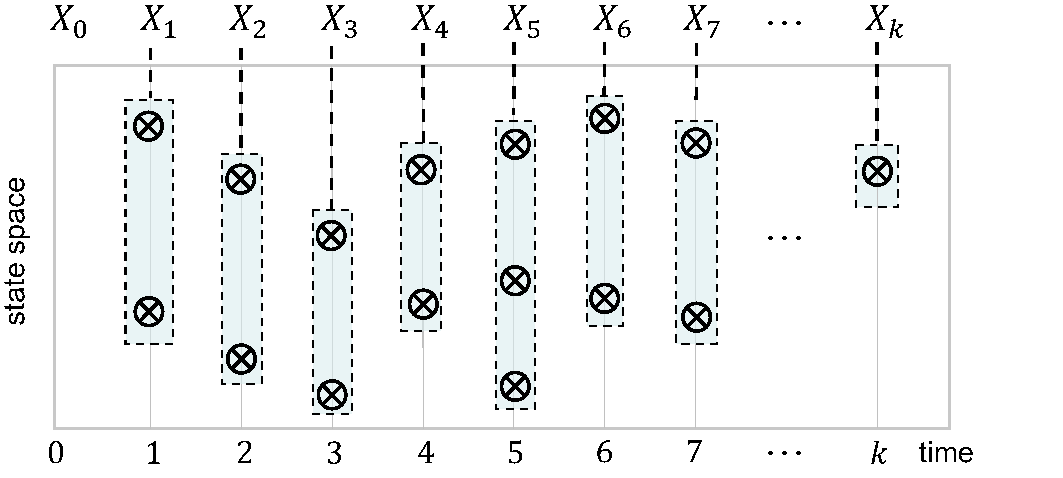}}
\par\end{centering}
\centering{}}
\par\end{centering}
\begin{centering}
\subfloat[\label{fig:unlabeled-traj}Multi-object trajectory.]{\begin{centering}
\resizebox{88mm}{!}{\includegraphics[bb=0bp 0bp 510bp 205bp,clip]{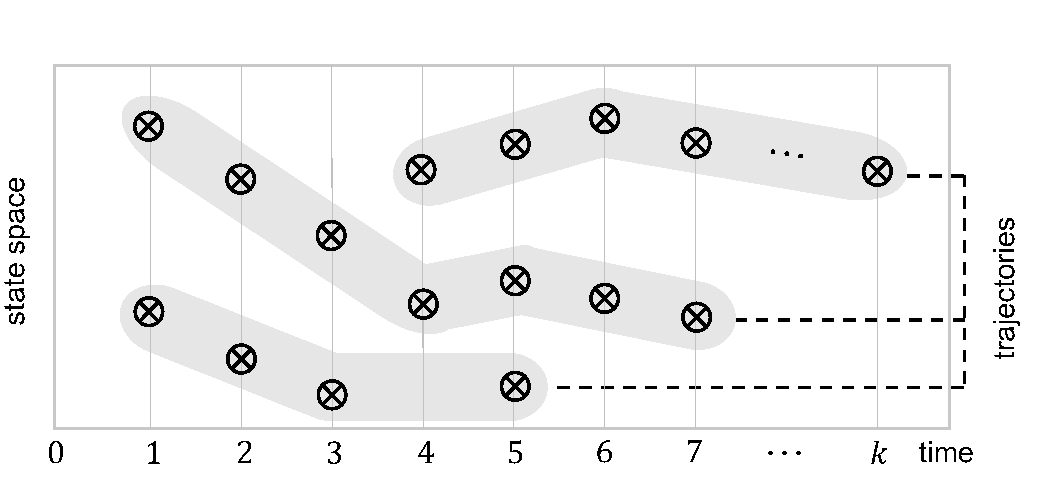}}
\par\end{centering}
}
\par\end{centering}
\caption{\protect\label{fig:unlabed-multi-obj}States and trajectory of a 1-D
multi-object system (note $X_{0}=\{\}$ in this example). Objects
may enter, exit, or reenter the state space. The multi-object state
history $X_{0:k}$ does not necessarily represent (nor contains sufficient
information to construct) the multi-object trajectory. }
\vspace{-0cm}
\end{figure}

For a versatile multi-object state representation emulating single-object
systems, the state history must be equivalent to the trajectory. Fundamentally,
this is accomplished by augmenting distinct labels or provisional
identities to each object state \cite[pp. 135, 196-197]{Goodmanetal97},
as illustrated in Fig. \ref{fig:tracklabel}. This labeled multi-object
representation enables multi-object (trajectory) estimation to be
done via filtering or smoothing over moving windows \cite{VoConj11,VoConj13}.
 Without the mechanism for linking trajectories from one window to
another, even if segments of the trajectories can be estimated over
short moving windows, multi-object (trajectory) estimation can only
be performed over a window that grows with time, thus, computationally
infeasible even for a single trajectory \cite{Sarkkabook13}. 
\begin{figure}[t]
\begin{centering}
\resizebox{88mm}{!}{\includegraphics[clip]{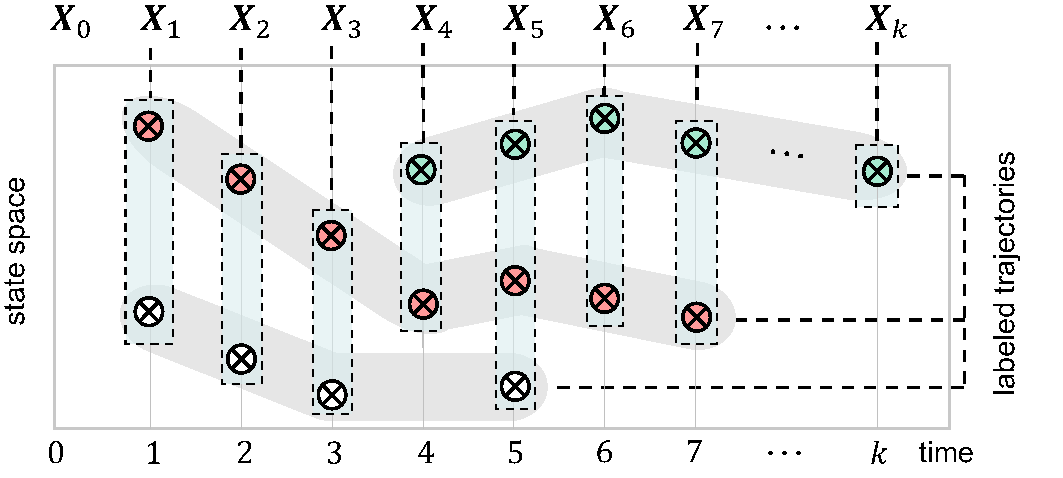}}
\par\end{centering}
\caption{\protect\label{fig:tracklabel}1-D labeled multi-object states and
trajectory. The three objects are augmented with red, green, and white
labels. Grouping the elements of the labeled multi-object state history
$\boldsymbol{X}_{\negthinspace0:k}=$ ($\boldsymbol{X}{}_{\negthinspace0}$,...,$\boldsymbol{X}_{\negthinspace k}$)
according to labels gives the multi-object trajectory. A labeled multi-object
trajectory estimate can be obtained from a history of labeled multi-object
state estimates.}
\vspace{0mm}
\end{figure}

Historically driven by the aerospace industry's interests in \textit{multi-object
tracking} (MOT), multi-object system problems are found in many diverse
application areas, including surveillance, situational awareness,
oceanography, autonomous vehicles/drones, field robotics, remote sensing,
computer vision, and cell biology  \cite{FS1985Radar,FS1986Radar,BarShalomWillettTian11,BlackmanBook99,Challaetal11,KochBook14,MahlerBook07,MahlerBook14,MalickBook12,Stoneetal13}.
MOT is a well-established field, with three main approaches: \textit{Multiple
Hypothesis Tracking} (\textit{MHT}), \textit{Joint Probabilistic Data
Association} (\textit{JPDA}), and \textit{Random Finite Set} (\textit{RFS});
we refer the reader to the texts \cite{BlackmanBook99}, \cite{BarShalom88},
\cite{MahlerBook14} for more details, or to \cite{Voetal-MMT-Wiley15}
for a brief overview. Due to false alarms, misdetections, and data
association uncertainty, multi-object system problems are far more
challenging than their single-object counterparts \cite{MahlerBook14}.

\vspace{-0.15cm}

\subsection{Aims and Scope}

It has been over two decades since the introduction of the RFS framework
for multi-object filtering \cite{Goodmanetal97}. Prior to this, related
approaches based on point process theory have been developed in \cite{BVL1964Methods,A1971Finding,BZI1980Detection,Washburn1987Random}
using random measure theoretic formulations. The RFS approach gained
significant traction due to its intuitive appeal as a geometric formulation,
and timely emergence in an era when computing technology as well as
numerical filtering techniques were sufficiently advanced. However,
the original RFS filtering formulation only considers multi-object
localization, and additional heuristics are needed to construct trajectory
estimates \cite[pp. 505-508]{MahlerBook07}. Indeed, early RFS-based
multi-object localization filters were so popular that their inability
to accommodate trajectories for MOT was mistakenly attributed to the
framework itself. On the contrary, the idea of augmenting distinct
labels to individual states to represent trajectories (\textcolor{red}{}illustrated
in Fig. \ref{fig:tracklabel}) has been discussed in \cite[pp. 135, 196-197]{Goodmanetal97}.
Due to its simplicity, Occam's Razor suggests that this labeled multi-object
representation is the proper approach to multi-object systems, but
it was deemed computationally impractical at the time and abandoned
in favor of the unlabeled representation. Nevertheless, it inspired
the development of Labeled RFS \cite{VoConj11,VoConj13}, culminating
in a suite of analytical and numerical tools for multi-object SSMs
as well as the growth in applications and commercial interests.

This article provides an overview of the Labeled RFS (LRFS) approach
to multi-object estimation, covering topics from the elements of RFS
to the Generalized Labeled Multi-Bernoulli (GLMB) filter \cite{VoVoP14}
capable of handling over a million trajectories \cite{Beard18-largescale}.
The LRFS formulation not only ensures that the multi-object trajectory
is indeed given by the history of the multi-object states, accommodating
trajectory crossings, fragmentations, and ancestries (for spawning
objects), but also enables characterization of uncertainty on the
multi-object trajectory ensemble and alleviation of critical computational
bottlenecks. Existing overviews such as the reviews \cite{MahlerFISST101,MahlerFISST102},
surveys \cite{MahlerPHDsurv-07,MahlerBriefsurv-15}, and texts \cite{MahlerBook07,Ristic-book-13,MahlerBook14},
mainly focus on RFS multi-object localization (which does not address
MOT), though some early developments in LRFS filters have been discussed
in the text \cite{MahlerBook14}, survey \cite{MahlerBriefsurv-15},
and particle filtering overview \cite{Risticetaloverview-16}. Keeping
in mind the balance between scientific rigor and utility, we present
an up-to-date coverage of LRFS analytical and computational tools,
as well as the ensuing multi-object estimation solutions that address
key computational bottlenecks in the number of objects, measurements,
sensors, and scans. This coverage also includes results on closed-form
information divergences for a versatile class of LRFSs, and estimators
based on the notion of joint existence probability unique to LRFS,
not previously published. 

The rest of the article is organized as follows. Section \ref{s:Background}
presents some background on Bayesian state estimation and the labeled
set representation of the multi-object state. To model set-valued
random variables, Section \ref{sec:RFS-fundamental} presents the
fundamentals of RFS theory and some classical RFS models. Section
\ref{sec:LRFS-models} introduces LRFS and the mathematical apparatus
for multi-object state/trajectory modeling. Building on this, Section
\ref{sec:Multi-object-Estimation} extends Bayesian state estimation
to multi-object SSM along with solutions and discussions of related
SSM problems. Closing remarks are given in Section \ref{s:Conclusion}. 

\section{Background\protect\label{s:Background}}

\begin{table}[!t]
\global\long\def\arraystretch{1.3}%
 \caption{Common notations from Section \ref{s:Background} onwards.}

\vspace*{0.4cm}
 {\footnotesize\label{tbl:notation-background}}{\footnotesize\par}
\centering{}%
\begin{tabular}{|c|l|}
\hline 
{\small\textbf{Notation}} & {\small\textbf{Description}}\tabularnewline
\hline 
{\small$x_{m:n}$} & {\small$x_{m},x_{m\text{+}1},\ldots,x_{n}$}\tabularnewline
{\small$\mathbb{X}$} & {\small Finite dimensional state space}\tabularnewline
{\small$x_{k}$} & {\small State vector at time $k$}\tabularnewline
{\small$\mathbb{Z}$} & {\small Finite dimensional observation space}\tabularnewline
{\small$z_{k}$} & {\small Observation vector at time $k$}\tabularnewline
{\small$f_{k}(\cdot|x_{k\texttt{-}1})$} & {\small Markov transition density to time $k$ given $x_{k\texttt{-}1}$ }\tabularnewline
{\small$g_{k}(z_{k}|x_{k})$} & {\small Likelihood of observing $z_{k}$ given $x_{k}$}\tabularnewline
{\small$p_{0:k}(x_{0:k})$} & {\small Posterior density at $x_{0:k}$}\tabularnewline
{\small$p_{k}(x_{k})$} & {\small Filtering density at $x_{k}$}\tabularnewline
{\small$\mathcal{D}(f)$} & {\small Domain of function $f$}\tabularnewline
{\small$\ell$} & {\small Label of an object }\tabularnewline
{\small$\mathbb{L}$} & {\small (Discrete) space of labels}\tabularnewline
{\small$\boldsymbol{x}$, $\boldsymbol{x}_{k}$} & {\small Labeled state $(x_{k},\ell)$ of an object at time $k$}\tabularnewline
{\small$\boldsymbol{X}$, $\boldsymbol{X}_{k}$} & {\small Labeled multi-object state at time $k$}\tabularnewline
\hline 
\end{tabular}
\end{table}

We begin with an outline of Bayesian state estimation for discrete-time
SSMs in Subsection \ref{subsec:Bayes-Filter}. Subsection \ref{subsec:Multi-object-rep}
formalizes the set representation for the multi-object state/trajectory
that emulates the single-object state/trajectory, while Subsection
\ref{subsec:Tricky} discusses the challenges of working with sets. 

\subsection{Bayesian State Estimation \protect\label{subsec:Bayes-Filter}}

In an SSM, the system state vector $x_{k}\in\mathbb{X}$ at time
$k$ evolves from its previous value $x_{k\texttt{-}1}$, and generates
an observation $z_{k}\in\mathbb{Z}$ according to the \textit{state
}and \textit{observation equations}
\begin{flalign}
x_{k} & =F_{k}\left(x_{k\texttt{-}1},u_{k\texttt{-}1},\nu_{k\texttt{-}1}\right),\label{eq:single-obj-state-eq}\\
z_{k} & =G_{k}\left(x_{k},u_{k},\mu_{k}\right),\label{eq:single-obj-meas-eq}
\end{flalign}
where $F_{k}$ and $G_{k}$ are non-linear mappings, $u_{k}$ (and
$u_{k\texttt{-}1}$) is the control or input signal, $\nu_{k\texttt{-}1}$
and $\mu_{k}$ are, respectively, process and measurement noise.

\textit{State estimation}, i.e., estimation of the state trajectory,
together with \textit{system identification} and \textit{control}
are the three fundamental problems in dynamical system. In system
identification, the goal is to estimate the system parameters from
the system input and output, while in control, the goal is to use
the input signals to drive the system state/trajectory to prescribed
regions of the state space. In this article, we focus on state estimation,
and hence the control signal is omitted.

The Bayesian estimation paradigm models the state and observation
as random vectors. The state equation is characterized by the Markov
\textit{transition density} $f_{k}(x_{k}|x_{k\texttt{-}1})$, i.e.,
the probability density of a transition to state $x_{k}$ at time
$k$\ given the previous state $x_{k\texttt{-}1}$. The observation
equation is characterized by the measurement \textit{likelihood function}
$g_{k}(z_{k}|x_{k})$, i.e., the probability density of observing
$z_{k}$ given state $x_{k}$. Further, it is assumed that measurements
are conditionally independent, i.e., the probability density of the
observation history $z_{1:k}$ condition on $x_{1:k}$ is given by
\[
p_{1:k}(z_{1:k}|x_{1:k})=g_{k}(z_{k}|x_{k})g_{k\texttt{-}1}(z_{k\texttt{-}1}|x_{k\texttt{-}1})...g_{1}(z_{1}|x_{1}).
\]

All information about the state history to time $k$ is encapsulated
in the \emph{posterior} density $p_{0:k}(x_{0:k})\triangleq p_{0:k}(x_{0:k}|z_{1:k})$,
which can be computed recursively for any $k\geq1$, starting from
an initial prior $p_{0}$, via the Bayes \textit{posterior recursion}:
\begin{equation}
p_{0:k}(x_{0:k})=\frac{g_{k}(z_{k}|x_{k})f_{k}(x_{k}|x_{k\texttt{-}1})p_{0:k\texttt{-}1}(x_{0:k\texttt{-}1})}{\int g_{k}(z_{k}|\varsigma_{k})f_{k}(\varsigma_{k}|\varsigma_{k\texttt{-}1})p_{0:k\texttt{-}1}(\varsigma_{0:k\texttt{-}1})d\varsigma_{0:k}}.\label{eq:posterior}
\end{equation}
Note that the dependence on $z_{1:k}$ is omitted for notational compactness.
The above recursion, also known as smoothing-while-filtering \cite{Briers-10}
(or simply smoothing in this paper), is not suitable for real-time
operations. Since the dimension of the trajectory probability density
grows with time, computational/memory requirement increases at each
time step and quickly becomes intractable. Real-time applications
require algorithms with computational complexity per time step that
does not grow with time \cite[pp. 53-54]{Sarkkabook13}.
\begin{figure}[h]
\begin{centering}
\resizebox{88mm}{!}{\includegraphics[bb=0bp 0bp 510bp 110bp,clip]{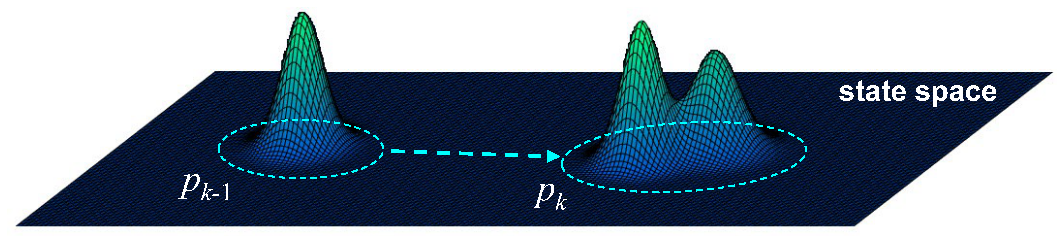}}
\par\end{centering}
\caption{Propagation of the filtering density $p_{k}$.}
\label{fig:single-obect-filter}
\end{figure}

The \emph{filtering} density $p_{k}(x_{k})\negthinspace\triangleq\negthinspace\int\!p_{0:k}(x_{0:k})dx_{0:k\texttt{-}1}$,
i.e., the marginal of the posterior density at the current time, encapsulates
all information about the current state, and can be computed recursively
using the Bayes \textit{filtering recursion}: 
\begin{equation}
p_{k}(x_{k})=\frac{g_{k}(z_{k}|x_{k})\int\!f_{k}(x_{k}|y)p_{k\texttt{-}1}(y)dy}{\int g_{k}(z_{k}|x)\int\!f_{k}(x|y)p_{k\texttt{-}1}(y)dydx}.\label{eq:bayes2}
\end{equation}
The above recursion propagates a function on a fixed dimensional space,
see Fig. \ref{fig:single-obect-filter}, and hence has a fixed computational
complexity per time step \cite{Sarkkabook13}. The \emph{smoothing}
density $p_{k\texttt{-}l|k}(\cdot|z_{1:k})$, i.e., the marginal of
the posterior at time $k-l$, can also be computed recursively (e.g.,
forward-backward or two-filter smoothing), but is not widely used,
and will not be covered here. Interested readers are referred to \cite{Sarkkabook13,Briers-10}.

The state can be estimated from appropriate probability densities
via the \emph{expected a posteriori} (conditional mean) or \emph{maximum
a posteriori} (conditional mode) estimators. These estimators are
Bayes optimal with respect to certain Bayes risks (or penalties),
and statistically \emph{consistent}, i.e., converge (almost surely)
to the true state as more data accumulate \cite{Robert_Bayesian_Choice,KayEstBook}.
The trajectory can be estimated from the posterior jointly as a sequence
of states, or as a sequence of individual state estimates from filtering
(or smoothing). Smoothing refines state estimates from data that arrives
later, and is expected to yield better estimates than filtering \cite{Briers-10,Meditch73,DoucetTutorial09}.

Numerical methods for Bayesian state estimation is an active area
of research \cite{AndersonMoore79,Jazwinski70,Ristic04,Sarkkabook13}.
The posterior recursion admits the Kalman smoother as an analytic
solution for linear Gaussian models \cite{AndersonMoore79}, while
for general non-linear models, Sequential Monte Carlo (SMC) approximations
have been proposed \cite{Kitagawa96}. To maintain a fixed computational
complexity per time step, smoothing is performed over short moving
windows. The most widely used is filtering, the special case with
a window length of one. Under linear Gaussian models, the filtering
recursion admits the Kalman filter as an analytic solution \cite{AndersonMoore79}.
In general, analytic solutions are not possible. Popular approximate
solutions include the extended Kalman filter \cite{AndersonMoore79,Jazwinski70},
unscented Kalman filter \cite{JulierUhlmann04}, Gaussian sum filter
\cite{Sorenson71}, cubature Kalman filter \cite{Arasaratnam09,JiaHDCKF13},
SMC/particle filter \cite{Gordon93,Doucet00,ArulampalamTutorial02,DoucetTutorial09,Ristic04}
and particle flow filter \cite{DH_Normal_SPIE14}.

\subsection{Multi-Object Representation\protect\label{subsec:Multi-object-rep}}

Regardless of the differences in MOT approaches, mathematically,
a \textit{trajectory} in a state space $\mathcal{\mathbb{X}}$ and
discrete-time window $\mathbb{T}$ is a mapping $\tau:\mathbb{T}\rightarrow\mathbb{X}$
\cite{Beard18-largescale}. The \textit{domain}, $\mathcal{D}(\tau)\subseteq\mathbb{T}$,
corresponds to the set of instants when the object exists. This definition
covers the so-called \textit{fragmented trajectories/tracks}, i.e.,
trajectories with non-contiguous domains, to accommodate objects
physically exiting and reentering the state space, as well as trajectory
estimates produced by most MOT algorithms, as depicted in Fig. \ref{fig:unlabed-multi-obj}(b).

For versatility and applicability to a wide range of problems, we
require a multi-object representation that allows: 
\begin{itemize}
\item (R.1) the multi-object trajectory to be determined from the multi-object
state history (analogous to single-object systems), ensuring multi-object
estimation to be accomplished with a complexity per time step that
does not grow with time, for computational tractability;
\item (R.2) fragmented trajectories (see Fig. \ref{fig:unlabed-multi-obj}(b))
and multiple objects occupying the same attribute state, to model
scenarios unique to challenging multi-object estimation problems such
as reappearance/reidentification and merging/occlusion \cite{ZWWZL2021Fairmot,MP2014Bayesian,ARS2011Analytical}. 
\end{itemize}
Neither of these requirements could be met with the unlabeled representation,
as illustrated in Fig. \ref{fig:unlabed-multi-obj}, and the fact
that a set cannot contain repeated elements. The labeled representation
\cite[pp. 135, 196-197]{Goodmanetal97},  \cite{VoConj11,VoConj13}
described in the following is most suitable, meeting both requirements. 

The \textit{labeled state} of an object at a particular instant (if
it exists) is represented by $\boldsymbol{x}=(x,\ell)$, where $x\in\mathbb{X}$
is its \textit{attribute} state (e.g., kinematics, visual features,
etc.), and $\ell$ is its \textit{label}, a provisional distinct identity
in some discrete space $\mathbb{L}$ called a \textit{label space}.
A common practice is to use $\ell=(s,\iota)$, where $s$ is the time
of birth and $\iota$ is an index to distinguish objects born at the
same time \cite{VoConj11,VoConj13}.

A \textit{labeled multi-object state} $\boldsymbol{X}$ is a finite
subset of the product space $\mathbb{X}\times\mathbb{L}$ with \textit{distinct
labels}, i.e., no two elements of $\boldsymbol{X}$ share the same
label. Fundamentally, macroscopic objects (whose extents are significantly
larger than the de Broglie wavelengths) can be labeled with distinct
identifiers \cite{MRZ2022Track}. In some applications, distinct identifiers
are even explicitly included in the intrinsic states of objects, e.g.,
in \cite{Nguyenetal-PP-19} the state of each object contains a unique
RFID signal (from a finite set of RFID signals).  In practice, track
labeling is a required functionality of a MOT system, see e.g., \cite{BlackmanBook99,MRZ2022Track}.
\begin{figure}[t]
\begin{centering}
\vspace{-0.3cm}
\par\end{centering}
\begin{centering}
\hspace{-0.1cm}\subfloat[With labels.]{\begin{centering}
\includegraphics[clip,height=7cm]{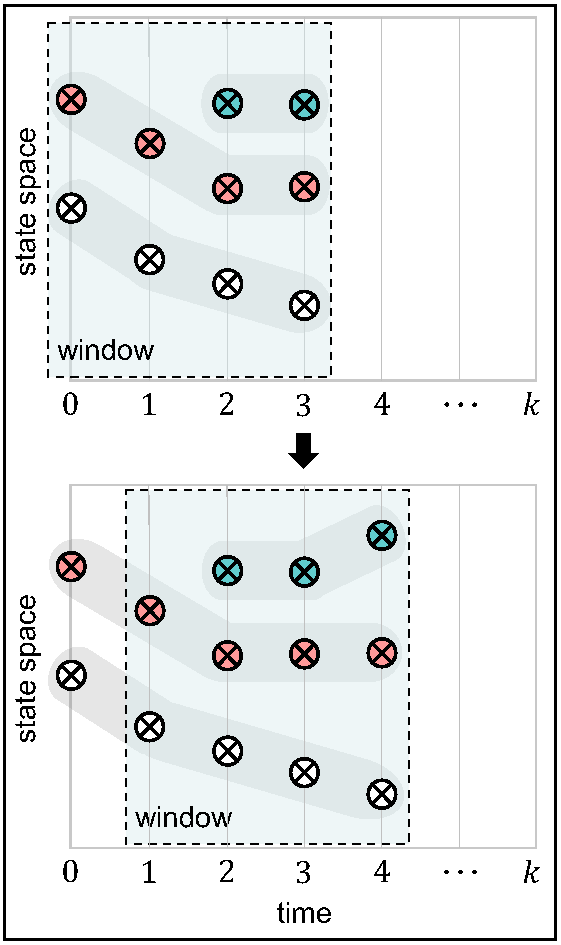}
\par\end{centering}
\centering{}}\hspace{-0.2cm}\subfloat[Without labels.\textcolor{purple}{{} }]{\begin{centering}
\includegraphics[clip,height=7cm]{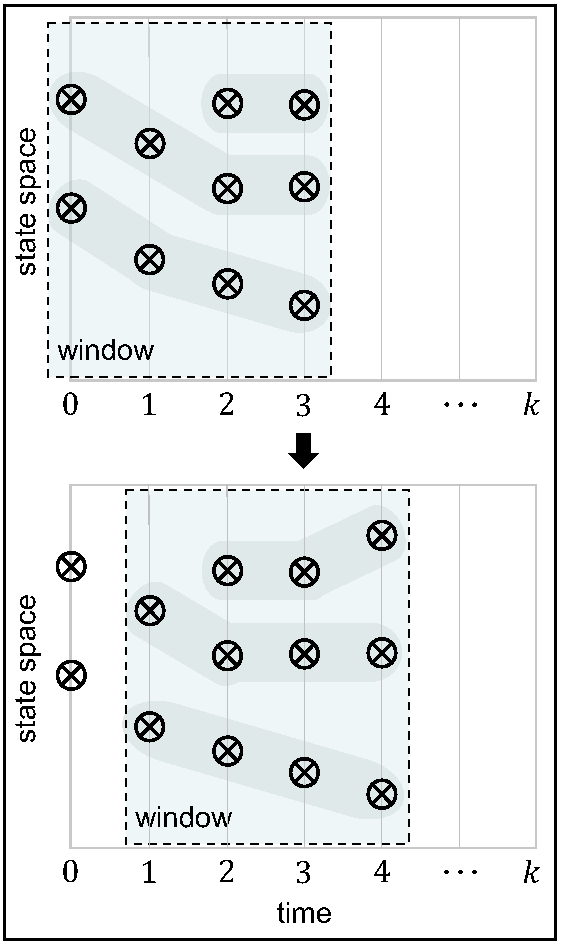}
\par\end{centering}
}
\par\end{centering}
\begin{centering}
\hspace{-0.5cm}
\par\end{centering}
\caption{\protect\label{fig:moving-windows}Multi-object trajectories in moving
windows. In (a) labels enable linking the trajectories between the
windows $\{0\textrm{:}3\}$ and $\{1\textrm{:}4\}$ to represent trajectories
on $\{0\textrm{:}4\}$. Without labels (b), there is no mechanism
to link these trajectories.  }

\centering{}\vspace{-0cm}
\end{figure}

Labeling enables a sequence $\boldsymbol{X}_{j:k}$ of multi-object
states to completely characterize the \textit{multi-object trajectory}\textendash the
set of trajectories of the objects\textendash on the interval $\{j{\textstyle \vcentcolon}k\}$
analogous to single-object systems, via the grouping of states according
to labels (illustrated in Fig. \ref{fig:tracklabel}). Formally,
a trajectory in $\boldsymbol{X}_{j:k}$ is a time-stamped sequence
$[(x_{s},\ell),...,(x_{t},\ell)]$ of labeled states in $\boldsymbol{X}_{j:k}$
 with a common label, which indeed defines a mapping $\tau_{\ell}\textrm{ : }i\mapsto x_{i}$,
$i\in\mathcal{D}(\tau_{\ell})$. The labeled multi-object representation
allows: multi-object estimation (filtering or smoothing on moving
windows, see Figs. \ref{fig:tracklabel} and \ref{fig:moving-windows}(a))
with computational complexity per time step that does not grow with
time; multiple objects simultaneously occupying the same attribute
state, e.g., $\{(x,\ell_{1})$, $(x,\ell_{2})\}$, which inevitably
occurs in a finite state space (as often is the case in the HMM literature),
representing mergings/collisions, and occlusions in a wide range of
applications, e.g., \cite{SLHN2009Tracking,MP2014Bayesian,ARS2011Analytical,HRMB2014High};
and fragmented trajectories, arising from objects exiting and reentering
the state space (especially when it is bounded), intrinsic to the
reappearance/reidentification problem in computer vision \cite{ZWWZL2021Fairmot,LZZLZH2022Rethinking,CAZS2018Real}.

\textit{Pragmatically, labeling is unavoidable in real multi-object
systems}. Trajectory estimation over a growing window is numerically
intractable as the computational complexity per time step grows with
time \cite{Sarkkabook13}. Thus, it is imperative to use shorter
moving windows, and link the trajectories between the windows together.
 This is straightforward with a labeled multi-object representation
as illustrated in Fig. \ref{fig:moving-windows}(a). In an unlabeled
representation, there is no means to link the trajectories as illustrated
in Fig. \ref{fig:moving-windows}(b), making multi-object trajectory
estimation infeasible without resorting to heuristic post-processing.
Note that heuristically matching trajectories between windows via
optimal assignment still requires labeling them.  

State representation goes hand in hand with (mathematical) metrics
for multi-object estimation error \cite[Sec. 2.4]{VVPS2010Joint}.
A popular metric in the literature is the Optimal Sub-Pattern Assignment
(OSPA) metric for multi-object states \cite{Schumacher2008}, which
has been extended to multi-object trajectories in \cite{Beard18-largescale}.
This extension, called the OSPA\textsuperscript{(2)} metric, penalizes
errors in localization, number of trajectories, fragmentation and
identity switching. The OSPA and OSPA\textsuperscript{(2)} metrics,
respectively, gauge localization and tracking errors. The reader
is referred to \cite{NRVVSR2022How} for a comprehensive study of
multi-object estimation performance evaluations, which includes such
metrics.  

\subsection{Why Working with Sets is Tricky?\protect\label{subsec:Tricky}}

The Bayesian framework models the state (and observation) as a random
variable, and consequently, in a multi-object SSM, a finite-set-valued
random variable or a random finite set (RFS) is needed to model the
multi-object state (and observation). Probability densities of RFSs
are needed for the Markov transition density, likelihood function
and posterior/filtering density characterizing the multi-object SSM. 

\textit{The probability density of an RFS cannot be treated like that
of a random vector} as illustrated by the following example. \vspace{-0.1cm}
\begin{example}
\label{ex2:apple_tree} (Borrowed from \cite{VDPTV2018Modelbased})
Suppose that the number of apples falling in a day is uniformly distributed
between 0 and 9, and that conditional on the number of fallen apples,
the landing positions are independently and identically distributed
(i.i.d.) according to the probability density $p_{f}$, shown in Fig.
\ref{fig:landingPDF}. Treating the set of landing positions as a
random vector (see e.g., \cite{MS2003Novelty}), and noting that the
probability of $m$ apples falling ($m<10$) is $\frac{1}{10}$, we
have $p(\{x_{1:m}\})=\frac{1}{10}\prod_{i=1}^{m}p_{f}\negthinspace\left(x_{i}\right)$.
Consider the landing sets $\{x_{1}\}$ and $\{x_{2},x_{3}\}$, where
$x_{1},x_{2},$ and $x_{3}$ are shown in Fig. \ref{fig:landingPDF}.
Which of the sets is more likely? Since $p\negthinspace\left(\{x_{1}\}\right)=2\times10^{\texttt{-}2}$
and $p\negthinspace\left(\{x_{2},x_{3}\}\right)=3.6\times10^{\texttt{-}2}$,
it would appear that $\{x_{2},x_{3}\}$ is more likely. However, had
we measured distance in centimeters, then the probability density
$p_{f}$ is scaled by $\mathrm{10}^{\texttt{-}2}$ (because it is
non-zero on the interval $(-100,100)$ and must integrate to 1), which
results in $p\negthinspace\left(\{x_{1}\}\right)=2\times10^{\texttt{-}4}$
\textgreater{} $p(\{x_{2},x_{3}\})=3.6\times10^{\texttt{-}6}$, contradicting
the previous conclusion!
\begin{figure}[t]
\begin{centering}
\vspace{-0.2cm}
\par\end{centering}
\begin{centering}
\subfloat[Landing set $\{x_{1}\}$.]{\begin{centering}
\includegraphics[clip,width=1\columnwidth]{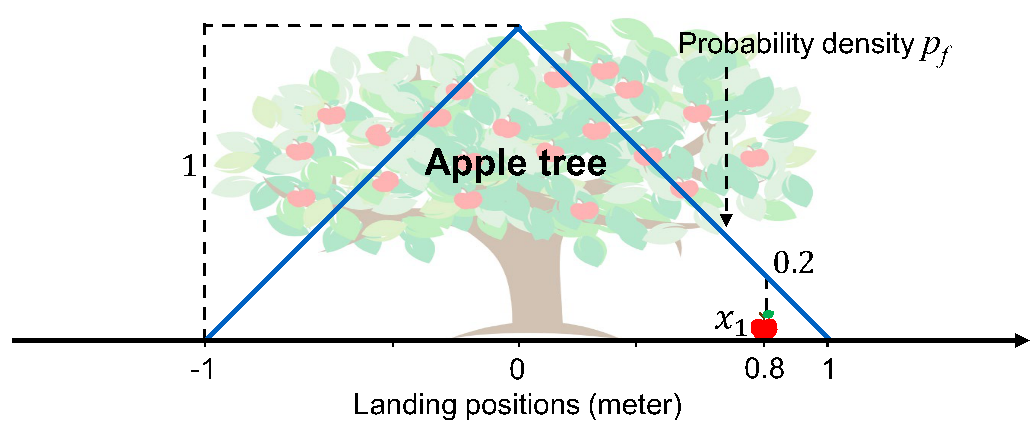}
\par\end{centering}
\centering{}}\hspace{0.3cm}\subfloat[Landing set $\{x_{2},x_{3}\}$.]{\begin{centering}
\includegraphics[clip,width=1\columnwidth]{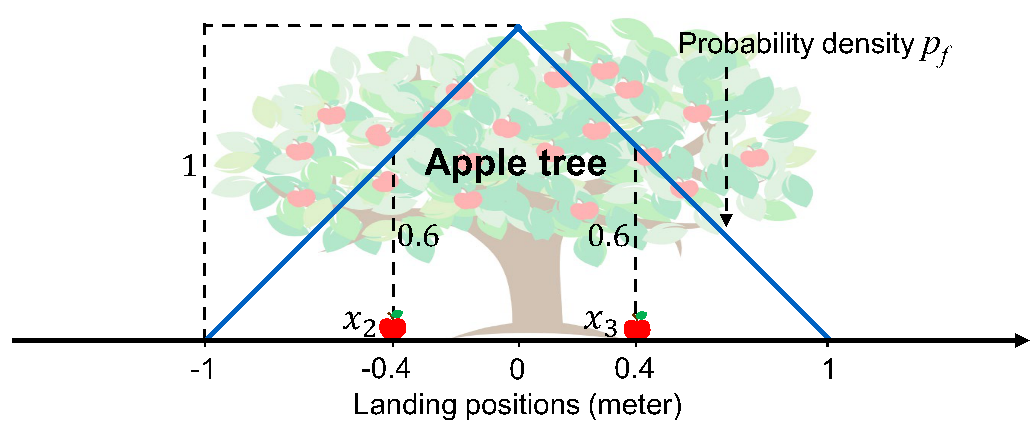}
\par\end{centering}
}
\par\end{centering}
\begin{centering}
\hspace{-0.5cm}
\par\end{centering}
\caption{\protect\label{fig:landingPDF}Distribution of landing positions.
Position $x_{1}$ is 3 times less likely than $x_{2}$, and also 3
times less likely than $x_{3}$. }

\centering{}\vspace{-0cm}
\end{figure}
\end{example}
The notion of (probability) density goes hand in hand with integration.
The density $p_{X}$ of a random vector $X$, is defined such that
$\Pr(X\in S)=\int_{S}p_{X}(x)dx$, for any (measurable) $S$. This
means $p_{X}(x)$ is the instantaneous probability per unit hyper-volume
(or formally, the Radon-Nikod\'{y}m derivative of the probability distribution
with respect to hyper-volume) at $x$, and is interpreted as the
likelihood of instantiation $x$. Instantiations with higher density
values are more likely (or probable) than those with lower ones, and
their relative likelihoods are invariant to the unit of hyper-volume.
The multi-object state space is the class of all finite subsets of
$\mathcal{X}$, and does not inherit the usual Euclidean notion of
density/integration. As noted in \cite{MahlerFISST101}, the inconsistency
in the above example arises because $p(\{x_{1}\})$ is measured in
``$\mathrm{m}^{\texttt{-}1}$'' or ``$\mathrm{cm}^{\texttt{-}1}$''
whereas $p(\{x_{2},x_{3}\})$ is measured in ``$\mathrm{m}^{\texttt{-}2}$''
or ``$\mathrm{cm}^{\texttt{-}2}$'', and hence cannot be meaningfully
compared.\vspace{0 cm}

Another ``tricky'' issue is the \textit{truncation of set-functions}
(functions of sets), made up of intractably large sums in most multi-object
system problems. Any function of a set $\{x_{1:n}\}$ must be symmetric
in $x_{1},...,x_{n}$, but truncation does not necessarily preserve
symmetry, rendering the truncated sum an invalid set-function. Consider,
e.g., the set-function $f(\{x,y\})=x+y+xy$, which is clearly symmetric.
However, truncating $y$ gives the approximation $\hat{f}(x,y)=x+xy$,
which is non-symmetric because $\hat{f}(x,y)\neq\hat{f}(y,x)=y+xy$,
and thus not a valid function of the set $\{x,y\}$. In general, care
must be taken to ensure validity of the results when truncating set-functions,
see also Subsections \ref{subsec:GLMB}, \ref{subsec:LRFS-approx},
and \ref{ss:Standard-Multi-object-SSM}. 

\section{Random Finite Sets Fundamentals\protect\label{sec:RFS-fundamental}}

\begin{table}[!t]
\global\long\def\arraystretch{1.3}%
 \caption{Common notations from Section \ref{sec:RFS-fundamental} onwards.}

\vspace*{0.4cm}
{\footnotesize\label{tbl:notation-fundamentals}}{\footnotesize\par}
\centering{}%
\begin{tabular}{|c|l|}
\hline 
{\small\textbf{Notation}} & {\small\textbf{Description}}\tabularnewline
\hline 
{\small$\mathcal{X}$} & {\small Finite dimensional Euclidean space}\tabularnewline
{\small$\mathcal{F}(\mathcal{X})$} & {\small Class of all finite subsets of $\mathcal{X}$}\tabularnewline
{\small$\mathbf{1}_{\mathcal{T}}(\cdot)$} & {\small Indicator function for a set $\mathcal{T}$}\tabularnewline
{\small$\pi_{\Sigma}$} & {\small Belief (or FISST) density of (the RFS) $\Sigma$}\tabularnewline
{\small$|X|$} & {\small Cardinality (number of elements) of $X$}\tabularnewline
{\small$h^{X}$} & {\small Multi-object exponential $\prod_{x\in X}h\left(x\right)$,
with $h^{\emptyset}=1$}\tabularnewline
{\small$\rho_{\Sigma}$} & {\small Cardinality distribution of $\Sigma$}\tabularnewline
{\small$v_{\Sigma}$} & {\small Probability hypothesis density (PHD) of $\Sigma$}\tabularnewline
{\small$\langle f,g\rangle$} & {\small Inner product $\int f(x)g(x)dx$ of functions $f$ and $g$}\tabularnewline
{\small$\delta_{Y}[X]$} & {\small Kronecker-}{\small$\delta$, $\delta_{Y}[X]=1$ if $X=Y$,
and 0 otherwise}\tabularnewline
\hline 
\end{tabular}
\end{table}

A \textit{random finite set} (\textit{RFS}) of $\mathcal{X}$ is a
\textit{random variable taking values in the class of all finite subsets
of} $\mathcal{X}$, hereon denoted as $\mathcal{F}(\mathcal{X})$.
The cardinality and elements of an RFS are random variables, and the
elements are unordered. A widely known example is the Poisson RFS,
where the number of points is Poisson distributed and the points are
i.i.d. according to some probability law on $\mathcal{X}$ \cite{MW2007Modern}.
In this article, we restrict ourselves to a finite dimensional Euclidean
space $\mathcal{X}$ so that RFSs of $\mathcal{X}$ have well-defined
probability densities. 

In Subsection \ref{subsec:Fundamental-Descriptors}, we summarize
some of the fundamental descriptors for RFSs. RFS statistics relevant
to multi-object systems such as the cardinality and Probability Hypothesis
Density (PHD) are presented in Subsection \ref{subsec:Card-PHD},
while multi-object estimators are discussed in Subsection \ref{subsec:Multi-Object-Est}.
Subsection \ref{subsec:Classical-Models} summarizes a number of popular
RFS models.

\subsection{Fundamental Descriptors\protect\label{subsec:Fundamental-Descriptors}}

RFSs fall under the broader class of random (closed) sets in stochastic
geometry \cite{Chiu-Stoyan-13}, which also enables Bayesian inference
with uncertainty models such as fuzzy set, Dempster-Shafer theory,
and rule-based expert system \cite{Goodmanetal97}. Stochastic geometry
has some overlaps with point process theory, and an RFS can be treated
as a \textit{simple finite point process}, i.e., a point process whose
realizations are finite and have no repeated points \cite{Daley88,Chiu-Stoyan-13}.
While the probability density of a point process may not exist, RFSs
of $\mathbb{R}^{n}$ do admit probability densities \cite[Prop. 5.4.V, p. 138]{Daley88},
\cite[Prop. 19, p. 159]{Goodmanetal97}, \cite{MW2007Modern}.

\subsubsection{Probability Density}

Since the notion of hyper-volume needed for density/integration on
$\mathcal{X}$ does not extend to $\mathcal{F}(\mathcal{X})$, we
adopt an equivalent construct\textendash the uniform (probability)
distribution. In point process theory, the unit Poisson RFS exhibits
complete spatial randomness analogous to the uniform distribution
on $\mathcal{X}$ (see e.g., \cite{Daley88,MW2007Modern,Chiu-Stoyan-13}).
Specifically, the number of points is Poisson distributed with unit
rate, and the points are uniformly i.i.d. in $\mathcal{X}$. The \textit{Poisson
measure}\textendash the unnormalized probability distribution of a
unit Poisson RFS\textendash is defined for each (measurable) $\mathcal{T}\subseteq\mathcal{F}(\mathcal{X})$
by 
\begin{align*}
\mu(\mathcal{T})= & \sum_{i=0}^{\infty}\frac{1}{i!U^{i}}\int\mathbf{1}_{\mathcal{T\negthinspace}}\left(\left\{ x_{1:i}\right\} \right)dx_{1:i},
\end{align*}
where $\mathbf{1}_{\mathcal{T}}(\cdot)$ is the indicator function
for $\mathcal{T}$, and $U$ is the unit of hyper-volume in $\mathcal{X}$,
with the convention that the integral for $i=0$ is the integrand
evaluated at $\emptyset$, i.e., $\mathbf{1}_{\mathcal{T\negthinspace}}\left(\emptyset\right)$.
Note that $1/U^{i}$ cancels out the unit $U^{i}$ of $dx_{1:i}$.
Further, integration of a unitless (or dimensionless) function $f$
on $\mathcal{F}(\mathcal{X})$ is realized via the Lebesgue integral
with respect to (w.r.t.) the Poisson measure $\mu$ \cite{MW2007Modern},
\cite{G1999Likelihood}:
\begin{align*}
\int\negthinspace f(X)\mu(dX)= & \sum_{i=0}^{\infty}\frac{1}{i!U^{i}}\negthinspace\int\negthinspace f\negthinspace\left(\left\{ x_{1:i}\right\} \right)\negthinspace dx_{1:i}.
\end{align*}

Analogous to random vector, the \textit{probability density }of an
RFS $\Sigma$ is defined as a non-negative function $p_{\Sigma}$
on $\mathcal{F}(\mathcal{X})$ such that for any (measurable) $\mathcal{T}\subseteq\mathcal{F}(\mathcal{X})$,
\[
\Pr(\Sigma\in\mathcal{T})=\int_{\mathcal{T}}p_{\Sigma}(X)\mu(dX)=\int\negthinspace\mathbf{1}_{\mathcal{T\negthinspace}}(X)p_{\Sigma}(X)\mu(dX),
\]
i.e., integrating the probability density yields the probability distribution.
This means the instantaneous probability per unit Poisson measure
at $X\in\mathcal{F}(\mathcal{X})$ is
\[
p_{\Sigma}(X)=\frac{\Pr(\Sigma\in dX)}{\mu(dX)},
\]
the (unit-less) Radon-Nikod\'{y}m derivative of the probability distribution
w.r.t. $\mu$. If $p_{\Sigma}(X)>0$ implies $p_{\Sigma}(Y)>0$ for
all $Y\subseteq X$, $\Sigma$ is said to be \textit{hereditary }\cite{Baddeley-etal-07}.
Unlike a random vector, the probability density $p_{\Sigma}(X)$ does
not have the usual likelihood interpretation, see \cite{VDPTV2018Modelbased}
for further details. 

\subsubsection{Belief Density}

\emph{Finite Set Statistics} (\textit{FISST}) offers an alternative
notion of density/integration on $\mathcal{F}(\mathcal{X})$, which
by-passes measure theoretic constructs \cite{Goodmanetal97,VoAES,MahlerBook07,MahlerBook14}.
In FISST, the\emph{ set integral} of a function $f$ on $\mathcal{F}(\mathcal{X})$
over a (compact) region $S\subseteq\mathcal{X}$ is defined as:
\begin{equation}
\int_{S}f(X)\delta X\triangleq\sum_{i=0}^{\infty}\frac{1}{i!}\int_{S^{i}}f(\{x_{1:i}\})dx_{1:i},\label{eq:set-integral}
\end{equation}
where $S^{i}$ denotes the $i$-fold Cartesian product of $S$, with
$S^{0}=\emptyset$, and the integral for $i=0$ is $f(\emptyset)$
by convention. Note that $f(\{x_{1:i}\})$ has unit of $U^{-i}$,
which cancels the unit $U^{i}$ of $dx_{1:i}$, hence $\int f(\{x_{1:i}\})dx_{1:i}$
and the set integral are unitless. 

The \textit{belief} (or \textit{FISST})\textit{ density} of an RFS
$\Sigma$ is defined as a non-negative function $\pi_{\Sigma}$ on
$\mathcal{F}(\mathcal{X})$ whose set integral over any region $S$
$\subseteq\mathcal{X}$ gives the \textit{belief functional} at $S$,
i.e.,
\[
\Pr(\Sigma\subseteq S)=\int_{S}\pi_{\Sigma}(X)\delta X.
\]
The belief functional $\Pr(\Sigma\subseteq S)$, as a function of
$S$, is not a probability distribution, but nonetheless, is a fundamental
descriptor of $\Sigma$ \cite{Goodmanetal97,MahlerBook07,MahlerBook14},
from which the belief density can be obtained by taking \textit{set
derivatives} (interested readers are referred to \cite{MahlerBook07,MahlerBook14}
for more details). A related fundamental descriptor is the \textit{void
probability} \textit{functional}, defined as $\Pr(\Sigma\cap S=\emptyset)$,
the probability that $S$ contains no points of $\Sigma$, i.e., $\Sigma$
is contained in the complement of $S$. This is the belief functional
at the complement of $S$ \cite{Chiu-Stoyan-13}, \cite{Daley88}.
Similarly, the \textit{capacity functional}, defined as the probability
that $S$ contains at least one point of $\Sigma$, is also a fundamental
descriptor because $\Pr(\Sigma\cap S\neq\emptyset)=1-\Pr(\Sigma\cap S=\emptyset)$
\cite{Chiu-Stoyan-13}. 

It is important to note that the probability density w.r.t. the Poisson
measure and the belief density are equivalent \cite{VoAES}:
\[
p_{\Sigma}(X)=U^{|X|}\pi_{\Sigma}(X),
\]
where $|X|$ denotes the cardinality (number of elements) of $X$.
Hence, they are collectively referred to as \textit{multi-object densities}.
For the purpose of introducing RFS algorithms, it is more convenient
to use the belief density. 

\subsubsection{Probability Generating Functional (PGFl)}

Another RFS fundamental descriptor pertinent to multi-object system
is the \textit{PGFl}, defined for any unitless \textit{test function}
$h:\mathcal{X}\rightarrow[0,1]$ as the expectation \cite{Daley88,MahlerBook07,Chiu-Stoyan-13,MahlerBook14}
\[
G_{\Sigma}[h]\triangleq\mathbb{E}\left[h^{\Sigma}\right]=\int h^{X}\pi_{\Sigma}(X)\delta X,
\]
where the\textit{ multi-object exponential} $h^{X}\triangleq\prod_{x\in X}h\left(x\right)$,
with the convention $h^{\emptyset}=1$. It is clear from the definition
that $G_{\Sigma}[h]\in[0,1]$, and $G_{\Sigma}[\mathbf{1}_{S}]=\Pr(\Sigma\subseteq S)$.
PGFls are analogous to probability generating functions. 

The convolution of multi-object densities translates to the multiplication
of PGFls. Suppose that $\Sigma$ is the union of disjoint and statistically
independent RFSs $\Sigma_{1}$,...,$\Sigma_{n}$. Then,
\begin{align}
G_{\Sigma}[h] & =G_{\Sigma_{1}}[h]...G_{\Sigma_{n}}[h],\label{eq:Convolution-PGFl}\\
\pi_{\Sigma}(X) & =\sum_{W_{1}\uplus...\uplus W_{n}=X}\pi_{\Sigma_{1}}(W_{1})...\pi_{\Sigma_{n}}(W_{n}),\label{eq:Convolution-PDF}
\end{align}
where the sum is taken over all mutually disjoint $W_{1}$,...,$W_{n}\subseteq X$
(including empty sets) that cover $X$ \cite[pp. 85-86]{MahlerBook07,MahlerBook14}. 

The multi-object density (and other statistics) can be obtained by
differentiating the PGFl \cite{Daley88,MahlerBook07,Chiu-Stoyan-13,MahlerBook14}.
For a functional $G$ on the space of test functions, let
\[
G^{(1)}[h;\zeta_{1}]\triangleq(dG)_{h}[\zeta]=\lim_{\epsilon\rightarrow0}\frac{G[h+\epsilon\zeta]-G[h]}{\epsilon},
\]
denote its \textit{G\^{a}teaux differential} at $h$ in the direction
$\zeta$ (if it exists). The $n$-th G\^{a}teaux \textit{differential}
at $h$ is a multi-linear form in the directions $\zeta_{1},...,\zeta_{n}$,
given recursively by
\[
G^{(n)}[h;\zeta_{1:n}]=(dG^{(n-1)}[\cdot;\zeta_{1:n-1}])_{h}[\zeta_{n}].
\]
Further, note that a multi-linear form can be expressed as $F[\zeta_{1:n}]=\int\zeta_{1}(y_{1})...\zeta_{n}(y_{n})f(y_{1:n})dy_{1:n}$,
and is completely characterized by the function $f$, which can be
rewritten in the Dirac delta notation $F[\delta_{x_{1}},...,\delta_{x_{n}}]\triangleq f(x_{1:n})$.
This is suggestive of evaluating $f$ at $x_{1},...,x_{n}$ via substituting
the Dirac deltas $\delta_{x_{1}},...,\delta_{x_{n}}$ into the integral
that defines $F$, i.e., $f(x_{1:n})$ can be treated as the value
of multi-linear form $F$ at $\delta_{x_{1}},...,\delta_{x_{n}}$.
Using the \textit{Volterra}\textit{ functional derivative} w.r.t.
a finite set \cite[p. 66]{MahlerBook14}, defined by
\begin{eqnarray*}
{\displaystyle {\displaystyle \frac{\delta G}{\delta\{x_{1:n}\}}[h]\triangleq G^{(n)}[h;\delta_{x_{1}},...,\delta_{x_{n}}]},}
\end{eqnarray*}
with $\frac{\delta G}{\delta\emptyset}[h]\triangleq G[h]$, we have
\cite[p. 95]{MahlerBook14}
\[
\frac{\delta G}{\delta X}[0]=\pi_{\Sigma}(X).
\]
Volterra functional set derivatives can be calculated using a suite
of differentiation rules \cite[pp. 383\textendash 395]{MahlerBook07},
including a powerful generalized Fa\`{a} di Bruno's chain rule \cite{CH2013Faa}.
Further properties and applications of PGFls in multi-object systems
can be found in the texts \cite{MahlerBook07,MahlerBook14}. \vspace{-0.2cm}

\subsection{Cardinality and Probability Hypothesis Density\protect\label{subsec:Card-PHD}}

In multi-object systems, relevant statistics of an RFS $\Sigma$ often
involve its \textit{cardinality} (number of elements). The probability
generating function (PGF) of the cardinality $|\Sigma|$, evaluated
at $z\in[0,1]$, is the PGFl evaluated at the test function $h(x)=z$,
i.e., $G_{|\Sigma|}(z)=G_{\Sigma}[h=z]$. The cardinality distribution
$\rho_{\Sigma}(n)\triangleq\Pr(\left\vert \Sigma\right\vert =n)$
can be computed via \cite{MahlerCPHD07,MahlerBook07,Chiu-Stoyan-13}
\[
\rho_{\Sigma}(n)=\frac{1}{n!}\int\pi_{\Sigma}(\{x_{1:n}\})dx_{1:n}=\frac{1}{n!}G_{|\Sigma|}^{(n)}(0),
\]
where $G_{|\Sigma|}^{(n)}$ is the $n$-th derivative of the PGF.
Statistics of the cardinality can be computed from $\rho_{\Sigma}$
or $G_{|\Sigma|}$.
\begin{figure}[t]
\begin{centering}
\resizebox{88mm}{!}{\includegraphics[clip]{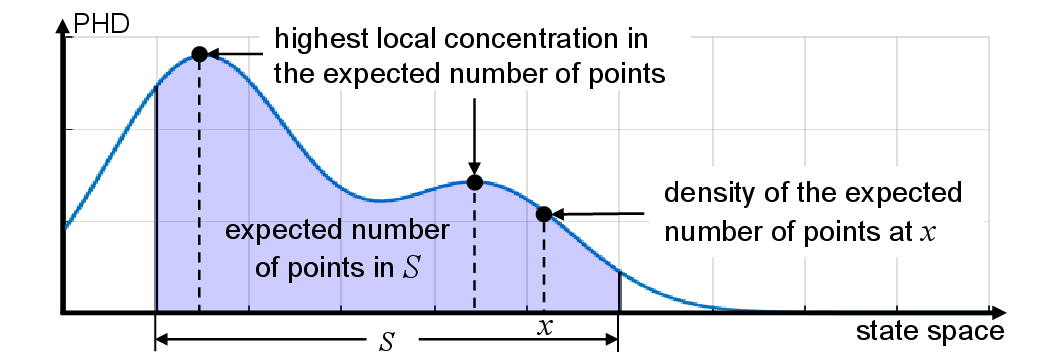}}
\par\end{centering}
\caption{PHD or intensity function on a 1-D state space.}
\label{fig:PHD}\vspace{-0mm}
\end{figure}

Another well-known cardinality-based statistic of an RFS is the \emph{intensity
function }\cite{Daley88,Chiu-Stoyan-13}, also known as the\emph{
Probability Hypothesis Density} (\textit{PHD}) in the MOT literature
\cite{Mahler2003,MahlerBook07}. As illustrated in Fig. \ref{fig:PHD},
it is defined as a non-negative function $v_{\Sigma}$ (on $\mathcal{X}$)
whose integral over any region $S$ $\subseteq\mathcal{X}$ gives
the expected cardinality in $S$, i.e.,
\begin{equation}
\mathbb{E}\left[\left\vert \Sigma\cap S\right\vert \right]=\int_{S}v_{\Sigma}(x)dx.\label{eq:PHD}
\end{equation}
The PHD can be computed from the multi-object density or PGFl by
\cite[p. 93]{MahlerBook14}
\[
v_{\Sigma}(x)=\int\pi_{\Sigma}(\{x\}\cup W)\delta W=\frac{\delta G}{\delta\{x\}}[1],
\]
and is the 1st of the \textit{factorial moment densities} \cite{Daley88,Chiu-Stoyan-13}:
\[
v_{\Sigma}(X)\triangleq\int\pi_{\Sigma}(X\cup W)\delta W=\frac{\delta G}{\delta X}[1].
\]
Interestingly, knowledge of the PHD $v_{\Sigma}$ is sufficient to
calculate the expectation of random sums of a (measurable) real function
$f$ over $\Sigma$, via Campbell's Theorem \cite{Chiu-Stoyan-13},
\cite{Baddeley-etal-07}
\[
\mathbb{E}\left[{\textstyle \sum_{x\in\Sigma}f(x)}\right]=\int f(x)v_{\Sigma}(x)dx.
\]
Definition (\ref{eq:PHD}) means that the PHD is the density of the
expected cardinality w.r.t. hyper-volume. This physically intuitive
interpretation is one of the factors behind the appeal of the PHD
filter. The local maxima of the PHD are points in $\mathcal{X}$ with
the highest local concentration of expected number of objects (per
unit hyper-volume), suggesting that they are the most likely states
for the underlying objects, see Fig. \ref{fig:PHD}. Based on such
interpretation, the \textit{PHD estimate} and \textit{cardinalized}\textit{
PHD }(\textit{CPHD})\textit{ estimate}, respectively, use $\hat{n}=\textrm{round}\left(\mathbb{E}\left[|\Sigma|\right]\right)$
and $\hat{n}=\arg\max_{n}\rho_{\Sigma}(n)$ as the estimated number
of objects, and the $\hat{n}$ highest local maxima of the PHD as
the estimated states, see also \cite[p. 595]{MahlerBook07} for more
details.

For a hereditary RFS, the intensity function can be extended to the
\textit{conditional intensity} at a point $x$, defined as \cite{Baddeley-etal-07}
\[
v_{\Sigma}\left(x|Y\right)=\frac{p_{\Sigma}\left(Y\cup\{x\}\right)}{p_{\Sigma}(Y)}.
\]
Note that the intensity function is given by the expectation $\int v_{\Sigma}\left(x|Y\right)\pi_{\Sigma}(Y)\delta Y.$
Moreover, the probability density $p_{\Sigma}$ (and belief density
$\pi_{\Sigma}$) is completely determined by its conditional intensity
\cite{Chiu-Stoyan-13,Baddeley-etal-07}. Working with the conditional
intensity eliminates the computation of the normalizing constant,
but requires certain consistency conditions.

\subsection{Multi-Object Estimators \protect\label{subsec:Multi-Object-Est}}

Given a probability density, estimators are needed to determine estimates
of the underlying random variable. Popular estimates such as the mean
and mode are not directly applicable to RFSs, because there is no
average for sets nor a likelihood interpretation for multi-object
density \cite{VDPTV2018Modelbased}. Nonetheless, as the 1st moment,
the PHD can be regarded as the expectation of an RFS, and hence the
PHD estimate (Subsection \ref{subsec:Card-PHD}) can be treated as
the mean estimate. The cardinality of the PHD estimate has a high
variance for a large number of objects since the variance of a Poisson
is equal to the mean. The CPHD estimate is an extension to improve
the cardinality estimate.

While the multi-object density does not have a likelihood interpretation
over $\mathcal{F}(\mathcal{X})$, when restricted to a cardinality
it can be interpreted as a likelihood. This observation leads to an
extension of the mode estimate, called the \textit{Marginal Multi-object
}(\textit{MaM})\textit{ estimate}, defined as \textit{the most probable
multi-object state given the most probable cardinality} \cite[pp. 497-498]{MahlerBook07}
(which requires computing: the most probable cardinality from the
cardinality distribution; and the supremum of the multi-object density
conditioned on this cardinality). Another estimate that emulates the
mode as an optimal Bayes estimator is the \textit{Joint Multi-object
}(\textit{JoM})\textit{ estimate} \cite[pp. 498-500]{MahlerBook07}
\[
\hat{X}=\arg\sup_{X\in\mathcal{F}(\mathcal{X})}\frac{c^{|X|}}{|X|!}\pi_{\Sigma}(X),
\]
where $c>0$ is a constant with magnitude in the order of the desired
accuracy, measured in units of hyper-volume in $\mathcal{X}$.\vspace{-0mm}

\subsection{Classical RFS Models\protect\label{subsec:Classical-Models}}

The following RFSs are popular models in point process and multi-object
SSM, especially for multi-object observations and multi-object localization.
Hereon, we denote the \emph{inner product} $\int f(\mathbf{\zeta})g(\mathbf{\zeta})d\mathbf{\zeta}$
of two functions $f$, $g$ (or $\sum_{\ell=0}^{\infty}f(\ell)g(\ell)$
when they are sequences) as $\langle f,g\rangle$, and the generalized
Kronecker delta that takes arbitrary arguments as
\[
\delta_{Y}\left[X\right]\triangleq\begin{cases}
1, & \text{if }X=Y\\
0, & \text{otherwise}
\end{cases}.
\]

\subsubsection{Poisson}

A \textit{Poisson RFS} (or Poisson point process) of $\mathcal{X}$
is completely characterized by its PHD (intensity function) $v_{\Sigma}$\textit{\emph{
}}\cite[p. 98]{Daley88,Chiu-Stoyan-13,MahlerBook14}, with multi-object
density and PGFl:
\begin{eqnarray*}
\pi_{\Sigma}(X)=e^{\texttt{-}\langle v_{\Sigma},1\rangle}\,v_{\Sigma}^{X}, &  & G_{\Sigma}[h]=e^{\langle v_{\Sigma},h\texttt{-}1\rangle}.
\end{eqnarray*}
Its cardinality is Poisson distributed with mean $\langle v_{\Sigma},1\rangle$,
i.e., $\rho_{\Sigma}(n)=e^{\texttt{-}\langle v_{\Sigma},1\rangle}\langle v_{\Sigma},1\rangle^{n}/n!$,
and conditioned on the cardinality (number of distinct elements),
each (distinct) element is i.i.d. according to $v_{\Sigma}/\langle v_{\Sigma},1\rangle$. 

The Poisson RFS models \textquotedblleft no interaction\textquotedblright{}
or \textquotedblleft complete spatial randomness\textquotedblright ,
and is one of the best-known and most tractable of point processes
\cite{Daley88,Chiu-Stoyan-13,VanLieshout-2000,MW2007Modern}. In multi-object
systems, it is a popular model for clutter or false alarms.

\subsubsection{I.I.D. Cluster}

An \emph{i.i.d. cluster} \textit{RFS} is a generalization of the Poisson
RFS to accommodate cardinality distributions other than Poisson, and
is completely characterized by its cardinality distribution $\rho_{\Sigma}$
and PHD $v_{\Sigma}$. Specifically, the multi-object density and
PGFl are \cite[p. 99]{Daley88,Chiu-Stoyan-13,MahlerBook14}:
\begin{eqnarray*}
\pi_{\Sigma}(X)=\frac{\rho_{\Sigma}(|X|)|X|!}{\langle v_{\Sigma},1\rangle^{|X|}}v_{\Sigma}^{X}, &  & G_{\Sigma}[h]=G_{|\Sigma|}\left(\frac{\langle v_{\Sigma},h\rangle}{\langle v_{\Sigma},1\rangle}\right).
\end{eqnarray*}
The cardinality of an i.i.d. cluster RFS is distributed according
to the prescribed $\rho_{\Sigma}$, and conditioned on the cardinality,
each (distinct) element is i.i.d. according to $v_{\Sigma}/\langle v_{\Sigma},1\rangle$.
When $\rho_{\Sigma}$ is Poisson, this reduces to the \textit{\emph{Poisson
RFS }}\cite{Daley88,Chiu-Stoyan-13}.
\begin{rem}
\label{rm2:iid_cluster}Multi-object densities for i.i.d. cluster
(and Poisson) RFSs are defined for a fixed dimensional $\mathcal{X}$.
They may not exist otherwise \cite[Prop. 5.4.V, p. 138]{Daley88},
\cite{MW2007Modern}. 
\end{rem}

\subsubsection{Multi-Bernoulli}

A \emph{Bernoulli} RFS, parameterized by the pair $(r,p)$, has probability
$1-r$ of being empty, and probability $r$ of being a singleton,
conditioned on which the element is distributed according to the probability
density $p$ (on $\mathcal{X}$). The multi-object density, PGFl,
cardinality distribution, and PHD are given by \cite[p. 100]{MahlerBook14}\allowdisplaybreaks
\begin{align*}
\pi_{\Sigma}(X) & =(1-r)\delta_{\emptyset}[X]+rp(x)\delta_{\{x\}}[X],\\
G_{\Sigma}[h] & =1-r+r\langle p,h\rangle,\\
\rho_{\Sigma}(n) & =(1-r)\delta_{0}[n]+r\delta_{1}[n],\\
v_{\Sigma}(x) & =rp(x).
\end{align*}

A \emph{multi-Bernoulli} RFS (for compactness we drop the `RFS') is
a union of disjoint and independent Bernoulli RFSs with parameters
$\{(r^{(\zeta)},p^{(\zeta)})\}_{\zeta\in\Psi}$ \cite[p. 101]{MahlerBook14}.
Its PGFl is the product of the constituent Bernoulli RFSs' PGFls,

\[
G_{\Sigma}[h]=\prod_{\zeta\in\Psi}(1-r^{(\zeta)}+r^{(\zeta)}\langle p^{(\zeta)},h\rangle),
\]
and its multi-object density is the convolution of the constituent
Bernoulli RFSs' densities \cite[p. 102]{MahlerBook14}
\begin{align}
\!\!\!\pi_{\Sigma}(\{x_{1:n}\})=\pi_{\Sigma}(\emptyset)\!\!\sum_{1\leq i_{1}\neq...\neq i_{n}\leq|\Psi|}\prod_{j=1}^{n}\frac{r^{(\zeta_{i_{j}})}p^{(\zeta_{i_{j}})}(x_{j})}{1-r^{(\zeta_{i_{j}})}},\!\!\label{eq:Multi-Bernoulli-density}
\end{align}
where $\pi_{\Sigma}(\emptyset)=\prod_{\zeta\in\Psi}(1-r^{(\zeta)})$,
and $\zeta_{1}$,..., $\zeta_{|\Psi|}$ enumerate all the elements
of $\Psi$. By convention, the sum reduces to $1$ when $n=0$, and
zero when $n>|\Psi|$. It is implicitly assumed that $r^{(\zeta)}\in[0,1)$.
Statistics such as the cardinality distribution and PHD are given
by \cite[p. 102]{MahlerBook14},
\begin{align}
\rho_{\Sigma}(n) & =\pi_{\Sigma}(\emptyset)\sum_{1\leq i_{1}<...<i_{n}\leq|\Psi|}\prod_{j=1}^{n}\frac{r^{(\zeta_{i_{j}})}}{1-r^{(\zeta_{i_{j}})}},\label{eq:Multi-Bernoulli-card}\\
v_{\Sigma}(x) & =\sum_{\zeta\in\Psi}r^{(\zeta)}p^{(\zeta)}(x).\label{eq:Multi-Bernoulli-PHD}
\end{align}
Multi-Bernoullis are often used for modeling object survival/death,
and detection uncertainty in observations. 
\begin{rem}
\label{rm3:negligible_dependence}For an i.i.d. cluster (and Poisson)
RFS, the joint distribution of the elements is conditioned on the
cardinality (i.e., the number of distinct elements), which ensures
distinctness of the elements. However, for a multi-Bernoulli, the
distribution of each element is conditional on its existence, independent
from others, and there is no mechanism to prevent two constituent
(independent) Bernoulli RFS's from sharing the same point. For example,
consider the Bernoulli RFSs $B^{(i)}$ with $r^{(i)}=0.5,p^{(i)}(\cdot)=\mathcal{N}(\cdot;0,1)$,
$i=1,2$. Noting that the likelihood of a realization $x_{1}$ from
$B^{(1)}$ and $x_{2}$ from $B^{(2)}$ is $0.25\mathcal{N}(x_{1};0,1)\mathcal{N}(x_{2};0,1)$,
the likelihood of $x_{1}=0$ and $x_{2}=0$ is not only positive,
but the highest possible among all values of $x_{1},x_{2}$, and hence,
not negligible. The disjoint condition between the constituent Bernoulli
RFSs implies dependence. The expressions for the multi-Bernoulli's
PGFl and multi-object density implicitly assume negligible dependence
between constituent components. This assumption is reasonable for
traditional detection processes based on thresholding, which return
distinct observations.  
\end{rem}

\section{Labeled Random Finite Set\protect\label{sec:LRFS-models}}

\begin{table}[!t]
\global\long\def\arraystretch{1.3}%
 \caption{Common notations from Section \ref{sec:LRFS-models} onwards.}

\vspace*{0.4cm}
 {\footnotesize\label{tbl:notation-LRFS}}{\footnotesize\par}
\centering{}%
\begin{tabular}{|c|l|}
\hline 
{\small\textbf{Notation}} & {\small\textbf{Description}}\tabularnewline
\hline 
{\small$\mathcal{A}$} & {\small Attribute projection $(x,\ell)\mapsto x$}\tabularnewline
{\small$\mathcal{L}$} & {\small Label projection $(x,\ell)\mapsto\ell$}\tabularnewline
{\small$\Delta(\boldsymbol{X})$} & {\small Distinct label indicator $\delta_{|\boldsymbol{X}|}[|\mathcal{L}(\boldsymbol{X})|]$}\tabularnewline
$\left\langle f\right\rangle \!(\{\ell_{1:n}\})$ & {\small Label-marginal }${\textstyle \int}f(\{(x_{1},\ell_{1}),...,(x_{n},\ell_{n})\})dx_{1:n}$\tabularnewline
{\small$\boldsymbol{\pi}_{\boldsymbol{\Sigma}}$} & {\small Multi-object density of (the LRFS) $\boldsymbol{\Sigma}$}\tabularnewline
{\small$w_{\boldsymbol{\Sigma}}(L)$} & {\small Probability that $\boldsymbol{\Sigma}$ has label set $L$}\tabularnewline
{\small$\rho_{\boldsymbol{\Sigma}}$} & {\small Cardinality distribution of $\boldsymbol{\Sigma}$}\tabularnewline
{\small$\boldsymbol{v}_{\boldsymbol{\Sigma}}$} & {\small PHD or intensity of $\boldsymbol{\Sigma}$}\tabularnewline
{\small$\boldsymbol{X}_{j:k}$} & {\small Multi-object sequence/trajectory on $\{j\textrm{:}k\}$}\tabularnewline
{\small$\mathcal{L}(\boldsymbol{X}_{j:k})$} & {\small$(\mathcal{L}(\boldsymbol{X}_{j}),...,\mathcal{L}(\boldsymbol{X}_{k}))$}\tabularnewline
{\small$T(\ell)$} & {\small Set of instants that $\boldsymbol{X}_{j:k}$ contains label
$\ell$}\tabularnewline
{\small$\boldsymbol{x}_{T(\ell)}^{(\ell)}$} & {\small$[(x_{i},\ell)\in\boldsymbol{X}_{i}]_{i\in T(\ell)}$, trajectory
of $\ell$ in $\boldsymbol{X}_{j:k}$ }\tabularnewline
{\small$h^{\boldsymbol{X}_{j:k}}$} & {\small${\displaystyle {\textstyle \prod}_{\ell\in\mathcal{L}(\boldsymbol{X}_{j:k})}}h(\boldsymbol{x}_{T(\ell)}^{(\ell)}),$}\tabularnewline
{\small${\displaystyle {\textstyle \cprods}{}_{i\text{=}j}^{k}}S_{i}$} & {\small$S_{j}\times...\times S_{k}$}\tabularnewline
{\small$\Delta(\boldsymbol{X}_{j:k})$} & {\small Multi-scan distinct label indicator $\prod_{i=j}^{k}\Delta(\boldsymbol{X}_{i})$}\tabularnewline
\hline 
\end{tabular}
\end{table}

A labeled RFS is a special class of RFSs introduced in \cite{VoConj11,VoConj13}
for modeling multi-object states/trajectories. It provides a versatile
multi-object estimation framework that admits characterization of
uncertainty for the multi-object trajectory ensemble, and meaningful
approximations with characterizable errors that are requisite for
principled solutions. 

Formally, a \textit{labeled RFS} (LRFS) with \textit{attribute space}
$\mathbb{X}$ and (discrete) \textit{label space} $\mathbb{L}=\{\mathcal{\alpha}_{i\!}:\!i\in\mathbb{N}\}$,
is an RFS $\boldsymbol{\Sigma}$ of the product space $\mathbb{X}\times\mathbb{L}$
such that each realization has distinct labels. Defining the \textit{attribute
projection} $\mathcal{A\!}:\!(x,\ell)\mapsto x,$ and \textit{label
projection} $\mathcal{L}\!:\!(x,\ell)\mapsto\ell,$ so that $\mathcal{A}(\boldsymbol{X})$
and $\mathcal{L}(\boldsymbol{X})$ are, respectively, the sets of
\textit{attributes} and \textit{labels} of $\boldsymbol{X}$, then
$\boldsymbol{X}$ has \textit{distinct labels} if the \textit{distinct
label indicator}
\begin{equation}
\Delta(\boldsymbol{X})\triangleq\delta_{|\boldsymbol{X}|}[|\mathcal{L}(\boldsymbol{X})|],\label{eq:distinct-label-indicator}
\end{equation}
equals 1, i.e., $\mathcal{L}(\boldsymbol{X})$ and $\boldsymbol{X}$
have the same cardinality. 

An LRFS can be thought of as a simple finite \textit{marked point
process}\footnote{A marked point process $\boldsymbol{\Sigma}$ of $\mathbb{X}\times\mathbb{M}$,
with mark space $\mathbb{M}$, satisfies $\left\vert \boldsymbol{\Sigma}\cap(S\times\mathbb{M})\right\vert <\infty$
for any bounded $S\subseteq\mathbb{X}$ \cite{Daley88,Baddeley-etal-07,MW2007Modern,Chiu-Stoyan-13}.} of $\mathbb{X}\times\mathbb{L}$, with distinct marks from the discrete
mark space. Note the distinction between ``simple finite marked''
and ``marked simple finite'' point processes. The former is simple,
but the point process formed by \textit{unmarking} (discarding the
marks) is not necessarily simple. The latter is the special case with
a simple unmarked version, because it is constructed by marking an
RFS, and hence, has the same cardinality as its unmarked version\cite{VoConj13}.
In general, an LRFS does not necessarily have the same cardinality
as its unmarked version. 

Subsection \ref{subsec:LRFS-def} introduces the concept of joint
existence probability and the ensuing multi-object estimators unique
to LRFSs. Subsections \ref{subsec:LIIDCluster} to \ref{subsec:GLMB}
summarize popular LRFS models. Information divergences for LRFS are
presented in Subsection \ref{subsec:Info-Div}, including closed-form
expressions  not previously published. LRFS approximations and spatio-temporal
modeling are discussed in Subsections \ref{subsec:LRFS-approx} and
\ref{subsec:Multi-scan-GLMB}. 

Following \cite{VoConj13}, vectors are represented by lower case
letters (e.g., $x$ and $\boldsymbol{x}$), and finite sets are represented
by upper case letters (e.g., $X$ and $\boldsymbol{X}$), where the
symbols for labeled entities and their distributions are bolded (e.g.,
$\boldsymbol{x}$, $\boldsymbol{X}$, $\boldsymbol{\pi}$, etc.) to
distinguish them from unlabeled ones. 

\subsection{Joint Existence Probability\protect\label{subsec:LRFS-def}}

Pertinent to multi-object estimation and unique to LRFSs is the notion
of joint existence probability. Let us denote the \textit{label-marginal}
of a function $f:\mathcal{F}(\mathbb{X}\times\mathbb{L})\rightarrow\mathbb{R}$,
by 
\begin{equation}
\left\langle f\right\rangle \!(\{\ell_{1:n}\})\triangleq{\textstyle \int}f(\{(x_{1},\ell_{1}),...,(x_{n},\ell_{n})\})dx_{1:n},
\end{equation}
with $\left\langle f\right\rangle \!(\emptyset)\triangleq f(\emptyset)$.
For an LRFS $\boldsymbol{\Sigma}$ with multi-object density $\boldsymbol{\pi}_{\boldsymbol{\Sigma}}$,
we define the \textit{joint existence probability} of $L\subseteq\mathbb{L}$,
and the \textit{label-conditioned joint attribute }(\textit{probability})\textit{
density} for distinct labels $\ell_{1:n}$, respectively, as \cite{PapiVoVoetal15}
\begin{align}
w_{\boldsymbol{\Sigma}}(L)\thinspace\mathbf{\triangleq}\thinspace & \left\langle \boldsymbol{\pi}_{\boldsymbol{\Sigma}}\right\rangle \!(L),\label{eq:joint-existence-prob}\\
\pi_{\boldsymbol{\Sigma}|\ell_{1:n}}(x_{1:n})\thinspace\triangleq\thinspace & \frac{\boldsymbol{\pi}_{\boldsymbol{\Sigma}}(\{(x_{1},\ell_{1}),...,(x_{n},\ell_{n})\})}{w_{\boldsymbol{\Sigma}}(\{\ell_{1:n}\})},\label{eq:label-conditioned-joint}
\end{align}
with the convention that $\pi_{\boldsymbol{\Sigma}|\ell_{1:n}}(x_{1:n})=1$
whenever $w_{\boldsymbol{\Sigma}}(\{\ell_{1:n}\})=0$. Then $w_{\boldsymbol{\Sigma}}(\cdot)$
is a probability distribution on $\mathcal{F}(\mathbb{L})$, and $\pi_{\boldsymbol{\Sigma}|\ell_{1:n}}(\cdot)$
is a probability density on $\mathbb{X}^{n}$ \cite{PapiVoVoetal15}.
The joint existence probability $w_{\boldsymbol{\Sigma}}(L)$ is the
probability that the LRFS $\boldsymbol{\Sigma}$ has label set $L$.

In the context of estimators, the joint existence probability is more
informative than the cardinality distribution \cite{PapiVoVoetal15}
\begin{equation}
\rho{}_{\boldsymbol{\Sigma}}(n)=\sum_{L\subseteq\mathbb{L}}\delta_{n}[|L|]w_{\boldsymbol{\Sigma}}\left(L\right),\label{eq:LRFS-cardinality}
\end{equation}
in the following sense. The label sets of cardinality $n$, each
with small $w_{\boldsymbol{\Sigma}}\left(L\right)$, could accumulate
up to a large $\rho{}_{\boldsymbol{\Sigma}}(n)$. Thus, sets with
the most probable cardinality, e.g., the MaM estimate in Subsection
\ref{subsec:Multi-Object-Est}, could have negligible joint existence
probability compared to the highest joint existence probability achievable
(by sets with other cardinalities). 

A judicious alternative to the MaM estimate is the \textit{label-MaM}
estimate, defined as \textit{the most probable multi-object state
given the most probable label set} (which is meaningful because the
label set is a discrete variable). This estimator seeks the label
set $L^{*}$ with highest joint existence probability, whose attributes
are given by the mode of the corresponding label-conditioned joint
attribute density. If the most probable label set is not unique, we
select one with the most probable cardinality, and if this is still
not unique, we select the one that yields the highest label-conditioned
joint attribute density. 

LRFS furnishes the PHD or intensity value $\boldsymbol{v}_{\boldsymbol{\Sigma}}(x,\ell)$
with an additional interpretation as the attribute PHD at $x$, of
(the object with) label $\ell$. Recall the definition of the PHD
from (\ref{eq:PHD}) that for any $S\subseteq\mathbb{X}$, the expected
number of objects with attributes in $S$ and label $\ell$ is given
by
\begin{eqnarray*}
\mathbb{E}\left[\left\vert \boldsymbol{\Sigma}\cap(S\times\{\ell\})\right\vert \right]=\int_{S}\boldsymbol{v}_{\boldsymbol{\Sigma}}(x,\ell)dx.
\end{eqnarray*}
Thus, for a given label $\ell$, the function $\boldsymbol{v}_{\boldsymbol{\Sigma}}(\cdot,\ell)$
is its attribute PHD. Further, the distinct label property means that
the cardinality $\left\vert \boldsymbol{\Sigma}\cap(\mathbb{X}\times\{\ell\})\right\vert $
is either 1, if $\boldsymbol{\Sigma}$ has a member with label $\ell$,
or 0 otherwise. Consequently, the expectation $\mathbb{E}[\left\vert \boldsymbol{\Sigma}\cap(\mathbb{X}\times\{\ell\})\right\vert ]$
(i.e., the total attribute PHD mass of $\ell$) cannot exceed 1.
Hence, the \textit{existence probability} of $\ell$ and its \textit{attribute
}(\textit{probability})\textit{ density}, in an LRFS $\boldsymbol{\Sigma}$,
are given by
\begin{eqnarray}
\!\!\!\!\!\!\!\!\!r_{\boldsymbol{\Sigma}}(\ell)=\int\boldsymbol{v}_{\boldsymbol{\Sigma}}(x,\ell)dx, & \!\!\!\! & p_{\boldsymbol{\Sigma}}\left(\cdot,\ell\right)=\boldsymbol{v}_{\boldsymbol{\Sigma}}(\cdot,\ell)/r_{\boldsymbol{\Sigma}}(\ell).\label{eq:existence-prob}
\end{eqnarray}

\begin{rem}
\label{rm4:considering_other_labels}It is imperative to note that
$r_{\boldsymbol{\Sigma}}\left(\ell\right)\neq w_{\boldsymbol{\Sigma}}(\{\ell\})$
and $p_{\boldsymbol{\Sigma}}\left(x,\ell\right)\neq\pi_{\boldsymbol{\Sigma}|\ell}(x)$.
Whereas $w_{\boldsymbol{\Sigma}}(\{\ell\})$ is the joint probability
that (only) $\ell$ exists and the other labels do not, $r_{\boldsymbol{\Sigma}}\left(\ell\right)$
is the marginal probability that $\ell$ exists regardless of other
labels. Similarly, $\pi_{\boldsymbol{\Sigma}|\ell}(\cdot)$ is the
attribute density given that only $\ell$ exists, whereas $p_{\boldsymbol{\Sigma}}\left(\cdot,\ell\right)$
is the attribute density given that $\ell$ exists regardless of others.
The marginal nature of $r_{\boldsymbol{\Sigma}}\left(\ell\right)$
and $p_{\boldsymbol{\Sigma}}\left(x,\ell\right)$ can be seen from
the definition of the PHD $\boldsymbol{v}_{\boldsymbol{\Sigma}}(x,\ell)$.
\end{rem}
The PHD (of LRFS) offers inexpensive sub-optimal JoM and MaM estimates.
The sub-optimal JoM estimator seeks the labels with existence probabilities
above a prescribed threshold, and estimates their attributes from
the corresponding attribute densities (via the modes/means). This
estimate depends on the existence threshold similar to the constant
$c$ in the JoM estimate. The sub-optimal MaM estimator seeks the
most probable cardinality $n^{*}$ (either from $\rho_{\boldsymbol{\Sigma}}$,
if available, or from the multi-Bernoulli cardinality approximation
using the existence probabilities), and the $n^{*}$ labels with highest
existence probabilities, whose attributes are then estimated from
their attribute densities. The PHD also offers a tractable sub-optimal
label-MaM estimate, see Subsection \ref{subsec:LMB}.

\subsection{Labeled I.I.D. Cluster\protect\label{subsec:LIIDCluster}}

Analogous to its unlabeled counterpart, a\textit{ labeled i.i.d. cluster}
$\boldsymbol{\Sigma}$ is an LRFS characterized by a cardinality distribution
$\rho_{\boldsymbol{\Sigma}}$ and an attribute probability density
$f_{\boldsymbol{\Sigma}}$ (on $\mathbb{X}$) \cite{VoConj13}. Conditioned
on cardinality $n$, the $n$ (not necessarily distinct) attributes
i.i.d. according to $f_{\boldsymbol{\Sigma}}$, are marked with distinct
labels from $\mathbb{L}(n)\triangleq\{\mathcal{\alpha}_{i}\!:\!i=1\!:\!n\}$.
A labeled i.i.d. cluster has multi-object density and PGFl \cite[p. 450]{VoConj13,MahlerBook14}
\begin{align*}
\boldsymbol{\pi}_{\boldsymbol{\Sigma}}\left(\boldsymbol{X}\right) & =\delta_{\mathbb{L}(|\boldsymbol{X}|)}[\mathcal{L}(\boldsymbol{X})]\rho_{\boldsymbol{\Sigma}}(|\mathcal{L}(\boldsymbol{X})|)\left(f_{\boldsymbol{\Sigma}}\circ\mathcal{A}\right)^{\boldsymbol{X}},\\
\boldsymbol{G}_{\boldsymbol{\Sigma}}[\boldsymbol{h}] & =\sum\limits_{n=0}^{\infty}\rho_{\boldsymbol{\Sigma}}(n)\prod\limits_{\ell\in\mathbb{L}(n)}\langle f_{\boldsymbol{\Sigma}},\boldsymbol{h}(\cdot,\ell)\rangle,
\end{align*}
respectively, for any unitless test function $\boldsymbol{h}$ on
$\mathbb{X}\mathcal{\times}\mathbb{L}$, where $\circ$ denotes composition.
It is clear from the above descriptors that a labeled i.i.d. cluster
is not an i.i.d. cluster. The PHD of a labeled i.i.d. cluster is given
by \cite[p. 451]{MahlerBook14}\vspace{-0.1cm}
\[
\boldsymbol{v}_{\boldsymbol{\Sigma}}(x,\ell)=f_{\boldsymbol{\Sigma}}(x)\sum\limits_{n=0}^{\infty}\mathbf{1}_{\mathbb{L}(n)\negthinspace}(\ell)\rho_{\boldsymbol{\Sigma}}(n).
\]
The sum over $n$ is the existence probability of $\ell$, and $f_{\boldsymbol{\Sigma}}$
is its attribute density (independent of $\ell$). The \textit{labeled
Poisson} is the special case with Poisson cardinality distribution.
Note that a labeled Poisson RFS is not a Poisson RFS of $\mathbb{X}\times\mathbb{L}$.
\vspace{-0.1cm} 

\subsection{Labeled Multi-Bernoulli\protect\label{subsec:LMB}}

Similar to a multi-Bernoulli, a \textit{labeled multi-Bernoulli} (LMB)
is an LRFS characterized by a collection of independent Bernoulli
RFSs (of the attribute space $\mathbb{X}$) with parameters $\{(r^{(\zeta)},p^{(\zeta)})\,{\textstyle :\,}\zeta\!\in\!\Psi\}$,
and additionally, a 1-1 (injective) mapping $\sigma\,{\textstyle :\,}\mathbb{L}\rightarrow\Psi$
that pairs each $\zeta\in\Psi$ with a distinct label $\ell\in\mathbb{L}$
\cite{VoConj13}. For each $\zeta\in\Psi$, a \textit{labeled Bernoulli
}RFS is constructed by marking the Bernoulli RFS parameterized by
$(r^{(\zeta)},p^{(\zeta)})$ with the associated label $\ell=\sigma^{\texttt{-}1}(\zeta)$.
The resulting labeled Bernoulli RFSs are disjoint (due to their distinct
labels), and their union is an LMB $\boldsymbol{\Sigma}$, whose
multi-object density is given by \cite{VoConj13}
\begin{equation}
\boldsymbol{\pi_{\boldsymbol{\Sigma}}}\left(\boldsymbol{X}\right)=\Delta(\boldsymbol{X})\left[\mathbf{1}_{\mathcal{D}(\sigma)}\right]^{\mathcal{L}(\boldsymbol{X})}\left[\boldsymbol{\pi_{\boldsymbol{\Sigma}}}(\boldsymbol{X};\cdot)\right]^{\Psi},\label{eq:LabeledMB}
\end{equation}
where $\mathcal{D}(\sigma)$ is the domain of $\sigma$, and
\[
\boldsymbol{\pi_{\boldsymbol{\Sigma}}}(\boldsymbol{X};\zeta)=\sum_{(x,\ell)\in\boldsymbol{X}}\!\!\!\!\delta_{\sigma(\ell)\!}[\zeta]r^{(\zeta)\!}p^{(\zeta)\!}(x)\!+\!(1\!-\!\delta_{\sigma(\ell)\!}[\zeta])(1\!-\!r^{(\zeta)\!}).
\]
The above sum either takes on: $r^{(\zeta)}p^{(\zeta)}(x)$ if $\zeta$
matches the label of $\left(x,\ell\right)\in\boldsymbol{X}$, i.e.,
$\left(x,\mathcal{\sigma}^{\texttt{-}1}(\zeta)\right)\in\boldsymbol{X}$;
or $1-r^{(\zeta)}$ if $\zeta$ is not matched with any labels in
$\mathcal{L}\left(\boldsymbol{X}\right)$, i.e., $\zeta\notin\mathcal{\sigma}(\mathcal{L}\left(\boldsymbol{X}\right))$.
Hence, it can be written in piece-wise form
\[
\boldsymbol{\pi_{\boldsymbol{\Sigma}}}(\boldsymbol{X};\zeta)=\begin{cases}
r^{(\zeta)}p^{(\zeta)}(x), & \!\!\!\text{if }\left(x,\mathcal{\sigma}^{\texttt{-}1}(\zeta)\right)\in\boldsymbol{X}\\
1-r^{(\zeta)}, & \!\!\!\text{if }\zeta\notin\mathcal{\sigma}(\mathcal{L}\left(\boldsymbol{X}\right))
\end{cases}.
\]
The cardinality distribution of an LMB is given by (\ref{eq:Multi-Bernoulli-card}).

An LMB is parameterized by the existence probability $r_{\boldsymbol{\Sigma}}(\ell)\triangleq\mathbf{1}_{\mathcal{D}(\sigma)}(\ell)r^{(\sigma(\ell))}$,
and attribute density $p_{\boldsymbol{\Sigma}}(\cdot,\ell)\triangleq p^{(\sigma(\ell))}$
of each $\ell\in\mathcal{D}(\sigma)$. Indeed, its PGFl and multi-object
density can be expressed as follows \cite[p. 456]{MahlerBook14},
\cite{VoConj13}
\begin{align}
\boldsymbol{G}_{\boldsymbol{\Sigma}}[\boldsymbol{h}] & =\prod\limits_{\ell\in\mathcal{D}(\sigma)}\negthinspace\negthinspace\negthinspace\negthinspace\left(1-r_{\boldsymbol{\Sigma}}(\ell)+r_{\boldsymbol{\Sigma}}(\ell)\langle p_{\boldsymbol{\Sigma}}(\cdot,\ell),\boldsymbol{h}(\cdot,\ell)\rangle\right)\negthinspace,\label{eq:labeled_MB-1-1-1}\\
\boldsymbol{\pi_{\boldsymbol{\Sigma}}}(\boldsymbol{X}) & =\Delta(\boldsymbol{X})\left[1-r_{\boldsymbol{\Sigma}}\right]^{\mathcal{D}(\sigma)\texttt{-}\mathcal{L}\left(\boldsymbol{X}\right)}r_{\boldsymbol{\Sigma}}^{\mathcal{L}\left(\boldsymbol{X}\right)}p_{\boldsymbol{\Sigma}}^{\boldsymbol{X}}.\label{eq:LMB-density-1}
\end{align}
Note that $\sum_{L\subseteq\mathbb{L}}\left[1-r_{\boldsymbol{\Sigma}}\right]^{\mathcal{D}(\sigma)\texttt{-}L}r_{\boldsymbol{\Sigma}}^{L}=1$
\cite[p. 454]{MahlerBook14}, and hence $\boldsymbol{\pi_{\boldsymbol{\Sigma}}}$
integrates to 1. For brevity, we write the LMB multi-object density
$\boldsymbol{\pi_{\boldsymbol{\Sigma}}}$ as $\{(r_{\boldsymbol{\Sigma}}(\ell),p_{\boldsymbol{\Sigma}}(\cdot,\ell))\}_{\ell\in\mathcal{D}(\sigma)}$.
An LMB permits elements with the same attributes, and unmarking it
only yields a multi-Bernoulli if the attributes are distinct. 

Interestingly, much like the Poissons and Bernoullis' (both unlabeled
and labeled), the PHD of an LMB \cite[p. 457]{MahlerBook14}
\[
\boldsymbol{v}_{\boldsymbol{\Sigma}}(x,\ell)=r_{\boldsymbol{\Sigma}}(\ell)p_{\boldsymbol{\Sigma}}(x,\ell),
\]
also provides a complete characterization, since $r_{\boldsymbol{\Sigma}}(\ell)$
and $p_{\boldsymbol{\Sigma}}(\cdot,\ell)$ can be recovered from the
PHD via (\ref{eq:existence-prob}). Indeed, assuming $r_{\boldsymbol{\Sigma}}(\ell)\in[0,1)$
and noting that $\left[1-r_{\boldsymbol{\Sigma}}\right]^{\mathcal{D}(\sigma)}=\left[1-r_{\boldsymbol{\Sigma}}\right]^{\mathbb{L}}$,
$\boldsymbol{\pi_{\boldsymbol{\Sigma}}}$ can be written as a multi-object
exponential
\begin{equation}
\boldsymbol{\pi_{\boldsymbol{\Sigma}}}(\boldsymbol{X})=\Delta(\boldsymbol{X})\left[1-r_{\boldsymbol{\Sigma}}\right]^{\mathbb{L}}\left[\frac{\boldsymbol{v}_{\boldsymbol{\Sigma}}}{1-r_{\boldsymbol{\Sigma}}}\right]^{\boldsymbol{X}},\label{eq:LMB-multi-object-exp}
\end{equation}
which has the same form (neglecting the distinct label indicator)
as the Poisson $e^{\texttt{-}\langle v_{\Sigma},1\rangle}v_{\Sigma}^{X}$.
Consequently, the LMB shares many analytical properties with the
Poisson, and in this sense, is more ``Poisson'' than the labeled
Poisson. However, unlike the Poisson, the LMB cardinality variance
(which cannot exceed the mean) can be controlled by the existence
probabilities $r_{\boldsymbol{\Sigma}}(\ell),\ell\in\mathcal{D}(\sigma)$,
making it more versatile. 

Noting that the joint existence probability $w_{\boldsymbol{\Sigma}}(\cdot)$
of an LMB is given by $\left[1-r_{\boldsymbol{\Sigma}}\right]^{\mathcal{D}(\sigma)\texttt{-}L}r_{\boldsymbol{\Sigma}}^{L}$
in (\ref{eq:LMB-density-1}), the LMB admits an analytic labeled-MaM
estimate, with the mode $L^{*}$ of $w_{\boldsymbol{\Sigma}}(\cdot)$
as the most probable label set, and the modes of $p_{\boldsymbol{\Sigma}}(\cdot,\ell),\ell\in L^{*}$
as the most probable attributes (provided all the relevant modes are
available). For an arbitrary LRFS, a PHD-based sub-optimal label-MaM
estimate can be obtained by applying the label-MaM estimator to an
approximate LMB constructed from its PHD. 

\subsection{Generalized Labeled Multi-Bernoulli\protect\label{subsec:GLMB}}

The labeled i.i.d. cluster and LMB both have tractable multi-object
densities of the form $\Delta(\boldsymbol{X})w(\mathcal{L}\left(\boldsymbol{X}\right))p^{\boldsymbol{X}}$.
 Extending this form to a mixture accommodates a larger class of
LRFS that provides trade-offs between tractability and versatility.
A\textit{ Generalized Labeled Multi-Bernoulli }(GLMB) $\boldsymbol{\Sigma}$
is an LRFS distributed according to such a mixture \cite{VoConj11,VoConj13},
i.e.,
\begin{equation}
\boldsymbol{\pi_{\boldsymbol{\Sigma}}}\left(\boldsymbol{X}\right)=\Delta(\boldsymbol{X})\sum_{\xi\in\Xi}w^{\left(\xi\right)}(\mathcal{L}\left(\boldsymbol{X}\right))\left[p^{\left(\xi\right)}\right]^{\boldsymbol{X}},\label{eq:glmb}
\end{equation}
where $\Xi$ is a discrete set, each $p^{(\xi)}(\cdot,\ell)$ is a
(probability) density on $\mathbb{X}$, and each $w^{(\xi)}(L)$ is
a non-negative weight that satisfies $\sum_{\xi\in\Xi}\sum_{L\subseteq\mathbb{L}}w^{\left(\xi\right)}\!\left(L\right)=1$.
Intuitively, $w^{(\xi)}(L)$ can be interpreted as the probability
of \textit{hypothesis} $\left(\xi,L\right)$ representing the `event'
$\xi$ and the joint existence of the labels in $L$. Conditional
on hypothesis $\left(\xi,L\right)$, $p^{(\xi)}(\cdot,\ell)$ is the
attribute (probability) density of object $\ell\in L$. 
\begin{rem}
Each term of the GLMB density (\ref{eq:glmb}) is rather general,
covering a board class of LRFSs including labeled i.i.d. clusters,
LMBs, and their disjoint union (provided their label sets are disjoint).
Thus, the unmarked version of a GLMB is a general class of non-simple
point processes that includes the Poisson Multi-Bernoulli Mixture
\cite{GWGS2018Poisson} as a special case. 
\end{rem}
The PGFl of a GLMB takes the form \cite[p. 460]{MahlerBook14}
\begin{eqnarray*}
\boldsymbol{G}_{\boldsymbol{\Sigma}}[\boldsymbol{h}] & = & \sum_{\xi\in\Xi}\sum_{L\subseteq\mathbb{L}}w^{\left(\xi\right)}\left(L\right)\prod\limits_{\ell\in L}\langle p^{(\xi)}(\cdot,\ell),\boldsymbol{h}(\cdot,\ell)\rangle.
\end{eqnarray*}
Note that the closed-form void probability functional in \cite{BVVA2017Void}
can be obtained by substituting $1-\mathbf{1}_{\boldsymbol{S}}$ into
the PGFl. Further, the cardinality and PHD are given by \cite{VoConj13}
\begin{align*}
\rho{}_{\boldsymbol{\Sigma}}(n) & =\sum_{\xi\in\Xi}\sum_{L\subseteq\mathbb{L}}\delta_{n}[|L|]w^{\left(\xi\right)}\left(L\right),\\
\boldsymbol{v}_{\boldsymbol{\Sigma}}(x,\ell) & =\sum_{\xi\in\Xi}p^{(\xi)}(x,\ell)\sum_{L\subseteq\mathbb{L}}\mathbf{1}_{L}(\ell)w^{\left(\xi\right)}\left(L\right).
\end{align*}
Hence, the existence probability and attribute density of object $\ell$
are \cite{VoVoP14}, \cite{Reuter2014}
\begin{align}
r_{\boldsymbol{\Sigma}}\left(\ell\right) & =\sum_{\xi\in\Xi}\sum_{L\subseteq\mathbb{L}}\mathbf{1}_{L}(\ell)w^{\left(\xi\right)}\left(L\right),\label{eq:existence-prob-GLMB}\\
p_{\boldsymbol{\Sigma}}\left(x,\ell\right) & =\frac{1}{r_{\boldsymbol{\Sigma}}\left(\ell\right)}\sum_{\xi\in\Xi}p^{\left(\xi\right)}\left(x,\ell\right)\sum_{L\subseteq\mathbb{L}}\mathbf{1}_{L}(\ell)w^{\left(\xi\right)}\left(L\right).\label{eq:track-density-GLMB}
\end{align}

For the label-MaM estimate, the most probable label set of a GLMB
is the mode $L^{*}$ of the joint existence probability
\begin{eqnarray}
w_{\boldsymbol{\Sigma}}\left(L\right) & = & \sum_{\xi\in\Xi}w^{\left(\xi\right)}\left(L\right),\label{eq:GLMB-JointexistenceProb}
\end{eqnarray}
(note the distinction between the probability that $L$ exists regardless
of other labels, $\sum_{\xi\in\Xi}\sum_{I\supseteq L}w^{(\xi)}(I)$).
However, unlike the LMB, finding the mode of the label-conditioned
joint attribute density is not tractable in general. A sub-optimal
strategy is to find the mode/mean of the attribute density $p_{\boldsymbol{\Sigma}}\left(\cdot,\ell\right)$,
for each $\ell\in L^{*}$. Additionally, the structure of the GLMB
suggests a tractable and intuitive sub-optimal version of the MaM
estimate, especially when each $\xi$ corresponds to an actual probability
event. Based on the most probable cardinality $n^{*}$, instead of
finding the $n^{*}$ most probable states, the \textit{GLMB estimator}
seeks the most probable hypothesis $\left(\xi^{*},L^{*}\right)$ with
$|L^{*}|=n^{*}$, and estimates the attribute of each $\ell\in L^{*}$
from $p^{\left(\xi^{*}\right)}\left(\cdot,\ell\right)$, via the mode
(or mean). 

In numerical implementations, it is more convenient to write the
GLMB density in $\delta$\textit{-GLMB} form 
\[
\boldsymbol{\pi_{\boldsymbol{\Sigma}}}\left(\boldsymbol{X}\right)=\Delta(\boldsymbol{X})\sum\limits_{(\xi,I)\in\Xi\times\mathcal{F}(\mathbb{L})}w^{(\xi)}(I)\delta_{I}[\mathcal{L}(\boldsymbol{X})]\left[p^{(\xi)}\right]^{\boldsymbol{X}},
\]
using the identity $w^{\left(\xi\right)}(L)=\sum_{I\subseteq\mathbb{L}}w^{(\xi)}(I)\delta_{I}[L]$.
For brevity, we denote a GLMB density $\boldsymbol{\pi_{\boldsymbol{\Sigma}}}$
by the set 
\begin{eqnarray}
\{(w^{\left(\xi\right)}(I),p^{\left(\xi\right)})\!:\!(\xi,I)\in\Xi\times\mathcal{F}(\mathbb{L})\}\label{eq:GLMB-component}
\end{eqnarray}
of its basic\textit{ components}, corresponding to set-functions that
cannot be decomposed as sums of simpler terms. Note that truncating
arbitrary basic components still leaves the resulting $\delta$-GLMB
a valid set-function, and does not suffer from the\textit{ non-symmetric
problem} discussed in Subsection \ref{subsec:Tricky}. 
\begin{rem}
\label{rm6:important_feature}An important feature of the GLMB family
is its closure under the Bayes filtering and posterior recursions
\cite{VoConj13} (see Subsections \ref{ss:GLMB-filter} and \ref{ss:Multi-object_Smoother}).
Moreover, it is furnished with convenient mathematical properties
that facilitate principled approximations (see Subsection \ref{subsec:LRFS-approx}).
\end{rem}

\subsection{Information Divergence\protect\label{subsec:Info-Div}}

Analogous to the classical SSM, information divergences measuring
statistical similarities/dissimilarities between RFSs are fundamental
in multi-object SSMs. Well-known divergences such as $f$-divergences
(or Csisz\'{a}r-Morimoto, Ali-Silvey) have been extended to RFS \cite[pp. 153-160]{MahlerBook14}.
This subsection presents a number of divergences that admit tractable
analytic expressions for certain LRFS models, including linear complexity
divergences for LMBs, which have not been previously published. 

\subsubsection{R\'{e}nyi}

The \textit{R\'{e}nyi divergence} (\textit{R\'{e}nyi-D}) between the multi-object
densities $\pi_{1}$ and $\pi_{2}$ is given by \cite[p. 156]{MahlerBook14}
\begin{equation}
D_{R}\!\left(\pi_{1}||\pi_{2}\right)=\frac{1}{\alpha-1}\ln\int\pi_{1}^{\alpha}\left(X\right)\pi_{2}^{1\texttt{-}\alpha}\left(X\right)\delta X.\label{e:CS_Divergence-1}
\end{equation}
When $\alpha=0.5$, $D_{R}$ reduces to the \textit{Bhattacharyya}\textit{
distance}, and is related to the squared \textit{Hellinger distance}.

While the R\'{e}nyi-D is intractable in general, it admits closed-forms
for (unlabeled) Poissons and LMBs. Indeed, for two Poissons with PHDs
$v_{i}$, $i=1,2$, i.e., $\pi_{i}(X)=K_{i}v_{i}^{X}$, $K_{i}\triangleq e^{\texttt{-}\left\langle v_{i},1\right\rangle }$,
the R\'{e}nyi-D is given by \cite[p. 158]{RVC2011Note,MahlerBook14}
\[
D_{R}\!\left(\pi_{1}||\pi_{2}\right)=\frac{\left\langle v_{1}^{\alpha},v_{2}^{1\texttt{-}\alpha}\right\rangle \texttt{+}\ln\left(K_{1}^{\alpha}K_{2}^{1\texttt{-}\alpha}\right)}{\alpha\texttt{-}1}.
\]
For each LMB parameterized by $\{(r_{i}(\ell),p_{i}(\cdot,\ell))\!:\!\ell\in\mathcal{D}(\sigma_{i})\}$,
$i=1,2$, let $\tilde{r}_{i}\triangleq1-r_{i}$, and $\boldsymbol{f}_{\!i}\triangleq r_{i}p_{i}/\tilde{r}_{i}$,
so that its multi-object density has the form\footnote{We implicitly assume $r_{i}(\ell)\in[0,1)$ to illustrate the similarities
between LMBs and Poissons. To include $r_{i}(\ell)=1$, we need to
use the form (\ref{eq:LMB-density-1}). } $\boldsymbol{\pi}_{i}(\boldsymbol{X})=\Delta(\boldsymbol{X})\tilde{r}_{i}^{\mathbb{L}}\boldsymbol{f}_{\!i}^{\boldsymbol{X}}$.
It can be shown that (see Supplementary Materials for proof and
alternative forms)
\begin{equation}
\!\!\!\!D_{R}\!\left(\boldsymbol{\pi}_{1\!}||\boldsymbol{\pi}_{2}\right)\!=\!\sum_{\ell\in\mathbb{L}}\!\cfrac[l]{\!\ln\!\left[1\texttt{+}\!\left\langle \boldsymbol{f}_{\!1}^{\alpha}\!\boldsymbol{f}_{\!2}^{1\texttt{-}\alpha}\right\rangle \!(\ell)\right]\!\texttt{+}\!\ln\!\left[\left(\tilde{r}_{1}^{\mathbb{\alpha}}\tilde{r}_{2}^{1\texttt{-}\alpha}\right)\!\!(\ell)\right]}{\alpha\texttt{-}1}.\!\!\label{eq:Renyi-LMB}
\end{equation}

The above expression reduces a series of high-dimensional integrals
to a series of integrals on the attribute space $\mathbb{X}$. The
sum over $\mathbb{L}$ is, in fact, only a sum over the union $\mathcal{D}(\sigma_{1})\cup\mathcal{D}(\sigma_{2})\subset\mathbb{L}$.
Outside this union, the existence probabilities $r_{1}(\ell)$ and
$r_{2}(\ell)$ are zero, which means $\left(\tilde{r}_{1}^{\mathbb{\alpha}}\tilde{r}_{2}^{1\texttt{-}\alpha}\right)\!(\ell)=1$,
and has no contribution to $D_{R}\!\left(\boldsymbol{\pi}_{1}||\boldsymbol{\pi}_{2}\right)$.
Moreover, $\left\langle \boldsymbol{f}_{\!1}^{\alpha}\boldsymbol{f}_{\!2}^{1\texttt{-}\alpha}\right\rangle (\ell)$
only contributes to $D_{R}\!\left(\boldsymbol{\pi}_{1}||\boldsymbol{\pi}_{2}\right)$
on the intersection $\mathcal{D}(\sigma_{1})\cap\mathcal{D}(\sigma_{2})$,
where both $r_{1}(\ell)$ and $r_{2}(\ell)$ are non-zero.

\subsubsection{Kullback-Leibler}

The limiting case of the R\'{e}nyi-D when $\alpha$ tends to 1 is the
\textit{Kullback-Leibler divergence} (\textit{KL-D}), given by \cite[p. 155]{MahlerBook14}
(with $0\ln0=0$ by convention)
\begin{equation}
D_{KL}\!\left(\pi_{1}||\pi_{2}\right)=\int\pi_{1}\left(X\right)\ln\frac{\pi_{1}\left(X\right)}{\pi_{2}\left(X\right)}\delta X.\label{e:CS_Divergence-1-1}
\end{equation}

Similar to the R\'{e}nyi-D, the KL-D is computationally intractable in
general, but admits closed-forms for Poissons and LMBs. Indeed, the
KL-D for Poissons \cite[p. 157]{MahlerBook14}
\[
D_{KL}\!\left(\pi_{1}||\pi_{2}\right)=\ln\frac{K_{1}}{K_{2}}+\left\langle v_{1},\ln\frac{v_{1}}{v_{2}}\right\rangle ,
\]
resembles that for LMBs \cite{SZDJL2022Consensus}
\begin{equation}
{\color{black}{\color{black}D_{KL}\!\left(\boldsymbol{\pi}_{1}||\boldsymbol{\pi}_{2}\right)=\sum_{\ell\in\mathbb{L}}}\mathinner{\color{black}\left[\ln\!\frac{\tilde{r}_{1}(\ell)}{\tilde{r}_{2}(\ell)}+{\color{red}{\color{black}\tilde{r}_{1}(\ell}\mathclose{\color{black})}}\!\left\langle \boldsymbol{f}_{\!1}\!\ln\!\frac{\boldsymbol{f}_{\!1}}{\boldsymbol{f}_{\!2}}\right\rangle \!(\ell)\right]}}\!,\label{eq:KL-LMB}
\end{equation}
(see Supplementary Materials for proof and alternative forms). Note
that the differential entropy for LMBs derived in \cite{NVVRR2022Multi}
follows from the above expression by setting $\boldsymbol{\pi}_{2}(\boldsymbol{X})$
to $\Delta(\boldsymbol{X})U^{\texttt{-}|\boldsymbol{X}|}$, where
$U$ is the unit of hyper-volume. 

\subsubsection{Chi-squared}

Similar to R\'{e}nyi-D and KL-D, the \textit{Chi-squared divergence} ($\chi^{2}$-\textit{D}),
given by \cite[p. 155]{MahlerBook14}
\begin{equation}
D_{\chi^{2}}\!\left(\pi_{1}||\pi_{2}\right)=\int\frac{\pi_{1}^{2}\left(X\right)}{\pi_{2}\left(X\right)}\delta X-1,\label{e:CS_Divergence-1-1-1}
\end{equation}
belongs to the general class of $f$-divergences that extends to RFS
simply by replacing standard density/integration with FISST density/integration.
The $\chi^{2}$-D is a 2nd-order Taylor's series approximation of
the KL-D, and together with the squared Hellinger distance, provides
respectively, the upper and lower bounds for the KL-D \cite{GS2002Choosing}.

For LMBs, the $\chi^{2}$-D admits the following closed-form 
\begin{equation}
D_{\chi^{2}}\!\left(\boldsymbol{\pi}_{1}||\boldsymbol{\pi}_{2}\right)=\prod_{\ell\in\mathbb{L}}\frac{\tilde{r}_{1}^{2}(\ell)}{\tilde{r}_{2}(\ell)}\!\left[1+\left\langle \frac{\boldsymbol{f}_{1}^{2}}{\boldsymbol{f}_{2}}\right\rangle \!(\ell)\right]-1\label{eq:Chi-S-D}
\end{equation}
(see Supplementary Materials for proof and alternative forms), which
incurs similar computational complexity to those of R\'{e}nyi-D and KL-D.
Note that only $\ell\in\mathcal{D}(\sigma_{1})\cup\mathcal{D}(\sigma_{2})$
contributes to $D_{\chi^{2}}\!\left(\boldsymbol{\pi}_{1}||\boldsymbol{\pi}_{2}\right)$,
because outside this union $\tilde{r}_{1}$ and $\tilde{r}_{2}$ become
unity whilst $\boldsymbol{f}_{1}^{2}/\boldsymbol{f}_{2}$ vanishes
(using the convention $0^{2}/0=0$). Like the R\'{e}nyi-D and KL-D, the
$\chi^{2}$-D for LMBs bears some resemblance to that for Poissons
\cite[p. 157]{MahlerBook14}
\[
D_{\chi^{2}}\!\left(\pi_{1}||\pi_{2}\right)=\frac{K_{1}^{2}}{K_{2}}e^{\left\langle v_{1}^{2}/v_{2},1\right\rangle }-1.
\]

\subsubsection{Cauchy-Schwarz}

Unlike the $f$-divergences, the \textit{Cauchy-Schwarz divergence}
(\textit{CS-D}) cannot be extended to RFS by replacing standard density/integration
with FISST density/integration due to incompatibility with the unit
of measurements \cite{HVVM2015Cauchy}. Consequently, extension to
RFS is accomplished using density/integration w.r.t. the Poisson measure.
Nonetheless, using their equivalence with the FISST density/integral,
the CS-D can be expressed as \cite{HVVM2015Cauchy}
\begin{equation}
\!\!\!\!D_{CS}\!\left(\pi_{1},\pi_{2}\right)=-\ln\!\frac{\int U^{\left|X\right|}\pi_{1\!}\left(X\right)\!\pi_{2\!}\left(X\!\right)\!\delta X}{\!\sqrt{\!\int\!U^{\left|X\right|}\pi_{1\!}^{2\!}\left(X\right)\!\delta\!X\!\int\!U^{\left|X\right|}\pi_{2\!}^{2\!}\left(X\right)\!\delta\!X}},\!\!\label{e:CS_Divergence-2-1}
\end{equation}
where $U$ is the unit of hyper-volume in $\mathcal{X}$ (the factor
$U^{|X|}$ ensures that all constituent integrals of the set integrals
are dimensionless). It can also be interpreted as an approximation
to the KL-D \cite{KHP2011Closed}. Note that the CS-D between the
square roots of the multi-object densities is the Bhattacharyya distance
(R\'{e}nyi-D with $\alpha=0.5$). 

Geometrically, the dissimilarity between $\pi_{1}$ and $\pi_{2}$,
according to the CS-D, is based on the angle they subtend in the space
of square integrable functions (via the Cauchy-Schwarz inequality).
Hence, the CS-D is symmetric, and invariant to the choice of hyper-volume
unit $U$. For Poissons with PHDs $v_{1}$ and $v_{2}$, the CS-D
simply reduces to \cite{HVVM2015Cauchy}
\[
D_{CS}\!\left(\pi_{1},\pi_{2}\right)=\frac{U}{2}\left\Vert v_{1}-v_{2}\right\Vert ^{2},
\]
where $\left\Vert f\right\Vert ^{2}\triangleq\left\langle f,f\right\rangle $
denotes the squared $L_{2}$-norm. This means the angle subtended
by $\pi_{1}$ and $\pi_{2}$ translates to the squared $L_{2}$-distance
between their PHDs. 

The versatility of the CS-D in multi-object system lies in the tractable
closed-form for the broader GLMB family, which includes LMBs and labeled
i.i.d. clusters. Specifically, for $\boldsymbol{\pi}_{i}=\{(w_{i}^{\left(\xi_{i}\right)\!}(I),p_{i}^{\left(\xi_{i}\right)})\!:\!(\xi_{i},I)\in\Xi_{i}\times\mathcal{F}(\mathbb{L})\}$,
$i=1,2$, the CS-D is given by \cite{BVVA2017Void}
\begin{equation}
D_{CS}\left(\boldsymbol{\pi}_{1},\boldsymbol{\pi}_{2}\right)=-\ln\frac{\left\langle \boldsymbol{\pi}_{1},\boldsymbol{\pi}_{2}\right\rangle _{U}}{\sqrt{\left\langle \boldsymbol{\pi}_{1},\boldsymbol{\pi}_{1}\right\rangle _{U}}\sqrt{\left\langle \boldsymbol{\pi}_{2},\boldsymbol{\pi}_{2}\right\rangle _{U}}},\label{e:GLMB_Cauchy-Schwarz-1}
\end{equation}
where
\begin{align*}
\left\langle \boldsymbol{\pi}_{i},\boldsymbol{\pi}_{j}\right\rangle _{U} & =\sum_{L\subseteq\mathbb{L}}\sum_{\substack{\xi_{i}\in\Xi_{i}}
}\sum_{\substack{\xi_{j}\in\Xi_{j}}
}\!w_{i}^{\left(\xi_{i}\right)}\!\left(L\right)\!w_{j}^{\left(\xi_{j}\right)}\!\left(L\right)\!\bigl\langle Up_{i}^{\left(\xi_{i}\right)}p_{j}^{\left(\xi_{j}\right)\!}\bigr\rangle^{\!L}\!\!.
\end{align*}
The above expression involves summing over all label sets $L\subseteq\mathbb{L}$,
though only those with non-zero $w_{i}^{\left(\xi_{i}\right)\!}\left(L\right)$,
$w_{j}^{\left(\xi_{j}\right)\!}\left(L\right)$ and $\bigl\langle Up_{i}^{\left(\xi_{i}\right)}p_{j}^{\left(\xi_{j}\right)\!}\bigr\rangle^{\!}(\ell)$,
for all $\ell\in L$, contribute to the sum. 

For the LMB special case, the CS-D (\ref{e:GLMB_Cauchy-Schwarz-1})
reduces to a sum of logarithms over $\ell\in\mathbb{L}$, similar
to the R\'{e}nyi-D. Specifically, it can be shown that (see Supplementary
Materials)
\begin{equation}
\!\!\!\!\!D_{CS}\!\left(\boldsymbol{\pi}_{1},\boldsymbol{\pi}_{2}\right)=-\!\sum_{\ell\in\mathbb{L}}\ln\!\frac{1\texttt{+}\left\langle U\!\boldsymbol{f}_{\!1}\boldsymbol{f}_{\!2}\right\rangle \!(\ell)}{\sqrt{1\texttt{+}\!\left\langle U\!\boldsymbol{f}_{\!1}^{2}\right\rangle \!(\ell)}\sqrt{1\texttt{+}\!\left\langle U\!\boldsymbol{f}_{\!2}^{2}\right\rangle \!(\ell)}},\!\!\label{eq:CS-D-LMB}
\end{equation}
though we only need to sum over the union $\mathcal{D}(\sigma_{1})\cup\mathcal{D}(\sigma_{2})$
of the LMB parameter domains, for the same reason as per the R\'{e}nyi.
Unlike R\'{e}nyi-D, KL-D, and $\chi^{2}$-D, the CS-D for LMBs bears little
resemblance to that for Poissons. Note that the CS-D for labeled i.i.d.
clusters is more expensive than LMBs, even though both are one-term
special cases of the GLMB. 

\subsection{Labeled RFS Approximations\protect\label{subsec:LRFS-approx}}

While approximations are indispensable for real-world applications,
what differentiates heuristics from principled engineering practice
is whether the approximation error can be characterized/quantified.
The GLMB is an analytic solution to the multi-object Bayes filter
that provides trajectory estimation, and is closed under truncation
with analytic error characterization \cite{VoVoP14}. Moreover, the
GLMBs also provide principled approximations to other LRFSs, thus,
hitting many birds with one stone. This subsection presents a number
of results for LRFS approximations based on GLMBs. 

\subsubsection{Truncation of GLMBs}

This is a crucial task in practice, where GLMBs consist of intractably
large sums. The key consideration of ensuring validity of the truncated
expression as a set-function, discussed in Subsection \ref{subsec:Tricky},
is automatically fulfilled (see Subsection \ref{subsec:GLMB}). Moreover,
the truncation error can be quantified analytically, similar to that
for Fourier series. 

For any $\mathbb{\mathbb{H}}\subseteq\Xi\times\mathcal{F}(\mathbb{L})$,
let $\boldsymbol{\pi}^{(\mathbb{\mathbb{H}})}$ denote the unnormalized
GLMB $\{(w^{\left(\xi\right)}(L),p^{\left(\xi\right)})\!:\!(\xi,L)\in\mathbb{\mathbb{H}}\}$
(i.e., does not necessarily integrate to 1). Then, the normalizing
constant is
\[
\int\boldsymbol{\pi}^{(\mathbb{\mathbb{H}})}(\boldsymbol{X})\delta\boldsymbol{X}=||\boldsymbol{\pi}^{(\mathbb{\mathbb{H}})}||_{1}=\sum\limits_{(\xi,L)\in\mathbb{H}}w^{(\xi)}(L),
\]
where $\left\Vert \boldsymbol{\mathbf{\mathit{f}}}\right\Vert _{1}\triangleq\int\left\vert \boldsymbol{\mathbf{\mathit{f}}}(\boldsymbol{X})\right\vert \delta\boldsymbol{X}$
denotes the $L_{1}$-norm of a function $\mathbf{\boldsymbol{\mathbf{\mathit{f}}}}$
on $\mathbf{\mathcal{F}(}\mathbb{X}\mathcal{\times}\mathbb{L})$.
Moreover, the $L_{1}$-error between $\boldsymbol{\pi}^{(\mathbb{\mathbb{H}})}$
and its truncated version $\boldsymbol{\pi}^{(\mathbb{T})}$ is given
by \cite{VoVoP14}
\begin{equation}
||\boldsymbol{\pi}^{(\mathbb{\mathbb{H}})}-\boldsymbol{\pi}^{(\mathbb{\mathbb{T}})}||_{1}\mathbf{=}\sum\limits_{(\xi,L)\in\mathbb{H-T}}w^{(\xi)}(L),\label{eq:L1-error}
\end{equation}
(it is implicit that $\mathbb{\mathbb{T}}\subseteq\mathbb{\mathbb{H}}$),
and for their normalized versions
\begin{equation}
\left\Vert \frac{\boldsymbol{\pi}^{(\mathbb{\mathbb{H}})}}{||\boldsymbol{\pi}^{(\mathbb{\mathbb{H}})}||_{1}}-\frac{\boldsymbol{\pi}^{(\mathbb{\mathbb{T}})}}{||\boldsymbol{\pi}^{(\mathbb{\mathbb{T}})}\mathbf{||}_{1}}\right\Vert _{1}\leq2\frac{||\boldsymbol{\pi}^{(\mathbb{\mathbb{H}})}||_{1}-||\boldsymbol{\pi}^{(\mathbb{\mathbb{T}})}\mathbf{||}_{1}}{||\boldsymbol{\pi}^{(\mathbb{\mathbb{H}})}||_{1}}.\label{eq:L1-error-bound}
\end{equation}

The above result means that the intuitive strategy of discarding basic
GLMB components with the smallest weights minimizes the $L_{1}$-error
between the actual and approximate multi-object densities. Similar
truncations are widely used in MOT algorithms such as MHT, and JPDA,
but without mathematical justifications. The LRFS approach characterizes
the effect of truncation by the $L_{1}$-error, and provides a mathematical
justification for truncating low-weighted components. 

\subsubsection{Label-Partitioned GLMB Approximation}

A `large' GLMB can be approximated, via label partitioning, as a product
of much `smaller' GLMBs that can be processed more efficiently in
parallel \cite{Beard18-largescale}. The rationale is that such approximations
incur minimal information loss when the smaller GLMBs are almost independent,
which is usually the case in practice since the objects are not uniformly
distributed across the state space, but often in separate groups.

Given a partition $\mathfrak{L}$ of the label space $\mathbb{L}$,
each $\boldsymbol{X}\in\mathfrak{\mathcal{F}}\left(\mathbb{X}\times\mathbb{L}\right)$
can be written as $\boldsymbol{X}=\biguplus_{L\in\mathfrak{L}}\boldsymbol{X}\cap\left(\mathbb{X}\times L\right)$,
and hence $\left\{ \mathfrak{\mathcal{F}}\left(\mathbb{X}\times L\right):L\in\mathfrak{L}\right\} $
also forms a partition of $\mathfrak{\mathcal{F}}\left(\mathbb{X}\times\mathbb{L}\right)$.
A labeled multi-object density on $\mathfrak{\mathcal{F}}(\mathbb{X}\times\mathbb{L})$
is said to be $\mathfrak{L}$\textit{-partitioned} if it can be written
as the product
\[
\boldsymbol{\pi}_{\mathfrak{L}}\left(\boldsymbol{X}\right)=\prod\limits_{L\in\mathfrak{L}}\boldsymbol{\pi}_{\mathfrak{L}}^{\left(L\right)}\left(\boldsymbol{X}\cap\left(\mathbb{X}\times L\right)\right),
\]
where each factor $\boldsymbol{\pi}_{\mathfrak{L}}^{\left(L\right)}$
is a labeled multi-object density on $\mathfrak{\mathcal{F}}(\mathbb{X}\times L)$
\cite{Beard18-largescale}. We denote $\boldsymbol{\pi}_{\mathfrak{L}}$,
by its factors $\{\boldsymbol{\pi}_{\mathfrak{L}}^{\left(L\right)}\}_{L\in\mathfrak{L}}$,
and if each factor $\boldsymbol{\pi}_{\mathfrak{L}}^{\left(L\right)}$
is a GLMB $\{(w_{\mathfrak{L,}L}^{\left(I,\xi\right)},p_{\mathfrak{L,}L}^{\left(\xi\right)})\}_{\left(I,\xi\right)\in\mathcal{F}(L)\times\Xi^{(L)}}$,
then $\boldsymbol{\pi}_{\mathfrak{L}}$ is said to be an $\mathfrak{L}$\textit{-partitioned}
\textit{GLMB}. 

An important numerical problem is to approximate a given $\mathfrak{L}$\textit{-}partitioned
GLMB $\boldsymbol{\pi}_{\mathfrak{L}}=\{\boldsymbol{\pi}_{\mathfrak{L}}^{\left(L\right)}\}_{L\in\mathfrak{L}}$
using another partition $\mathfrak{S}$ of $\mathbb{L}$. Indeed,
the $\mathfrak{S}$\textit{-}partitioned labeled multi-object density
$\boldsymbol{\pi}_{\mathfrak{S}}=\{\boldsymbol{\pi}_{\mathfrak{S}}^{\left(S\right)}\}_{S\in\mathfrak{S}}$
that minimizes $D_{KL}\left(\boldsymbol{\pi}_{\mathfrak{L}}||\boldsymbol{\pi}_{\mathfrak{S}}\right)$,
is an $\mathfrak{S}$\textit{-}partitioned GLMB, with GLMB factors
\cite{Beard18-largescale}
\[
\boldsymbol{\pi}_{\mathfrak{S}}^{\left(S\right)}(\boldsymbol{X}\cap(\mathbb{X}\times S))=\prod_{L\in\mathfrak{L}}\boldsymbol{\pi}_{\mathfrak{L},\mathfrak{S}}^{\left(L,S\right)}(\boldsymbol{X}\cap(\mathbb{X}\times S)),
\]
where, for each $(L,S)\in\mathfrak{L\times S}$ such that $L\cap S\neq\emptyset$
\begin{align*}
\boldsymbol{\pi}_{\mathfrak{L},\mathfrak{S}}^{\left(L,S\right)} & =\{(w_{L,S}^{\left(H,\xi\right)},p_{L,S}^{\left(\xi\right)})\}_{(H,\xi)\in\mathcal{F}(L\cap S)\times\Xi^{(L)}},\\
w_{L,S}^{\left(H,\xi\right)} & =\sum_{W\in\mathcal{F}\left(L\texttt{-}S\right)}w_{\mathfrak{L,}L}^{\left(H\cup W,\xi\right)},\\
p_{L,S}^{\left(\xi\right)}(x,\ell) & =\mathbf{1}_{L\cap S}(\ell)p_{\mathfrak{L,}L}^{\left(\xi\right)}(x,\ell).
\end{align*}
An efficient algorithm for finding the `best' partitions suitable
for large-scale multi-object estimation can be found in \cite{Beard18-largescale}. 

\subsubsection{Approximation by GLMB}

Estimation involving non-standard multi-object models invariably results
in intractable multi-object densities, see e.g., \cite{Beardetal14,PapiVoVoetal15,PapiKim-15,LiYiHoseinnezhadetal18,Bryantetal18,MahlerSuperPS18,Gostar-Interacting-19,Ong-TPAMI-20,Nguyenetal-CellSpawning21}.
Hence, it is important to make principled approximations by tractable
families such as GLMBs. Apart from statistical approximation techniques
such as moment matching and information minimization, approximating
the cardinality distribution is an important consideration in multi-object
estimation \cite{MahlerBook07}. 

An LRFS, characterized by a multi-object density $\mathbf{\boldsymbol{\pi}}$,
can be approximated by an LMB with matching 1st moment (PHD), by choosing
the parameters according to (\ref{eq:existence-prob}). Indeed, this
technique was first used in \cite{Reuter2014} to approximate GLMBs
with LMBs by selecting the parameters according to (\ref{eq:existence-prob-GLMB}),
(\ref{eq:track-density-GLMB}). However, LMBs are not versatile enough
to match the cardinality or capture the dependence between the attributes
of the multi-object state. To this end, approximation by GLMBs can
provide trade-offs between tractability and versatility.

A \textit{Marginalized-GLMB} (\textit{M-GLMB}) is a GLMB with density
of the form \cite{FVPV2015Marginalized,PapiVoVoetal15,Fantacci18}
\begin{equation}
\mathbf{\boldsymbol{\bar{\pi}}}(\boldsymbol{X})=\Delta(\boldsymbol{X})\sum_{L\subseteq\mathbb{L}}\bar{w}^{(L)}\delta_{L}[\mathcal{L}(\boldsymbol{X})]\left[\bar{p}^{(L)}\right]^{\boldsymbol{X}}.\label{eq:M-GLMB}
\end{equation}
Using a smaller number of components, an M-GLMB can approximate a
GLMB $\{(w^{(\xi)}(I),p^{(\xi)}):(\xi,I)\in\Xi\times\mathcal{F}(\mathbb{L})\}$
with matching PHD and cardinality, by choosing \cite{FVPV2015Marginalized}
\begin{align*}
\bar{w}^{(L)} & =\sum\limits_{\xi\in\Xi}w^{(\xi)}(L),\\
\bar{p}^{(L)}(\boldsymbol{x}) & =\frac{\mathbf{1}_{L}(\mathcal{L}(\boldsymbol{x}))}{\bar{w}^{(L)}}\sum\limits_{\xi\in\Xi}w^{(\xi)}(L)p^{(\xi)}(\boldsymbol{x}).
\end{align*}

For a labeled multi-object density $\mathbf{\boldsymbol{\pi}}$, the
M-GLMB that matches the PHD and cardinality whilst minimizing the
KL-D from $\mathbf{\boldsymbol{\pi}}$ has components given by \cite{PapiVoVoetal15}
\begin{align}
\bar{w}^{(L)} & =\left\langle \boldsymbol{\pi}\right\rangle \!(L),\label{eq:MarginalizeGeneral-1}\\
\bar{p}^{(L)}(\boldsymbol{x}) & =\frac{\mathbf{1}_{L}(\mathcal{L}(\boldsymbol{x}))}{\left\langle \boldsymbol{\pi}\right\rangle \!(L)}\left\langle \boldsymbol{\pi}(\{\boldsymbol{x}\}\uplus\cdot)\right\rangle \!(L-\{\mathcal{L}(\boldsymbol{x})\}).\label{eq:MarginalizeGeneral-2}
\end{align}
Thus, to match the PHD and cardinality of a labeled multi-object density
$\mathbf{\boldsymbol{\pi}}$ by an M-GLMB, we set each weight $\bar{w}^{(L)}$
to the joint existence probability $\left\langle \boldsymbol{\pi}\right\rangle \!(L)$,
and the $\bar{p}^{(\{\ell_{1:n}\})}(\cdot,\ell_{i})$'s to the marginals
of the label-conditioned joint attribute densities $\boldsymbol{\pi}(\{(\cdot,\ell_{1}),...,(\cdot,\ell_{n})\})/\!\left\langle \boldsymbol{\pi}\right\rangle \!(\{\ell_{1:n}\})$.
This approximation minimizes the KL-D in a similar way to the approximation
of a joint density by the product of its marginals \cite{C2003Dependence}.

In certain applications, e.g., \cite{Beardetal14,Bryantetal18,Gostar-Interacting-19,Ong-TPAMI-20,Nguyenetal-CellSpawning21},
the labeled multi-object density of interest takes on a multi-modal
form similar to a GLMB:
\begin{equation}
\boldsymbol{\pi}(\boldsymbol{X})=\Delta(\boldsymbol{X})\sum\limits_{\xi\in\Xi}w^{(\xi)}(\mathcal{L}(\boldsymbol{X}))p^{(\xi)}(\boldsymbol{X}),\label{eq:PropConj0-1-1}
\end{equation}
where each $\xi$ represents a mode, $\sum_{\xi\in\Xi}\sum_{L\subseteq\mathbb{L}}w^{\left(\xi\right)}\!\left(L\right)=1$,
and $\left\langle p^{(\xi)}\right\rangle \!(I)=1$. The M-GLMB approximation
cannot capture the modes and the associated information. Nonetheless,
the cardinality and PHD matching strategy can be applied to each of
the modes \cite{PapiVoVoetal15}. Specifically, a GLMB that matches
the PHD and cardinality whilst preserving the modes of (\ref{eq:PropConj0-1-1})
is given by $\boldsymbol{\hat{\pi}}=\{(\hat{w}^{(\xi,I)}\!,\hat{p}^{(\xi,I)})\!:\!(\xi,I)\in\Xi\!\times\!\mathcal{F}(\mathbb{L})\}$,
where 
\begin{flalign}
\hat{w}^{(\xi,I)} & =w^{(\xi)}(I),\label{eq:MarginalizeGeneral-1-1-1}\\
\hat{p}^{(\xi,I)}(\boldsymbol{x}) & =\mathbf{1}_{I}(\mathcal{L}(\boldsymbol{x}))\left\langle p^{(\xi)}(\{\boldsymbol{x}\}\uplus\cdot)\right\rangle \!(I-\{\mathcal{L}(\boldsymbol{x})\}).\label{eq:MarginalizeGeneral-2-1-1}
\end{flalign}
While this approximation requires many more components than the M-GLMB
approximation, intuitively, it incurs less information loss, by retaining
the information contained in the modes. However, there are no formal
results on the KL-D. 

\subsection{Spatio-Temporal Modeling and Multi-Scan GLMB\protect\label{subsec:Multi-scan-GLMB}}

\begin{figure}[t]
\begin{centering}
\resizebox{88mm}{!}{\includegraphics[clip]{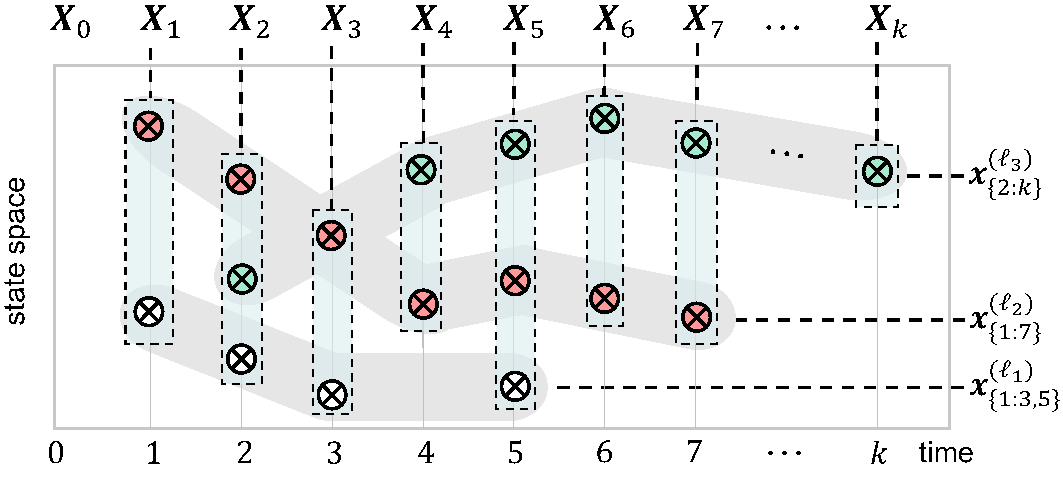}}
\par\end{centering}
\caption{A multi-object trajectory $\boldsymbol{X}_{0:k}$ on the interval
$\{0\!:\!k\}$. Individual trajectories are determined by grouping
the states of $\boldsymbol{X}_{0:k}$ according to labels (depicted
by the different colors). This representation covers trajectory fragmentation
(e.g., at time 4) and crossing (e.g., at time 3).}
\label{fig:multi-scan-model}
\end{figure}

So far, we have only discussed modeling of the multi-object state
via LRFS. This subsection extends the discussion to multi-object trajectory
modeling. In particular, we present an extension of the GLMB, known
as the multi-scan GLMB, as a tractable LRFS model of the multi-object
trajectory.

Recall that the multi-object trajectory on an interval $\{j{\displaystyle :}k\}$
is the sequence $\boldsymbol{X}_{j:k}$ of labeled multi-object states,
and that trajectories in $\boldsymbol{X}_{j:k}$ are determined by
grouping the states according to labels, see Fig. \ref{fig:multi-scan-model}.
More concisely, the trajectory of each (object with label) $\ell\in\mathcal{L}(\boldsymbol{X}_{j:k})$\footnote{Strictly speaking $\mathcal{L}(\boldsymbol{X}_{\!j:k})\!=\!(\mathcal{L}(\boldsymbol{X}_{\!j}),...,\mathcal{L}(\boldsymbol{X}_{\!k}))$
and we should write $\ell\in{\textstyle \cup_{i=j}^{k}}\mathcal{L}(\boldsymbol{X}_{i})$.
Nonetheless, the notation $\ell\in\mathcal{L}(\boldsymbol{X}_{j:k})$
is more compact, and also suggestive (that $\ell$ belongs to any
of the sets $\mathcal{L}(\boldsymbol{X}_{j})$, ..., $\mathcal{L}(\boldsymbol{X}_{k})$).} is the time-stamped sequence $\boldsymbol{x}_{T(\ell)}^{(\ell)}=[(x_{i},\ell)\in\boldsymbol{X}_{i}]_{i\in T(\ell)}$,
where $T(\ell)$ is the set of instants in $\{j\!:\!k\}$ such that
$\ell$ exists. This sequence, defines the mapping $\tau_{\ell}:i\mapsto x_{i}$,
$i\in T(\ell)$. It is clear that $\boldsymbol{X}_{j:k}$ can be reconstructed
from all the trajectories in $\boldsymbol{X}_{j:k}$, and hence 
is equivalently represented as the set of labeled trajectories 
\begin{equation}
\boldsymbol{X}_{j:k}\equiv\left\{ \boldsymbol{x}_{T(\ell)}^{(\ell)}:\ell\in\mathcal{L}(\boldsymbol{X}_{j:k})\right\} .\label{eq:multi_obj_state_seq-1}
\end{equation}

Since the multi-object trajectory on $\{j\textrm{:}k\}$ is represented
by a sequence of labeled multi-object states, it is naturally modeled
as a sequence $\boldsymbol{\Sigma}_{j:k}$ of LRFSs described by
a joint LRFS density $\mathbf{\mathbf{\boldsymbol{\pi}}}_{\boldsymbol{\Sigma}_{j:k}}$.
Beyond multi-object trajectory modeling, statistical characterization
of variables/parameters pertaining to the underlying multi-object
trajectory ensemble can be computed from the joint LRFS density \cite{VoVomultiscan18}.
Suppose that $\boldsymbol{f}(\boldsymbol{X}_{j:k})$ is the statistic
of a multi-object trajectory $\boldsymbol{X}_{j:k}$, then the ensemble
statistic is the expectation
\[
\mathbb{E}_{\boldsymbol{\Sigma}_{j:k}}\left[\boldsymbol{f}\right]=\int...\int\boldsymbol{f}(\boldsymbol{X}_{j:k})\mathbf{\mathbf{\boldsymbol{\pi}}}_{\boldsymbol{\Sigma}_{j:k}}(\boldsymbol{X}_{j:k})\delta\boldsymbol{X}_{j}...\delta\boldsymbol{X}_{k}.
\]
Fig. \ref{fig:ms-exp} shows some examples of multi-object trajectory
statistics. More examples can be found in \cite{VoVomultiscan18}.

\subsubsection{Multi-Scan Multi-Object Exponential}

The basic building block for a multi-scan GLMB is a multi-scan version
of the multi-object exponential. Noting that $h^{\boldsymbol{X}}$
is the product of the values of the function $h$ at every state in
$\boldsymbol{X}$, a natural multi-scan extension to $h^{\boldsymbol{X}_{j:k}}$
is the product of the values of $h$ at every trajectory in $\boldsymbol{X}_{j:k}$.
More concisely, the \textit{multi-scan multi-object exponential} is
defined as \cite{VoVomultiscan18}
\begin{eqnarray}
 & \!\!\!h^{\boldsymbol{X}_{j:k}}\triangleq h^{\{\boldsymbol{x}_{T(\ell)}^{(\ell)}:\ell\in\mathcal{L}(\boldsymbol{X}_{j:k})\}}={\displaystyle {\textstyle \prod}_{\ell\in\mathcal{L}(\boldsymbol{X}_{j:k})}}h(\boldsymbol{x}_{T(\ell)}^{(\ell)}),\label{eq:multi-scan-exp}
\end{eqnarray}
for any real function $h$ on ${\scriptstyle {\displaystyle {\textstyle \biguplus}}}_{I\subseteq\{j:k\}}\mathbb{T}_{I}$,
where 
\begin{eqnarray*}
\mathbb{T}_{\{i_{1},i_{2},...,i_{n}\}}\triangleq{\textstyle \cprod}{}_{j=1}^{n}\big(\mathbb{X}\times\mathbb{L}_{i_{j}}\big), & {\displaystyle \!\!\!{\textstyle \cprod}{}_{i=j}^{k}}S_{i}\triangleq S_{j\!}\times\!...\!\times\!S_{k}, & {\textstyle }
\end{eqnarray*}
for $i_{1}\!<\!i_{2}\!<\!...\!<\!i_{n}\!\in\!\{j{\textstyle :}k\}$.
If $\boldsymbol{x}_{T(\ell)}^{(\ell)}=[(x_{i},\ell)]_{i\in T(\ell)}$,
we write $h(\boldsymbol{x}_{T(\ell)}^{(\ell)})$ as $h(x_{T(\ell)};\ell)$,
and if the trajectory is unfragmented, i.e., $T(\ell)=\{s(\ell)\!:\!t(\ell)\}$,
where $s(\ell)\triangleq\min T(\ell)$, $t(\ell)\triangleq\max T(\ell)$,
we write $h(\boldsymbol{x}_{T(\ell)}^{(\ell)})$ as $h(x_{s(\ell):t(\ell)};\ell)$.

The multi-scan multi-object exponential (\ref{eq:multi-scan-exp})
satisfies the exponential-like property: 
\begin{flalign*}
\left[gh\right]^{\boldsymbol{X}_{j:k}}= & \left[g\right]^{\boldsymbol{X}_{j:k}}\left[h\right]^{\boldsymbol{X}_{j:k}},
\end{flalign*}
where $g$ is another function on ${\scriptstyle \biguplus}_{I\subseteq\{j:k\}}\mathbb{T}_{I}$,
see \cite{VoVomultiscan18} for additional properties. When $j=k$,
(\ref{eq:multi-scan-exp}) reduces to the single-scan multi-object
exponential $h^{\boldsymbol{X}_{j}}$.  
\begin{figure}
\begin{centering}
\includegraphics[width=1\columnwidth]{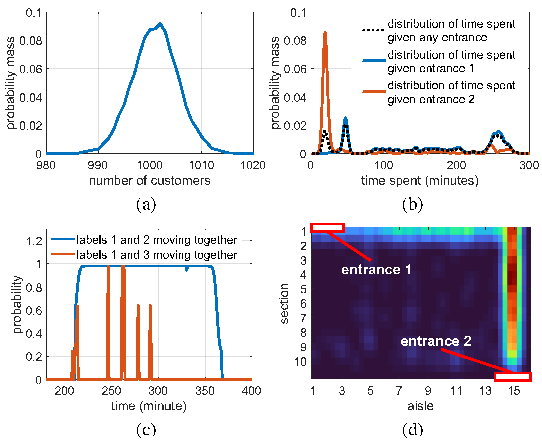}
\par\end{centering}
\caption{Multi-object trajectory statistics in a retail store with 15 aisles
and 2 entrances. Computed from a hypothetical joint LRFS density of
a 1-day scenario: (a) distribution of customer numbers in the store;
(b) distribution of the time customers spend in the store given the
entrance they use; (c) probability that certain customers move together
(within 2m of each other) over time; (d) location intensity of customers
who spend less than 60 minutes in the store.\protect\label{fig:ms-exp}\textcolor{teal}{}}
\end{figure}

\subsubsection{Multi-Scan GLMB}

This model was proposed in \cite{VoVomultiscan18} for smoothing under
the standard multi-object system model that only permits unfragmented
trajectories. A \emph{multi-scan GLMB} on the interval $\{j\textrm{:}k\}$
is a sequence $\boldsymbol{\Sigma}_{j:k}$ of LRFSs described by a
joint multi-object density on $\cprod{}_{i=j}^{k}\mathcal{F}(\mathbb{X\times L}_{i})$
of the form
\begin{equation}
\mathbf{\mathbf{\boldsymbol{\pi}}}_{\boldsymbol{\Sigma}_{j:k}}(\boldsymbol{X}_{j:k})=\Delta(\boldsymbol{X}_{j:k})\sum_{\xi\in\Xi}w^{(\xi)}(\mathcal{L}(\boldsymbol{X}_{j:k}))[p^{(\xi)}]^{\boldsymbol{X}_{j:k}},\label{eq:MSGLMB}
\end{equation}
where: $\Delta(\boldsymbol{X}_{\!j:k})\!\triangleq\!\prod_{i=j\!}^{k}\Delta(\boldsymbol{X}_{\!i})$;
$w^{(\xi)}(I_{j:k})$ is non-negative such that $\sum_{\xi,I_{j:k}}w^{(\xi)}(I_{j:k})=1$
(it is understood that the sum is taken over $\xi\in\Xi$ and $I_{j:k}\in\cprod_{i=j}^{k}\mathcal{F}(\mathbb{L}_{i})$);
$\mathcal{L}(\boldsymbol{X}_{\!j:k})\!=\!(\mathcal{L}(\boldsymbol{X}_{\!j}),...,\mathcal{L}(\boldsymbol{X}_{\!k}))$;
and $p^{(\xi)}(x_{s(\ell):t(\ell)};\ell)$ is a joint density of the
attribute sequence $x_{s(\ell):t(\ell)}$, for each $\ell\in I_{j:k}$,
with $s(\ell)$, $t(\ell)$ implicitly depend on $\left(\xi,I_{j:k}\right)$.
The joint density (\ref{eq:MSGLMB}) indeed integrates to 1 \cite{VoVomultiscan18}.
Similar to the GLMB, the multi-scan GLMB (\ref{eq:MSGLMB}) can also
be written in delta form
\[
\mathbf{\mathbf{\boldsymbol{\pi}}}_{\boldsymbol{\Sigma}_{j:k\!}}(\boldsymbol{X}_{\!j:k})=\Delta(\boldsymbol{X}_{\!j:k})\!\sum_{\xi,I_{\!j:k}}\!w^{(\xi,I_{\!j:k})}\delta_{I_{\!j:k\!}}[\mathcal{L}(\boldsymbol{X}_{\!j:k})][p^{(\xi)\!}]^{\boldsymbol{X}_{\!j:k}},
\]
and denoted by $\mathbf{\mathbf{\boldsymbol{\pi}}}_{\boldsymbol{\Sigma}_{j:k}}=\{(w^{(\xi)}(I_{j:k}),p^{\left(\xi\right)}){\textstyle :}(\xi,I_{j:k})\}$,
where it is understood that $\xi\in\Xi$ and $I_{j:k}\in\cprod_{i=j}^{k}\mathcal{F}(\mathbb{L}_{i})$.

Conceptually, a multi-scan GLMB can be regarded as a GLMB with the
set of labeled states replaced by the set of labeled trajectories
(\ref{eq:multi_obj_state_seq-1}). Analogous to the GLMB, $w^{(\xi)}(I_{j:k})$
is the probability of hypothesis $\left(\xi,I_{j:k}\right)$ representing
the `event' $\xi$ and the joint existence of the trajectories according
to the sequence $I_{j:k}$ of label sets, and conditional on hypothesis
$\left(\xi,I_{j:k}\right)$, $p^{(\xi)}(\cdot,\ell)$ is the attribute
density of trajectory $\ell\in I_{j:k}$. The multi-scan GLMB is closed
under the Bayes posterior recursion as well as truncation. Indeed,
the error expressions (\ref{eq:L1-error}), (\ref{eq:L1-error-bound})
also hold for multi-scan GLMBs \cite{VoVomultiscan18}. 

Some closed-form multi-object statistics from the multi-scan GLMB
are given as follows (see \cite{VoVomultiscan18} for further details).
 
\begin{itemize}
\item Trajectory cardinality distribution over interval $\{j\textrm{:}k\}$:
\begin{equation}
\!\!\Pr\!\left(|\boldsymbol{X}_{j:k}|=n\right)=\sum_{\xi,I_{j:k}}\delta_{\left\vert {\scriptstyle {\textstyle \cup}_{i=j}^{k}}I_{i}\right\vert }\!\left[n\right]w^{(\xi)\!}(I_{j:k}).\label{eq:cardinality}
\end{equation}
\item Joint existence probability of only the labels in $L$ (and no other
labels exist) over interval $\{j\textrm{:}k\}$:
\begin{equation}
\!\!\Pr\!\left(\textrm{only }L\text{ exist}\right)=\sum_{\xi,I_{j:k}}\delta_{{\scriptstyle {\scriptstyle {\textstyle \cup}_{i=j}^{k}}}I_{i}}[L]w^{(\xi)}(I_{j:k}).\label{eq:jointexistence}
\end{equation}
\item Joint existence probability of the labels in $L$ regardless of other
labels, over interval $\{j\textrm{:}k\}$:
\begin{equation}
\!\!\Pr\!\left(L\text{ exist}\right)=\sum_{\xi,I_{j:k}}\mathbf{1}_{\mathcal{F}\left({\scriptstyle {\textstyle {\scriptstyle {\textstyle \cup}_{i=j}^{k}}}}I_{i}\right)}(L)w^{(\xi)}(I_{j:k}).\label{eq:jointexistence-2}
\end{equation}
\item Trajectory length distribution, i.e., the probability that a trajectory
has length $n$, over interval $\{j\textrm{:}k\}$:
\begin{equation}
\!\!\Lambda\!\left(n\right)=\sum_{\xi,I_{j:k}}\!\frac{\sum_{\ell\in{\scriptstyle {\textstyle \cup}_{i=j}^{k}}I_{i}}\delta_{t(\ell)\texttt{-}s(\ell)\texttt{+}1}\!\left[n\right]}{\left\vert {\textstyle \cup_{i=j}^{k}}I_{i}\right\vert }w^{(\xi)\!}(I_{j:k}).\label{eq:lengthdistribution}
\end{equation}
\item Length distribution of trajectory $\ell$, i.e., the probability that
trajectory $\ell$ has length $n$, over interval $\{j\textrm{:}k\}$:
\begin{equation}
\!\!\Lambda_{\ell\!}\left(n\right)=\sum_{\xi,I_{j:k}}\mathbf{1}_{{\scriptstyle {\textstyle \cup}_{i=j}^{k}}I_{i}\!}(\ell)\delta_{t(\ell)\texttt{-}s(\ell)\texttt{+}1}\!\left[n\right]w^{(\xi)\!}(I_{j:k}).\label{eq:lengthprob}
\end{equation}
\end{itemize}

\subsubsection{Multi-Scan GLMB Estimator}

Various GLMB estimators can be extended to multi-scan GLMB. The simplest
would be to find the most probable hypothesis $(\xi^{*},I_{j:k}^{*})$
(with highest weight $w^{(\xi^{*})}(I_{j:k}^{*})$) and compute the
most probable or expected trajectory estimate from $p^{(\xi^{*})}(\cdot;\ell)$
for each $\ell$ in $I_{j:k}^{*}$. Alternatively, instead of the
most significant, we can use the most significant among components
with the most probable trajectory cardinality $n^{\ast}$ determined
by maximizing (\ref{eq:cardinality}).

The label-MaM estimate can also be adapted for multi-scan GLMB, using
the most probable label set sequence $I_{j:k}^{*}$, determined by
maximizing $\sum_{\xi}w^{(\xi)\!}(I_{j:k})$, and computing the trajectory
density for each $\ell$ in $I_{j:k}^{*}$ 
\[
p(\cdot;\ell)\propto\sum_{\xi}w^{(\xi)}(I_{j:k}^{*})p^{(\xi)}(\cdot;\ell),
\]
from which the mode or mean trajectory can then be determined. One
variation is to choose $I_{j:k}^{*}$ as the most probable label set
sequence with trajectory cardinality $n^{\ast}$. 

Another adaptation of the label-MaM estimate is based on the most
probable set of labels $L^{\ast}$, determined by maximizing the joint
existence probability (\ref{eq:jointexistence}). However, for each
$\ell$ $\in L^{\ast}$, its trajectory length for different hypotheses
$(\xi,I_{j:k})$ (that yield the same $L^{\ast}$) may not be the
same, and hence there is no meaningful most probable trajectory (due
to different units of measurements in the probability densities of
the trajectories). Nonetheless, we can determine the most probable
length $m^{\ast}$ for each $\ell\in L^{\ast}$, by maximizing (\ref{eq:lengthprob}),
and estimate the trajectory according to
\[
p_{m^{\ast}}(\cdot;\ell)\propto\!\!\sum_{\xi,I_{j:k}}\!\!\mathbf{1}_{{\scriptstyle \bigcup_{i=j}^{k}}I_{i}\!}(\ell)w^{(\xi)}\!(I_{j:k})\delta_{t(\ell)\text{-}s(\ell)\text{+}1}\!\left[m^{\ast}\right]p^{(\xi)}(\cdot;\ell).
\]
Variations of this estimator are to use the label set $L^{\ast}$
of cardinality $n^{\ast}$ with highest joint existence probability,
or the set of $n^{\ast}$ labels with highest individual existence
probabilities.  

\section{Dynamic Multi-Object Estimation\protect\label{sec:Multi-object-Estimation}}

\begin{table}[!t]
\global\long\def\arraystretch{1.3}%
 \caption{Common notations from Section \ref{sec:Multi-object-Estimation}.}

\vspace*{0.4cm}
 {\footnotesize\label{tbl:notation-dynamic}}{\footnotesize\par}
\centering{}%
\begin{tabular}{|c|l|}
\hline 
{\small\textbf{Notation}} & {\small\textbf{Description}}\tabularnewline
\hline 
{\small$\mathbb{B}_{k}$} & {\small Space of labels (of objects) born at time $k$}\tabularnewline
{\small$\mathbb{L}_{k}$} & {\small$\biguplus_{t=0}^{k}\mathbb{B}_{t}$, label space at time $k$ }\tabularnewline
{\small$\boldsymbol{X}_{k}$} & {\small Multi-object state at time $k$}\tabularnewline
{\small$Z_{k}$} & {\small Multi-object measurement at time $k$}\tabularnewline
{\small$\boldsymbol{f}_{k}(\cdot|\boldsymbol{X}_{k\texttt{-}1})$} & {\small Multi-object transition density given $\boldsymbol{X}_{k\texttt{-}1}$ }\tabularnewline
{\small$\boldsymbol{g}_{k}(Z_{k}|\boldsymbol{X}_{k})$} & {\small Likelihood of observing $Z_{k}$ given $\boldsymbol{X}_{k}$}\tabularnewline
{\small$\mathbf{\boldsymbol{\pi}}_{0:k}(\boldsymbol{X}_{0:k})$} & {\small Multi-object posterior density at $\boldsymbol{X}_{0:k}$}\tabularnewline
{\small$\mathbf{\boldsymbol{\pi}}_{k}(\boldsymbol{X}_{k})$} & {\small Multi-object filtering density at $\boldsymbol{X}_{k}$}\tabularnewline
{\small$\boldsymbol{f}_{\!B,k}$} & {\small Density of LRFS of new objects at time $k$}\tabularnewline
{\small$P_{B,k}(\ell)$} & {\small Birth probability of label} {\small$\ell$ at time $k$}\tabularnewline
{\small$p_{B}(\cdot,\ell)$} & {\small Attribute density for newborn label $\ell$ at time $k$ }\tabularnewline
{\small$\boldsymbol{f}_{\!S,k\!}\left(\cdot|\boldsymbol{X}_{k\texttt{-}1}\right)$} & {\small Density of LRFS at time $k$ generated by $\boldsymbol{X}_{k\texttt{-}1}$ }\tabularnewline
{\small$P_{S,k}(x_{k\texttt{-}1},\ell)$} & {\small Survival probability to time $k$ of state $(x_{k\texttt{-}1},\ell)$ }\tabularnewline
{\small$f_{S,k}(\cdot|\cdot,\ell)$} & {\small Attribute transition density to time $k$ for $\ell$ }\tabularnewline
{\small$P_{D,k}(\boldsymbol{x})$} & {\small Detection probability of state $\boldsymbol{x}$ at time $k$}\tabularnewline
{\small$\kappa_{k}$} & {\small Clutter intensity function at time $k$}\tabularnewline
{\small$\varPsi_{Z,k}^{(j)}(\boldsymbol{x})$} & {\small$\boldsymbol{1}_{\{1\text{:}|Z|\}}(j)\frac{P_{D,k}(\boldsymbol{x})g_{k}(z_{j}|\boldsymbol{x})}{\kappa_{k}(z_{j})}\texttt{+}\delta_{0}[j](1\texttt{-}P_{D,k}(\boldsymbol{x}))$}\tabularnewline
{\small$\Lambda_{S,k}^{\!(j)\!}(x|\varsigma,\ell)$} & {\small$\varPsi_{Z,k}^{(j)}(x,\ell)f_{S,k}(x|\varsigma,\ell)P_{S,k}(\varsigma,\ell)$}\tabularnewline
{\small$\Lambda_{B,k}^{\!(j)\!}(x,\ell)$} & {\small$\varPsi_{Z,k}^{(j)}(x,\ell)p_{B,k}(x,\ell)P_{B,k}(\ell)$}\tabularnewline
{\small$\theta,\theta_{k}$} & {\small Association map $\mathbb{L}_{k}\rightarrow\{0:|Z_{k}|\}$ at
time $k$}\tabularnewline
{\small$\Theta(I),\Theta_{k}(I)$} & {\small Space of association maps $\theta_{k}$ with domain $I$}\tabularnewline
{\small$(\xi,I)$} & {\small Prior GLMB component}\tabularnewline
{\small$(\xi,\theta,I)$} & {\small Updated GLMB component}\tabularnewline
{\small$\gamma,\gamma_{k}$} & {\small Extended association $I_{k\texttt{-}1}\uplus\mathbb{B}_{k}\!\rightarrow\!\{-1\!:\!\left|Z_{k}\right|\}$ }\tabularnewline
{\small$\Gamma,\Gamma_{k}$} & {\small Space of extended associations at time $k$}\tabularnewline
{\small$P,P_{k}$} & {\small Number of (hypothesized) labels at time $k$ }\tabularnewline
{\small$M,M_{k}$} & {\small Number of measurements at time $k$ }\tabularnewline
\hline 
\end{tabular}
\end{table}

The apparatus developed in the previous sections enables classical
SSM concepts to be applied to multi-object state estimation, as shown
in this section. Bayesian estimation  for multi-object SSM is formulated
in Subsection \ref{ss:Multi-object-SSM}, while Subsection \ref{ss:Standard-Multi-object-SSM}
presents the standard multi-object SSM. An exact solution to the
Bayes multi-object filter using GLMB, and an approximate solution
using LMB, are presented in Subsections \ref{ss:GLMB-filter} and
\ref{ss:LMB-filter}, while numerical implementations are discussed
in Subsection \ref{ss:GLMB-Imp}. Extension of the GLMB filter to
multi-object smoothing is presented in Subsection \ref{ss:Multi-object_Smoother},
while Subsection \ref{ss:Multi-object_Smoother-1} discusses its implementation.
Extensions to non-standard models, robust and distributed estimation,
and control for multi-object system are discussed in Subsections \ref{ss:Non-Standard-model},
\ref{ss:Robust-Est}, \ref{ss:Distributed-Est}, and \ref{ss:Control}.

\subsection{Multi-Object State Space Model\protect\label{ss:Multi-object-SSM}}

Recall that a label $\ell=(s,\iota)$ consists of the starting time
$s$ and an index $\iota$, let $\mathbb{B}_{k}$ denote the space
of labels with starting time $k$. Then, the label space and labeled
state space at time $k$ are, respectively, the disjoint union $\mathbb{L}_{k}=\biguplus_{t=0}^{k}\mathbb{B}_{t}$
and $\mathbb{X}\times\mathbb{L}_{k}$. 

Analogous to traditional SSMs, the multi-object state $\boldsymbol{X}_{k}\in\mathcal{F}(\mathbb{X}\times\mathbb{L}_{k})$,
at time $k$, evolves from its previous value $\boldsymbol{X}_{k\texttt{-}1}\in\mathcal{F}(\mathbb{X}\!\times\!\mathbb{L}_{k\texttt{-}1})$,
and generates an observation $Z_{k}\in\mathcal{F}(\mathcal{Z})$,
according to the multi-object state and observation equations
\begin{flalign}
\boldsymbol{X}_{k} & =\boldsymbol{S}_{k}(\boldsymbol{X}_{k\texttt{-}1})\cup\boldsymbol{B}_{k},\label{eq:RFStransition}\\
Z_{k} & =D_{k}(\boldsymbol{X}_{k})\cup K_{k},\label{eq:RFSmeasurement}
\end{flalign}
where $\boldsymbol{S}_{k}(\boldsymbol{X}_{k\texttt{-}1})$ is the
set of states generated from $\boldsymbol{X}_{k\texttt{-}1}$, $\boldsymbol{B}_{k}$
is the set of newly appearing states, $D_{k}(\boldsymbol{X}_{k})$
is the set of detections generated from $\boldsymbol{X}_{k}$, and
$K_{k}$ is the set of clutter, (or false alarms). 

Under the Bayesian paradigm, the multi-object state and observation
are modeled as RFSs. Using the FISST notion of multi-object density,
the multi-object state and observation equations (\ref{eq:RFStransition}),
(\ref{eq:RFSmeasurement}) can be characterized, respectively, by
the \emph{multi-object} (\emph{Markov})\emph{ transition density}
and\emph{ multi-object} (\emph{observation})\emph{ likelihood function}
\cite{Goodmanetal97,Mahler2003}
\begin{eqnarray*}
\boldsymbol{f}_{k}(\boldsymbol{X}_{k}|\boldsymbol{X}_{k\texttt{-}1}), &  & \boldsymbol{g}_{k}(Z_{k}|\boldsymbol{X}_{k}).
\end{eqnarray*}
The transition density $\boldsymbol{f}_{k}(\cdot|\cdot)$ captures
the underlying evolution, appearance and disappearance of the objects.
The observation likelihood $\boldsymbol{g}_{k}(\cdot|\cdot)$ captures
the underlying detections, false negatives/positives and data association
uncertainty. Examples of $\boldsymbol{f}_{k}(\cdot|\cdot)$ and $\boldsymbol{g}_{k}(\cdot|\cdot)$
are given in the standard (commonly used) SSM presented in Subsection
\ref{ss:Standard-Multi-object-SSM}.

Given the observation history $Z_{1:k}$, all information on the multi-object
trajectory is captured in the \emph{multi-object posterior density
}$\mathbf{\boldsymbol{\pi}}_{0:k}(\boldsymbol{X}_{0:k})\triangleq\mathbf{\boldsymbol{\pi}}_{0:k}(\boldsymbol{X}_{0:k}|Z_{1:k})$
(the dependence on $Z_{1:k}$ is omitted for brevity). Similar to
Bayesian estimation for traditional SSMs \cite{Briers-10}, \cite{DoucetTutorial09},
the (multi-object) posterior density can be propagated forward recursively
by \cite{VoVomultiscan18},
\begin{eqnarray}
\!\!\!\!\!\!\!\!\!\mathbf{\boldsymbol{\pi}}_{0:k}(\boldsymbol{X}_{0:k})\!\!\! & =\label{eq:MO_posterior}\\
 &  & \!\!\!\!\!\!\!\!\!\!\!\!\!\!\!\!\!\!\!\!\!\!\!\!\!\!\!\!\!\!\!\!\frac{\boldsymbol{g}_{k}(Z_{k}|\boldsymbol{X}_{k})\boldsymbol{f}_{\!k}\!\left(\boldsymbol{X}_{k}|\boldsymbol{X}_{k\texttt{-}1}\right)\mathbf{\boldsymbol{\pi}}_{0:k\texttt{-}1}(\boldsymbol{X}_{0:k\texttt{-}1})}{\int\!\boldsymbol{g}_{k}(Z_{k}|\boldsymbol{Y}_{k})\boldsymbol{f}_{\!k}\!\left(\boldsymbol{X}_{k}|\boldsymbol{Y}_{k\texttt{-}1}\right)\boldsymbol{\mathbf{\pi}}_{0:k\texttt{-}1}(\boldsymbol{Y}_{0:k\texttt{-}1})\delta\boldsymbol{Y}_{0:k}}.\nonumber 
\end{eqnarray}
Since the dimension of $\boldsymbol{\pi}_{0:k}$ increases with $k$,
the computational complexity for each iterate of (\ref{eq:MO_posterior}),
which is a function of the dimension, grows with time. This growth
is far worse than its single-object counterpart because the multi-object
state space $\mathcal{F}(\mathbb{X}\times\mathbb{L}_{k})$ is far
larger than the state space $\mathbb{X}$. 

A cheaper alternative is the \emph{multi-object filtering density},
$\mathbf{\boldsymbol{\pi}}_{k}(\boldsymbol{X}_{k})\triangleq\int\mathbf{\boldsymbol{\pi}}_{0:k}(\boldsymbol{X}_{0:k})\delta\boldsymbol{X}_{0:k\texttt{-}1}$,
which can be propagated by the \textit{multi-object} \textit{Bayes
filter} \cite{Mahler2003,MahlerBook14}
\begin{align}
\!\mathbf{\boldsymbol{\pi}}_{k}(\boldsymbol{X}_{k}) & =\frac{\boldsymbol{g}_{k}(Z_{k}|\boldsymbol{X}_{k})\!\int\!\!\boldsymbol{f}_{\!k}\!\left(\boldsymbol{X}_{k}|\boldsymbol{Y}\right)\!\mathbf{\boldsymbol{\pi}}_{k\texttt{-}1\!}(\boldsymbol{Y})\delta\boldsymbol{Y}}{\int\!\boldsymbol{g}_{k}(Z_{k}|\boldsymbol{X})\!\int\!\!\boldsymbol{f}_{\!k}\!\left(\boldsymbol{X}|\boldsymbol{Y}\right)\!\mathbf{\boldsymbol{\pi}}_{k\texttt{-}1\!}(\boldsymbol{Y})\delta\boldsymbol{Y}\delta\boldsymbol{X}}.\label{eq:MO-bayes-2}
\end{align}
Similar to the single-object Bayes filter, the computational complexity
for each iterate of (\ref{eq:MO-bayes-2}) does not increase with
time since the dimension of $\mathbf{\boldsymbol{\pi}}_{k}$ does
not grow with $k$. 

Under the standard multi-object model (to be discussed next), numerical
solutions to the posterior and filtering recursions were first developed
at around the same time, respectively, in \cite{VVE2011Particle,Vu2014}
using particle marginal Metropolis-Hasting simulation, and in \cite{VoConj11,VoConj13}
using GLMBs. The latter is an analytic solution, and was later extended
to solve the posterior recursion in \cite{VV2018Multiscan,VoVomultiscan18}
(see Subsections \ref{ss:GLMB-filter}, and \ref{ss:Multi-object_Smoother}).

\subsection{Standard Multi-Object Model\protect\label{ss:Standard-Multi-object-SSM}}

\subsubsection{Multi-Object State Dynamic\protect\label{sss:Multi-object_Dynamics}}

Let $\boldsymbol{f}_{\!B,k}\!$ denote the density of the LRFS $\boldsymbol{B}_{k}$
of new objects, and $\boldsymbol{f}_{\!S,k\!}\left(\cdot|\boldsymbol{X}_{\!k\texttt{-}1}\right)$
the density of the LRFS $\boldsymbol{S}_{k}(\boldsymbol{X}_{\!k\texttt{-}1})$
of objects generated from the previous multi-object state $\boldsymbol{X}_{\!k\texttt{-}1}$.
Assuming $\boldsymbol{S}_{k}(\boldsymbol{X}_{\!k\texttt{-}1})$ and
$\boldsymbol{B}_{k}$ are independent, due to the labeling construct,
the multi-object transition density is given by \cite{VoConj13}
\begin{eqnarray}
\boldsymbol{f}_{\!k}\!\left(\boldsymbol{X}_{\!k}|\boldsymbol{X}_{\!k\texttt{-}1}\right)\!\!\! & =\\
 &  & \!\!\!\!\!\!\!\!\!\!\!\!\!\!\!\!\!\!\!\!\!\!\!\!\!\!\!\!\!\!\!\!\boldsymbol{f}_{\!B,k}\!\left(\boldsymbol{X}_{\!k}\cap(\mathbb{X}\times\mathbb{B}_{k})\right)\boldsymbol{f}_{\!S,k}\!\left(\boldsymbol{X}_{\!k}\cap(\mathbb{X}\times\mathbb{L}_{k\texttt{-}1})|\boldsymbol{X}_{\!k\texttt{-}1}\right)\!.\nonumber
\end{eqnarray}
%\[
%\boldsymbol{f}_{\!k}\!\left(\boldsymbol{X}_{\!k}|\boldsymbol{X}_{\!k\texttt{-}1}\right)=\boldsymbol{f}_{\!B,k}\!\left(\boldsymbol{X}_{\!k}\!-\!\boldsymbol{X}_{\!k\texttt{-}1}\right)\boldsymbol{f}_{\!S,k}\!\left(\boldsymbol{X}_{\!k}\cap\boldsymbol{X}_{\!k\texttt{-}1}|\boldsymbol{X}_{\!k\texttt{-}1}\right)\!.
%\]

In the \textit{standard} (commonly used)\textit{ multi-object dynamic
model}, $\boldsymbol{B}_{k}$ is a GLMB on $\mathcal{F}(\mathbb{X}\times\mathbb{B}_{k})$.
For simplicity (but without loss of generality), we use the one-term
GLMB
\begin{equation}
\boldsymbol{f}_{\!B,k}(\boldsymbol{X})=\Delta(\boldsymbol{X})w_{B,k}(\mathcal{L}(\boldsymbol{X}))p_{B,k}^{\boldsymbol{X}}.\label{eq:Birth_transition-1}
\end{equation}
Note that $\boldsymbol{f}_{\!B,k}(\boldsymbol{X})=0$ if $\boldsymbol{X}$
contains any element with label outside of $\mathbb{B}_{k}$. This
birth model is general enough to include labeled i.i.d cluster and
LMB, though the latter, given by $\boldsymbol{f}_{\!B,k}\!=\!\{(P_{B,k}(\ell),p_{B,k}(\cdot,\ell))\}_{\ell\in\mathbb{B}_{k}}$,
is most popular.

Further, for a given $\boldsymbol{X}_{\!k\texttt{-}1}$, each $(x_{k\texttt{-}1},\ell)\in\boldsymbol{X}_{\!k\texttt{-}1}$
either survives with probability $P_{S,k}(x_{k\texttt{-}1},\ell)$
and evolves to state $(x_{k},\ell)$ at time $k$, with the same label
and attribute transition density $f_{S,k}(x_{k}|x_{k\texttt{-}1},\ell)$,
or dies with probability $1-P_{S,k}(\boldsymbol{x}_{k\texttt{-}1})$.
This means objects keep the same labels for their entire lives. Assuming
that conditional on $\boldsymbol{X}_{\!k\texttt{-}1}$ each object
survives and evolves independently of one another, $\boldsymbol{S}_{k}(\boldsymbol{X}_{\!k\texttt{-}1})$
is an LMB on $\mathcal{F}(\mathbb{X}\times\mathbb{L}_{k\texttt{-}1})$,
with parameters $\{(P_{S,k}(\boldsymbol{\zeta}),f_{S,k}(\cdot|\boldsymbol{\zeta}))\}_{\boldsymbol{\zeta}\in\boldsymbol{X}_{k\texttt{-}1}}$,
and density given by \cite{VoConj13}
\begin{align}
\boldsymbol{f}_{\!S,k\!}\left(\boldsymbol{X}|\boldsymbol{X}_{\!k\texttt{-}1}\right) & =\Delta\!\left(\boldsymbol{X}\right)\mathbf{1}_{\mathcal{L}\left(\boldsymbol{X}_{\!k\texttt{-}1}\right)}^{\mathcal{L}\left(\boldsymbol{X}\right)}\left[\boldsymbol{f}_{\!S,k\!}\left(\boldsymbol{X};\cdot\right)\right]^{\boldsymbol{X}_{\!k\texttt{-}1}},\label{eq:Survival-trans}
\end{align}
where $\mathbf{1}_{\mathcal{L}\left(\boldsymbol{X}_{\!k\texttt{-}1}\right)}^{\mathcal{L}\left(\boldsymbol{X}\right)}=\prod_{\ell\in\mathcal{L}\left(\boldsymbol{X}\right)}\mathbf{1}_{\mathcal{L}\left(\boldsymbol{X}_{\!k\texttt{-}1}\right)}(\ell)$,
and
\begin{equation}
\!\!\!\boldsymbol{f}_{\!S,k\!}\left(\boldsymbol{X};y,\ell\right)=\begin{cases}
P_{S,k\!}\left(y,\ell\right)\!f_{S,k}(x|y,\ell), & \!\!\!\text{if }(x,\ell)\in\boldsymbol{X}\\
1-P_{S,k\!}\left(y,\ell\right), & \!\!\!\text{if }\ell\notin\mathcal{L}(\boldsymbol{X})
\end{cases}\!.\label{eq:Survival_transition2-1}
\end{equation}
The standard multi-object transition density only generates unfragmented
trajectories, and is completely characterized by the model parameters
$w_{B,k},p_{B,k},P_{S,k},f_{S,k}$. 

More sophisticated multi-object dynamic models include spawnings \cite{MahlerBook07,MahlerBook14,Mahler2003,Bryantetal18},
division \cite{Nguyenetal-CellSpawning21}, and interactions between
objects \cite{Gostar-Interacting-19}, see also Subsection \ref{ss:Non-Standard-model}.

\subsubsection{Multi-Object Observation\protect\label{sss:Multi-object_Observation}}

At time $k$, an \textit{association map} $\theta:\mathbb{L}_{k}\rightarrow\{0:|Z_{k}|\}$
associates the labels of $\boldsymbol{X}_{k}$ with the elements of
$Z_{k}$, satisfying the \textit{positive 1-1} property that \emph{no
two distinct arguments are mapped to the same positive value} \cite{VoConj13}.
Here, $\theta(\ell)>0$ means $\ell$ generates detection $z_{\theta(\ell)}\in Z_{k}$,
and $\theta(\ell)=0$ means $\ell$ is misdetected. The positive 1-1
property ensures each detection comes from at most one object. The
space of all such association maps is denoted as $\Theta_{k}$.

In the \textit{standard multi-object observation model}, each $\boldsymbol{x}\in\boldsymbol{X}_{k}$,
is either detected  with probability $P_{D,k\!}\left(\boldsymbol{x}\right)$
and generates a detection $z$ with likelihood $g_{k}\left(z|\boldsymbol{x}\right)$
or missed with probability $1-P_{D,k\!}\left(\boldsymbol{x}\right)$.
Assuming that conditional on $\boldsymbol{X}_{k}$ detections are
independently generated, the RFS $D_{k}(\boldsymbol{X}_{k})$ of detections
is a multi-Bernoulli with parameters $\{(P_{D,k}(\boldsymbol{x}),g_{k}(\cdot|\boldsymbol{x}))\}_{\boldsymbol{x}\in\boldsymbol{X}_{k}}$.
Clutter $K_{k}$ is modeled as a Poisson RFS with intensity $\kappa_{k}$,
and assumed independent of the detections. The standard multi-object
likelihood function is given by \cite{VoConj13}  
\begin{equation}
\boldsymbol{g}_{k}(Z_{k}|\boldsymbol{X}_{k})\propto\sum_{\theta\in\Theta_{k}(\mathcal{L}(\boldsymbol{X}_{k}))}\left[\varPsi_{Z_{k}}^{(\theta\circ\mathcal{L}(\cdot))}(\cdot)\right]^{\boldsymbol{X}_{k}},\label{e:Single_Sensor_Lkhd-1}
\end{equation}
where $\Theta_{k}(I)\subseteq\Theta_{k}$ denotes the collection
of association maps with domain $I$, and 
\begin{equation}
\varPsi_{k,Z}^{(j)}(\boldsymbol{x})=\begin{cases}
\frac{P_{D,k}\left(\boldsymbol{x}\right)g_{k}\left(z_{j}|\boldsymbol{x}\right)}{\kappa_{k}\left(z_{j}\right)}, & j\in\{1\text{:}|Z|\}\\
1-P_{D,k}\left(\boldsymbol{x}\right), & j=0
\end{cases}.
\end{equation}
Note that the likelihood function (\ref{e:Single_Sensor_Lkhd-1})
is characterized by the `Signal to Noise Ratio' (SNR) function $\varPsi_{k,Z}^{(\theta\circ\mathcal{L})}$,
and does not suffer from the non-symmetry problem discussed in Subsection
\ref{subsec:Tricky} because \textit{each of its terms is symmetric
in the labeled states}. Hence, truncation before or after multiplication
by a valid prior, still leaves it a valid set function. 

The standard multi-object observation model accommodates non-homogeneous
clutter and state-dependent detection probability, covering a broad
range of problems. However, it does not address merged measurements
\cite{Beardetal14}, occlusions \cite{Ong-TPAMI-20}, extended measurements,
superpositional and image measurements \cite{MahlerBook07,MahlerBook14},
which require more sophisticated models. 

\subsubsection{Multi-Sensor Multi-Object Observation}

In a multi-sensor setting with $V$ sensors, each sensor registers
an observation set $Z_{k}^{(v)}$, with (standard) multi-object likelihood
$\boldsymbol{g}_{k}^{(v)\!}(Z_{k}^{(v)}|\boldsymbol{X}_{k})$, $v\in\{1\textrm{:}V\}$.
Assuming that $Z_{k}^{(1)},...,Z_{k}^{(V)}$ are independent conditional
on $\boldsymbol{X}_{k}$, the multi-sensor multi-object likelihood
function is given by the product \cite{MahlerBook07,MahlerBook14}
\vspace{-0.2cm}
\begin{equation}
\boldsymbol{g}_{k}(Z_{k}^{(1)},...,Z_{k}^{(V)}|\boldsymbol{X}_{k})\triangleq\overset{V}{\underset{v=1}{\prod}}\boldsymbol{g}_{k}^{(v)}(Z_{k}^{(v)}|\boldsymbol{X}_{k}).
\end{equation}
\vspace{-0.2cm}

Let $\varPsi_{k,Z^{\left(v\right)}}^{(v,j^{\left(v\right)})}$ denote
the `SNR' function of sensor $v$, and define the multi-sensor SNR
function
\begin{align}
\varPsi_{k,(Z^{(1)},...,Z^{(V)})}^{\left(j^{\left(1\right)},\dots,j^{\left(V\right)}\right)}\left(\boldsymbol{x}\right) & \triangleq\prod\limits_{v=1}^{V}\varPsi_{k,Z^{\left(v\right)}}^{(v,j^{\left(v\right)})}\left(\boldsymbol{x}\right),\label{e:MS_Lkhd_Abbrev_6-2}
\end{align}
the multi-sensor observation $Z_{k}\triangleq(Z_{k}^{(1)},...,Z_{k}^{(V)})$,
the multi-sensor association map as the $V$-tuple $\theta\triangleq(\theta^{\left(1\right)},...,\theta^{\left(V\right)})$
of association maps $\theta^{\left(v\right)}\in\Theta_{k}^{(v)}(I)$
from every sensor, and $\Theta_{k}(I)\triangleq\cprod_{v=1}^{V}\Theta_{k}^{\left(v\right)}(I)$.
Then, the multi-sensor multi-object observation likelihood can be
written in the form (\ref{e:Single_Sensor_Lkhd-1}).

\subsubsection{Relation with Unlabeled Models\protect\label{ss:Multi-object-Loc}}

Historically, only the unlabeled multi-object transition density
and observation likelihood were developed (for multi-object localization).

The unlabeled multi-object transition density $f_{k}\!$ is given
by the convolution (see (\ref{eq:Convolution-PDF}) with $n=2$)
\[
f_{k}\!\left(X_{k}|X_{k\texttt{-}1}\right)=\sum_{W\subseteq X_{k}}f_{\!B,k}\!\left(W\right)f_{\!S,k}\!\left(X_{k}-W|X_{k\texttt{-}1}\right)\!,
\]
of the multi-object densities $f_{B,k}$ (of new born objects) and
$f_{\!S,k}(\cdot|X_{k\texttt{-}1})$ (of objects generated from 
$X_{k\texttt{-}1}$) \cite{MahlerBook07}. Without labeling, the sum
over all subsets of $X_{k}$ does not reduce to a single term like
its labeled counterpart. In the standard multi-object dynamic model,
$f_{\!S,k}(\cdot|X_{k\texttt{-}1})$ is a multi-Bernoulli of the form
(\ref{eq:Multi-Bernoulli-density}), which is (algebraically) more
complex than the LMB (\ref{eq:Survival-trans}). Further, the multi-Bernoulli
model cannot ensure distinct states (in the transition to time $k$)
necessary for the set representation of the multi-object state (see
Remark \ref{rm3:negligible_dependence}). The LRFS formulation avoids
such problems.

The standard unlabeled multi-object observation likelihood $g_{k}$
takes on a nearly identical form to its labeled counterpart (\ref{e:Single_Sensor_Lkhd-1}),
specifically \cite[p. 421]{MahlerBook07}
\begin{eqnarray*}
g_{k}(Z_{k}|\{x_{1:n}\})\propto\sum_{\theta\in\Theta_{k}(\{1:n\})}\prod_{i=1}^{n}\varPsi_{Z_{k}}^{(\theta(i))}(x_{i}),
\end{eqnarray*}
where $\Theta_{k}(\{1\!:\!n\})$ denotes the space of association
maps with domain $\{1\!:\!n\}$. The subtle, but important difference
is that $g_{k}(Z_{k}|X)$ cannot be written as a sum of multi-object
exponentials in $X$. Further, $g_{k}(Z_{k}|\cdot)$ (and hence the
\textit{unlabeled multi-object posterior} $\!\pi_{k}(\cdot|Z_{k})\propto\!g_{k}(Z_{k}|\cdot)\pi_{k}(\cdot)$)
is \textit{not closed under truncation} as illustrated in Example
\ref{ex3:unlabeled-nonsym}. 
\begin{example}
\label{ex3:unlabeled-nonsym}Consider a multi-object state $\{x_{1},x_{2}\}$,
with $Z_{k}=\{z_{1}\}$, $P_{D\!}^{\!}=0.5$, and $\kappa=1$. The
possible positive 1-1 mappings from $\{1,2\}$ to $\{0,1\}$ are:
$[\theta(1),\theta(2)]=[1,0]$ ($x_{1}$ detected and $x_{2}$ undetected);
$[\theta(1),\theta(2)]=[0,1]$ ($x_{1}$ undetected and $x_{2}$ detected);
and $[\theta(1),\theta(2)]=[0,0]$, ($x_{1}$ and $x_{2}$ undetected),
hence $g_{k}(\{z_{1}\}|\{x_{1},x_{2}\})\propto g_{k}(z_{1}|x_{1})+g_{k}(z_{1}|x_{2})+1$.
Truncating $g_{k}(z_{1}|x_{2})$ yields $\hat{g}_{k}(\{z_{1}\}|x_{1},x_{2})=g_{k}(z_{1}|x_{1})+1$,
which is not a function of the set $\{x_{1},x_{2}\}$, because $\hat{g}_{k}(\{z_{1}\}|x_{2},x_{1})=g_{k}(z_{1}|x_{2})+1\neq\hat{g}_{k}(\{z_{1}\}|x_{1},x_{2})$.
Note that the labeled representation avoids this problem because $\theta$
maps the label of each state to the detections, making each term invariant
to the listing order of the states. 
\end{example}
\vspace{-0.5cm}
\subsection{GLMB Filter\protect\label{ss:GLMB-filter}}

The crux of multi-object estimation lies in the solutions to the multi-object
filtering/posterior recursions. This subsection presents a GLMB-based
analytic solution to the multi-object Bayes filter (\ref{eq:MO-bayes-2}),
under the standard multi-object SSM model (inclusive of multiple sensors).
For convenience, hereon, we omit references to the time index $k$,
and denote $k\pm1$, with subscripts `$\pm$', e.g., $\mathbb{L_{\texttt{-}}\triangleq L}_{k\texttt{-}1}$,
$\mathbb{B}\mathbb{\triangleq B}_{k}$, $\mathbb{L}\mathbb{\triangleq L}_{\texttt{-}}\cup\mathbb{B}$.
Also, when we write $\{P(\xi,I_{\texttt{-}},\theta,I)\!:\!(\xi,I_{\texttt{-}},\theta,I)\}$,
it is understood that the variables $\xi$, $I_{\texttt{-}}$, $\theta$,
and $I$, respectively, range over the spaces $\Xi$, $\mathcal{F}(\mathbb{L}_{\texttt{-}})$,
$\Theta$, and $\mathcal{F}(\mathbb{L})$, unless otherwise stated.

\subsubsection{Chapman-Kolmogorov Prediction}

Suppose that a multi-object state with density $\boldsymbol{\pi}_{\texttt{-}}\triangleq\boldsymbol{\pi}_{k\texttt{-}1}$,
at time $k-1$, evolves to the current time $k$ according to a multi-object
transition density $\boldsymbol{f}\triangleq\boldsymbol{f}_{\!k}$.
Then the predicted multi-object density $\boldsymbol{\pi}\triangleq\boldsymbol{\pi}_{k}$
is given by the Chapman-Kolmogorov equation
\begin{align}
\mathbf{\boldsymbol{\pi}}(\boldsymbol{X}) & =\int\!\boldsymbol{f}\!\left(\boldsymbol{X}|\boldsymbol{X}_{\texttt{-}}\right)\mathbf{\boldsymbol{\pi}}_{\texttt{-}}(\boldsymbol{X}_{\texttt{-}})\delta\boldsymbol{X}_{\texttt{-}},\label{eq:MO-bayes-2-1-1}
\end{align}
which defines the \textit{prediction operator} $\varPi:\mathbf{\boldsymbol{\pi}_{\texttt{-}}}\mapsto\mathbf{\boldsymbol{\pi}}$.

In \cite{VoConj13,VoVoP14}, it was established that the GLMB family
is closed under the prediction operator $\varPi$ with the standard
multi-object transition density $\boldsymbol{f}$ (i.e., model parameters
$w_{B},p_{B},P_{S},f_{S}$). Specifically, for the GLMB 
\begin{equation}
\mathbf{\boldsymbol{\pi}}_{\texttt{-}}=\{(w_{\texttt{-}}^{(\xi)}(I_{\texttt{-}}),p_{\texttt{-}}^{(\xi)})\!:\!(\xi,I_{\texttt{-}})\},\label{eq:previousGLMB}
\end{equation}
the prediction density is the GLMB
\begin{equation}
\mathbf{\boldsymbol{\pi}}=\varPi(\mathbf{\boldsymbol{\pi}}_{\texttt{-}})=\{(w^{(\xi)}(I),p^{(\xi)})\!:\!(\xi,I)\},\label{eq:PropCKstrong1-1}
\end{equation}
where \vspace{-0.2cm}
\begin{eqnarray}
\!\!\!\!\!\!\!w^{(\xi)}(I)= & \!\!\!\!\!\! & \!\!\!w_{B}(I\cap\mathbb{B})w_{S}^{(\xi)}(I\cap\mathbb{L}_{\texttt{-}}),\label{eq:PropCKstrong2}\\
\!\!\!\!\!\!\!p^{(\xi)}(x,\ell)= & \!\!\!\!\!\! & \!\!\!\boldsymbol{1}_{\mathbb{B}}(\ell)p_{B}(x,\ell)+\boldsymbol{1}_{\mathbb{L}_{\texttt{-}}}(\ell)p_{S}^{(\xi)}(x,\ell),\label{eq:PropCKstrongws}\\
\!\!\!\!\!\!\!w_{S}^{(\xi)}(L)= & \!\!\!\!\!\! & \!\!\![\bar{P}_{S}^{(\xi)}]^{L}\sum_{I_{\texttt{-}}\supseteq L}[1-\bar{P}_{S}^{(\xi)}]^{I_{\texttt{-}}\texttt{-}L}w_{\texttt{-}}^{(\xi)}(I_{\texttt{-}}),\label{eq:PropCKstrong_eta}\\
\!\!\!\!\!\!\!p_{S}^{(\xi)}(x,\ell)= & \!\!\!\!\!\! & \!\!\!\bigl\langle p_{\texttt{-}}^{(\xi)}(\cdot)f_{S}(x|\cdot)P_{S}(\cdot)\bigr\rangle(\ell)/\bar{P}_{S}^{(\xi)}(\ell),\label{eq:PropCKstrong3}\\
\!\!\!\!\!\!\!\bar{P}_{S}^{(\xi)}(\ell)= & \!\!\!\!\!\! & \!\!\!\bigl\langle P_{S\,}p_{\texttt{-}}^{(\xi)}\bigr\rangle(\ell).\label{eq:PropCKstrong4}
\end{eqnarray}

The prior GLMB components $(\xi,I_{\texttt{-}})$ generate the predicted
components $(\xi,I)$. The predicted weight $w^{(\xi)}(I)$ of $(\xi,I)$
is the product of the birth weight for new objects with labels in
$I$, and the survival weight for old objects with labels in $I$.
The predicted attribute density $p^{(\xi)}(\cdot,\ell)$ is either
the birth attribute density for a new object with label $\ell$, or
the predicted attribute density for a surviving object with label
$\ell$. 

\subsubsection{Bayes Update}

Suppose that a multi-object state with prior density $\boldsymbol{\pi}$
generates an observation $Z$ according to a multi-object observation
likelihood $\boldsymbol{g}(Z|\boldsymbol{X})\triangleq\boldsymbol{g}_{k}(Z|\boldsymbol{X})$.
Then, the multi-object posterior density is given by Bayes rule
\begin{align}
\mathbf{\boldsymbol{\pi}}(\boldsymbol{X}|Z) & =\frac{\boldsymbol{g}(Z|\boldsymbol{X})\mathbf{\boldsymbol{\pi}}(\boldsymbol{X})}{\int\boldsymbol{g}(Z|\boldsymbol{Y})\mathbf{\boldsymbol{\pi}}(\boldsymbol{Y})\delta\boldsymbol{Y}},\label{eq:MO-bayes-2-1-1-1}
\end{align}
which defines the \textit{Bayes update operator} $\mathit{\Upsilon_{Z}}:\mathbf{\boldsymbol{\pi}}\mapsto\mathbf{\boldsymbol{\pi}}(\cdot|Z)$.

In \cite{VoConj13,VoVoP14}, it was established that the GLMB family
is closed under the update operator $\mathit{\Upsilon_{Z}}$, i.e.,
a conjugate prior w.r.t. the standard $g(Z|\boldsymbol{X})$ (with
SNR function $\varPsi_{Z}^{(\theta\circ\mathcal{L})}$ and model parameters
$\kappa,P_{D},g$). Specifically, if the prior $\boldsymbol{\pi}$
is the GLMB (\ref{eq:PropCKstrong1-1}), then the posterior is the
GLMB
\begin{equation}
\boldsymbol{\pi}(\cdot|Z)=\mathit{\Upsilon_{Z}}(\mathbf{\boldsymbol{\pi}})\propto\{(w_{Z}^{(\xi,\theta)\!}(I),p_{Z}^{\!(\xi,\theta)\!})\!:\!(\xi,\theta,I)\},\label{eq:PropBayes_strong0}
\end{equation}
where\vspace{-0.1cm}
\begin{eqnarray}
w_{Z}^{(\xi,\theta)\!}(I)\!\!\! & = & \!\!\!w^{(\xi)}(I)\boldsymbol{1}_{\Theta(I)}(\theta)\left[\bar{\varPsi}_{Z}^{(\xi,\theta(\cdot))}(\cdot)\right]^{I},\label{eq:PropBayes_strong1}\\
p_{Z}^{\!(\xi,\theta)\!}(x,\ell)\!\!\! & = & \!\!\!p^{(\xi)}(x,\ell)\varPsi_{Z}^{(\theta(\ell))}(x,\ell)/\bar{\varPsi}_{Z}^{(\xi,\theta)}(\ell),\label{eq:PropBayes_strong2}\\
\bar{\varPsi}_{Z}^{(\xi,j)}(\ell)\!\!\! & = & \!\!\!\bigl\langle p^{(\xi)}\varPsi_{Z}^{(j)}\bigr\rangle(\ell).\label{eq:PropBayes_strong3}
\end{eqnarray}

Note that each prior GLMB component $(\xi,I)$ generates a series
of updated components $(\xi,\theta,I)$. Due to the term $\boldsymbol{1}_{\Theta(I)\!}\left(\theta\right)$,
only components with $\mathcal{D}(\theta)=I$ are needed. The weight
$w_{Z}^{(\xi,\theta)\!}(I)$ of each $(\xi,\theta,I)$ is the product
of the validity check for $\theta$, the Bayes evidence from observation
$Z,$ and the prior weight $w^{(\xi)\!}(I)$. The corresponding attribute
density $p_{Z}^{\!(\xi,\theta)\!}(\cdot,\ell)$ is simply the Bayes
update of the prior attribute with detection $z_{\theta(\ell)}\in Z$
or a misdetection if $\theta(\ell)=0$. 

\subsubsection{GLMB Filtering Recursion}

The Bayes filtering recursion (\ref{eq:MO-bayes-2}) is the composition
of the prediction and update operators, i.e., $\mathit{\Upsilon_{Z}}\circ\varPi:\mathbf{\boldsymbol{\pi}}_{\texttt{-}}\mapsto\mathbf{\boldsymbol{\pi}}(\cdot|Z)$.
Hence, for the standard multi-object SSM, the GLMB family is closed
under the Bayes recursion, i.e., starting with an initial GLMB prior,
the multi-object filtering density at any time is a GLMB \cite{VoConj13,VoVoP14}.
This result also holds with a GLMB birth model, but for simplicity
in the subsequent discussion, we assume an LMB birth model. The direct
propagation of the GLMB  $\mathbf{\boldsymbol{\pi}}_{\texttt{-}}$
in (\ref{eq:previousGLMB}) to $\mathbf{\boldsymbol{\pi}}(\cdot|Z)$
in (\ref{eq:PropBayes_strong0}) is given by (see \cite{VoVoH2017})
\begin{align}
w_{Z}^{\left(\xi,\theta\right)}(I) & \propto\sum\limits_{I_{\texttt{-}}\subseteq\mathbb{L}}w_{Z}^{\left(\xi,I_{\texttt{-}},\theta,I\right)}w_{\texttt{-}}^{\left(\xi\right)}(I_{\texttt{-}}),\label{e:GLMB_JPU_0-1}\\
p_{Z}^{(\xi,\theta)}(x,\ell) & \propto\begin{cases}
\bigl\langle\Lambda_{S}^{(\theta(\ell))}(x|\cdot)p_{\texttt{-}}^{(\xi)}(\cdot)\bigr\rangle(\ell), & \!\!\!\!\ell\in\mathbb{L}_{\texttt{-}}\\
\Lambda_{B}^{(\theta(\ell))}(x,\ell), & \!\!\!\!\ell\in\mathbb{B}
\end{cases},\label{e:GLMB_JPU_0-2}
\end{align}
where
\begin{align}
\!\!\!\!w_{Z}^{\left(\xi,I_{\texttt{-}},\theta,I\right)} & =\boldsymbol{1}_{\mathcal{F}(I_{\texttt{-}}\uplus\mathbb{B})\!}\left(I\right)\boldsymbol{1}_{\Theta(I)\!}\left(\theta\right)\prod_{\ell\in I_{\texttt{-}}\uplus\mathbb{B}}\eta_{Z}^{(\xi,I,\ell)}(\theta(\ell)),\label{e:GLMB_JPU_3-1}\\
\!\!\!\!\eta_{Z}^{(\xi,I,\ell)}(j) & =\!\begin{cases}
1-\bigl\langle P_{S}p_{\texttt{-}}^{(\xi)}\bigr\rangle(\ell), & \!\!\!\!\ell\in\mathbb{L}_{\texttt{-}}\!-\!I\text{ }\\
{\textstyle \int}\bigl\langle\Lambda_{S}^{\!(j)\!}(x|\cdot)p_{\texttt{-}}^{(\xi)}(\cdot)\bigr\rangle(\ell)dx, & \!\!\!\!\ell\in\mathbb{L}_{\texttt{-}}\!\cap\!I\\
1-P_{B}(\ell), & \!\!\!\!\ell\in\mathbb{B}-I\\
{\textstyle \int}\Lambda_{B}^{(j)}(x,\ell)dx, & \!\!\!\!\ell\in\mathbb{B}\cap I
\end{cases}\!\!,\!\label{e:GLMB_JPU_5-1}\\
\!\!\!\!\Lambda_{S}^{\!(j)\!}(x|\varsigma,\ell) & =\varPsi_{Z}^{(j)}(x,\ell)f_{S}(x|\varsigma,\ell)P_{S}(\varsigma,\ell),\label{e:GLMB_JPU_6-1}\\
\!\!\!\!\Lambda_{B}^{\!(j)\!}(x,\ell) & =\varPsi_{Z}^{(j)}(x,\ell)p_{B}(x,\ell)P_{B}(\ell).\label{e:GLMB_JPU_7-1}
\end{align}

A salient feature of the GLMB filter is the provision for smoothed
(single-object) trajectory estimates \cite{VoVoBeard19,Nguyen-Psmooth-Sensors-19}.
Suppose that all observations up to the current time follow the standard
observation model (LRFS multi-object estimation accommodates updates
with different types of observations). Then, in the GLMB filtering
density propagation, say from $\{(w(J),p)\!:\!J\!\in\!\mathcal{F}(\mathbb{L}_{0})\}$
at time 0, to $\{(w^{(\xi)\!}(I_{\texttt{-}}),p^{(\xi)})\!:\!(\xi,I_{\texttt{-}})\}$
at time $k-1,$ we recursively constructs the `event' $\xi$ as the
history $\theta_{1:k\texttt{-}1}$ of association maps. Thus, the
attribute density $p^{(\xi,\theta)}(\cdot,\ell)=p^{(\theta_{0:k})}(\cdot,\ell)$
contains the entire history $\theta_{0:k}(\ell)$ of detections associated
with trajectory $\ell$ \cite{VoConj13,VoVoP14}. This information,
encapsulated in the multi-object filtering density, can be used to
estimate the entire trajectory (or over a moving window) via smoothing
\cite{VoVoBeard19,Nguyen-Psmooth-Sensors-19}.

The GLMB recursion, defined by the above so-called \textit{joint GLMB
prediction and update}, is a true Bayesian MOT filter with provably
Bayes-optimal track management. Since, each component $(\xi,I_{\texttt{-}})$
is propagated forward as a set of children components $\{(\xi,\theta,I)\!:\!(\theta,I)\}$,
the multi-object filtering density accumulates an intractably large
number of components with time. This inevitably requires approximation,
either by a simpler multi-object density, or by truncation.

\subsection{LMB Filter\protect\label{ss:LMB-filter}}

This subsection presents the LMB filter, regarded as the PHD filter
for trajectories. In the same way that the PHD filter approximates
the unlabeled multi-object filtering density by a Poisson with matching
PHD, the LMB filter approximates the GLMB filtering density by an
LMB with matching PHD. However, unlike the PHD filter, the LMB filter
provides trajectory estimates, and does not suffer from high cardinality
variance when the number of objects is large.

Using an LMB birth model, the LMB sub-family is also closed under
the prediction operator $\varPi$. Specifically, given a previous
LMB filtering density $\mathbf{\boldsymbol{\pi}}_{\texttt{-}}=\{(r_{\texttt{-}}(\ell),p_{\texttt{-}}(\cdot,\ell))\}_{\ell\in\mathbb{L}_{\texttt{-}}}$,
the multi-object prediction density is the LMB \cite{Reuter2014}
\begin{equation}
\mathbf{\boldsymbol{\pi}}=\boldsymbol{f}_{\!B}\cup\{(r_{S}(\ell),p_{S}(\cdot,\ell))\}_{\ell\in\mathbb{L}_{\texttt{-}}},\label{eq:PropCKstrong1-1-1}
\end{equation}
where $\boldsymbol{f}_{\!B}$ is the birth LMB density/parameter,
\begin{flalign}
r_{S}(\ell) & =r_{\texttt{-}}(\ell)\left\langle P_{S\,}p_{\texttt{-}}\right\rangle \!(\ell),\label{eq:PropCKstrong2-1}\\
p_{S}(x,\ell) & =\frac{\left\langle f_{S}(x|\cdot)P_{S}(\cdot)p_{\texttt{-}}(\cdot)\right\rangle \!(\ell)}{\left\langle P_{S\,}p_{\texttt{-}}\right\rangle \!(\ell)}.\label{eq:PropCKstrong4-1}
\end{flalign}
While the LMB prediction operation is exact and intuitively appealing,
the LMB family is not closed under Bayes rule.

The LMB filter uses the above LMB prediction, and applies the update
operator to the predicted LMB $\mathbf{\boldsymbol{\pi}}$, yielding
the GLMB $\mathit{\Upsilon_{Z}}(\mathbf{\boldsymbol{\pi}})=\{(w_{Z}^{\left(\theta\right)\!}(I),p_{Z}^{\left(\theta\right)})\!:\!(\theta,I)\}$.
Further, using (\ref{eq:existence-prob-GLMB}), (\ref{eq:track-density-GLMB}),
the GLMB $\mathit{\Upsilon_{Z}}(\mathbf{\boldsymbol{\pi}})$ is approximated
by the LMB 
\begin{equation}
\mathbf{\boldsymbol{\pi}}(\cdot|Z)=\{(r_{Z}(\ell),p_{Z}(\cdot,\ell))\}_{\ell\in\mathbb{L}},\label{eq:PropCKstrong1-1-1-1}
\end{equation}
with matching PHD, by setting
\begin{align}
r_{Z}\!\left(\ell\right) & =\sum_{\theta\in\Theta}\sum_{I\subseteq\mathbb{L}}\mathbf{1}_{I}(\ell)w_{Z}^{\left(\theta\right)}(I),\label{eq:existence-prob-GLMB-1}\\
p_{Z}\!\left(y,\ell\right) & \propto\sum_{\theta\in\Theta}p_{Z}^{\left(\theta\right)}\!\left(y,\ell\right)\sum_{I\subseteq\mathbb{L}}\mathbf{1}_{I}(\ell)w_{Z}^{\left(\theta\right)}(I).\label{eq:track-density-GLMB-1}
\end{align}

Due to the smaller number of components than the GLMB filter, the
LMB filter is faster, albeit with some degradation in tracking performance.
 
\begin{rem}
\label{rm10:LMB_mixture}If the weight $w_{Z}^{\left(\theta\right)}(I)$
of every component $(w_{Z}^{\left(\theta\right)}(I),p_{Z}^{\left(\theta\right)})$
of the GLMB $\mathit{\Upsilon_{Z}}(\mathbf{\boldsymbol{\pi}})$ can
be written in the form $K_{Z}^{(\theta)}[r_{Z}^{(\theta)}/(1-r_{Z}^{(\theta)})]^{I}$,
then $\Delta\!\left(\boldsymbol{X}\right)w_{Z}^{\left(\theta\right)}(\mathcal{L}\left(\boldsymbol{X}\right))[p_{Z}^{\left(\theta\right)}]^{\boldsymbol{X}}$
is indeed an (unnormalized) LMB, see (\ref{eq:LMB-density-1}), and
hence $\mathit{\Upsilon_{Z}}(\mathbf{\boldsymbol{\pi}})$ is a mixture
of LMBs. However, this is not the case because
\begin{eqnarray*}
w_{Z}^{\left(\theta\right)}(I)=\boldsymbol{1}_{\Theta(I)}(\theta)K_{Z}^{(\theta)}\left[r_{Z}^{(\theta)}/(1-r_{Z}^{(\theta)})\right]^{I},
\end{eqnarray*}
which cannot be expressed as a product over $I$ due to the term
$\boldsymbol{1}_{\Theta(I)}(\theta)$. Thus, LMB mixtures are not
conjugate priors.   
\end{rem}

\subsection{GLMB Filter Implementation\protect\label{ss:GLMB-Imp}}

Implementing the GLMB filter requires truncation of the multi-object
filtering density without exhaustive enumeration. Each component $(\xi,I_{\texttt{-}})$
of the GLMB filtering density at time $k-1$ propagates a (very large)
set \textbf{$\left\{ (\xi,I_{\texttt{-}},\theta,I):(\theta,I)\right\} $}
of ``children'' components (before marginalizing out $I_{\texttt{-}}$)
to the current time. Truncation by selecting ``children'' with significant
weights minimizes the $L_{1}$-approximation error \cite{VoVoP14},
and can be accomplished via solving the rank assignment problem or
via Gibbs sampling (GS). Computing $w_{Z}^{\left(\xi,\theta\right)}(I)$
and $p_{Z}^{\left(\xi,\theta\right)}$ can be accomplished via single-object
filtering techniques.

Given a fixed component $(\xi,I_{\texttt{-}})$ at time $k-1$, the
goal is to find a set of $(\theta,I)\!\in\!\Theta\times\mathcal{F}(\mathbb{L})$
with significant weight increment $w_{Z}^{\left(\xi,I_{\texttt{-}},\theta,I\right)}$.
Due to the terms $\boldsymbol{1}_{\mathcal{F}(I_{\texttt{-}}\uplus\mathbb{B})\!}\left(I\right)$
and $\boldsymbol{1}_{\Theta(I)\!}\left(\theta\right)$ in (\ref{e:GLMB_JPU_3-1}),
we only need to consider $(\theta,I)$ with $I\subseteq I_{\texttt{-}}\uplus\mathbb{B}$
and $\mathcal{D}(\theta)=I$. Hence, it is convenient to represent
such $(\theta,I)$ by an \textit{extended association map} (or simply
extended association) $\gamma\!:\!I_{\texttt{-}}\uplus\mathbb{B}\!\rightarrow\!\{-1\!:\!\left|Z\right|\}$,
defined by
\begin{equation}
\gamma(\ell)\triangleq\begin{cases}
\theta(\ell), & \text{if }\ell\in\mathcal{D}(\theta)\\
-1, & \text{\text{if }}\ell\in(I_{\texttt{-}}\uplus\mathbb{B})-\mathcal{D}(\theta)
\end{cases}.\label{eq:gamma}
\end{equation}
Note that $\gamma$ inherits the \emph{positive 1-1} property, and
the set of all such $\gamma$ is denoted by $\Gamma$. Since $(\theta,I)$
is recovered by $I=\{\ell\!:\!\gamma(\ell)\geq0\}$, $\theta(\ell)=\gamma(\ell)$
for each $\ell\in I$, it follows that $\boldsymbol{1}_{\mathcal{F}(I_{\texttt{-}}\uplus\mathbb{B})\!}\left(I\right)=1$
and $\mathcal{D}(\theta)=I$. Enumerating $I_{\texttt{-}}\uplus\mathbb{B}=\{\ell_{1:P}\}$,
and $Z=\{z_{1:M}\}$, $\gamma$ can be represented as a \textit{P}-tuple
in $\{-1$:$M\}^{P}$. Further, abbreviating
\begin{equation}
\!\!\eta^{(i)\!}(j)\triangleq\begin{cases}
\!1-\bigl\langle p_{\texttt{-}}^{(\xi)}P_{S}\bigr\rangle(\ell_{i}), & \!\!\!\!\ell_{i}\in I_{\texttt{-}},j<0\text{ }\\
{\textstyle \!\int}\bigl\langle\!\Lambda_{S}^{\!(j)\!}(x|\cdot)p_{\texttt{-}}^{(\xi)\!}(\cdot)\!\bigr\rangle(\ell_{i})dx, & \!\!\!\!\ell_{i}\in I_{\texttt{-}},j\geq0\\
\!1-P_{B}(\ell_{i}), & \!\!\!\!\ell_{i}\in\mathbb{B},j<0\\
{\textstyle \!\int}\Lambda_{B}^{(j)}(x,\ell_{i})dx, & \!\!\!\!\ell_{i}\in\mathbb{B},j\geq0
\end{cases}\!\!,\label{eq:eta-1}
\end{equation}
(the dependence on $\xi$, $I_{\texttt{-}}$, and $Z$ are suppressed)
the weight increment (\ref{e:GLMB_JPU_3-1}) can be expressed as \cite{VoVoH2017}
\begin{equation}
w_{Z}^{\left(\xi,I_{\texttt{-}},\theta,I\right)}=\omega(\gamma)\triangleq\boldsymbol{1}_{{\Gamma}}(\gamma)\prod\limits_{i=1\!}^{P}\eta^{(i)}(\gamma(\ell_{i})).\label{eq:thetajoint_dis}
\end{equation}
The values of $\eta^{(i)}(j)$ is precomputed as the $P\!\times\!(M+2)$
\textit{association score matrix} shown in Fig.~\ref{fig:matrix},
and the goal becomes finding a set of $\gamma$ with significant $\omega(\gamma)$.

\begin{figure}
\centering\includegraphics{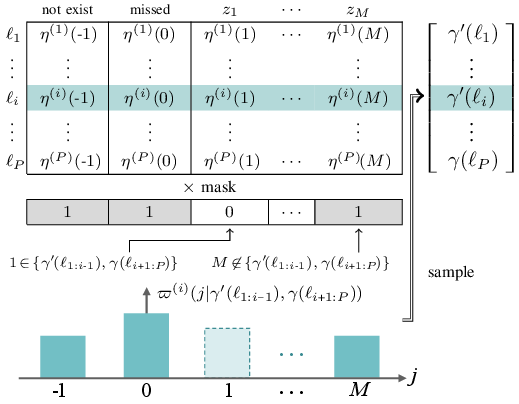} \caption{\protect\label{fig:matrix}Association score matrix and conditionals
(at time $k$). The $i$-th unnormalized conditional $\varpi^{(i)}$
is determined by taking the $i$-th row of the association score matrix,
i.e., $\eta^{(i)}(j)$, $j=-1:M$, and set $\eta^{(i)}(j)=0$ for
each positive $j$ associated with a label other than $\ell_{i}$.
\textcolor{blue}{}}
\vspace{-0.1cm}
\end{figure}

\subsubsection{Ranked Assignment}

The $K$ best (extended) associations in non-increasing order of $\omega(\gamma)$
can be obtained without exhaustive enumeration by solving a \textit{ranked
assignment problem}. Here, each association $\gamma\in\Gamma$ is
equivalently represented by a $P\!\times\!(M+2P)$ \emph{assignment
matrix} $S$ whose entries are either $0$ or $1$, with every row
summing to $1$, and every column summing to either $1$ or $0$.
The objective is to find the $K$ assignment matrices with smallest
$\text{trace}(S^{T}C)$, where $C$ is the $P\!\times\!(M+2P)$ \emph{cost
matrix }constructed from the association score matrix in Fig.~\ref{fig:matrix}
by \cite{VoVoH2017}
\begin{equation}
C_{i,j}=\begin{cases}
-\ln\eta^{(i)}(j), & j\in\{1:M\}\\
-\ln\eta^{(i)}(0), & j=M+i\\
-\ln\eta^{(i)}(-1), & j=M+P+i\\
\infty, & \text{otherwise}
\end{cases}.\label{eq:Ass_Cost_Matrix1}
\end{equation}
Since $\exp\!\left(-\text{trace}(S^{T}C)\right)=\prod_{i=1\!}^{P}\eta^{(i)}(\gamma_{i})=\omega(\gamma)$,
the $K$ best assignment matrices correspond to the $K$ associations
with largest weights \cite{VoVoH2017}. A GLMB filter implementation
with $\mathcal{O}\left(K(M+2P)^{4}\right)$ complexity was proposed
in \cite{hoang2015fast} using Murty's algorithm \cite{Murty68}
to solve the ranked assignment problems. More efficient algorithms
\cite{miller1997optimizing,pedersen2008algorithm} can reduce the
complexity to $\mathcal{O}\left(K(M+2P)^{3}\right)$. 

The initial GLMB filter implementation truncates the predicted and
updated multi-object densities at each time \cite{VoConj13,VoVoP14}.
Since truncation of the prediction is separated from the update, information
from the observation is not exploited, and a significant portion of
the predicted components subsequently generate updated components
with negligible weights, which waste computations. The joint prediction
and update avoids this problem, thus improving computational speed
\cite{VoVoH2017}. For the LMB special case, belief propagation can
also be exploited for a fast implementation \cite{Kropfreiteretal19}.
A drawback of using ranked assignment is the high computational cost
of generating a sequence of components ordered by their weights, whilst
such ordering is not needed in the GLMB approximation. 
\begin{figure*}
\begin{centering}
\resizebox{180mm}{!}{\includegraphics[clip]{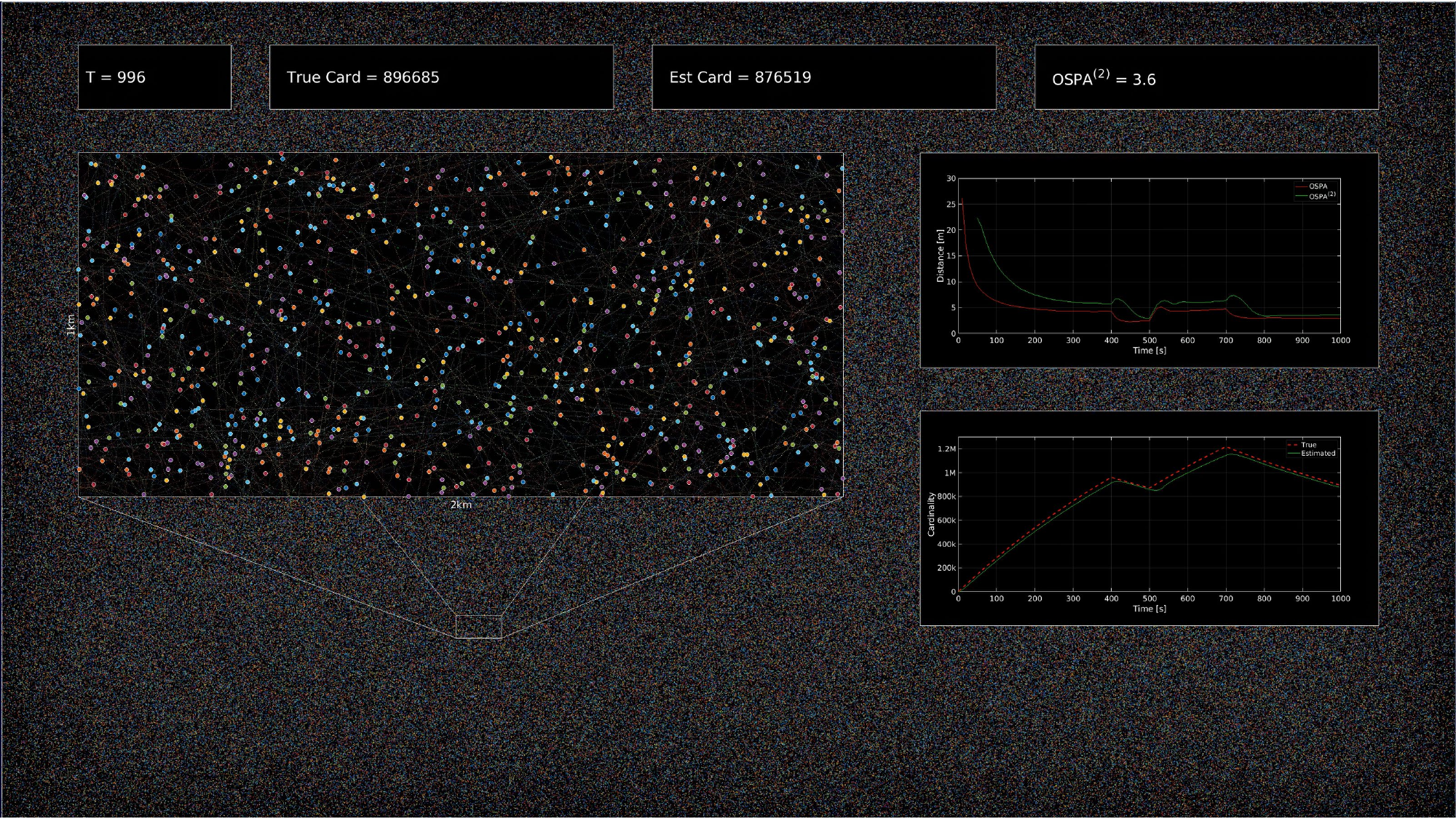}}
\par\end{centering}
\caption{Large-scale GLMB filter tracking over 1 million objects in a 64km$\times$36km
region \cite{Beard18-largescale}. The insets show a magnified 2km$\times$1km
region, OSPA/OSPA\protect\textsuperscript{(2)} errors, and cardinality.
Objects can appear anywhere, at a rate of 200-3000 per frame (unknown
to filter). The object's position and velocity follow a linear Gaussian
motion model with 0.2ms\protect\textsuperscript{-2} process noise
standard deviation, and 0.88 detection probability. Position observations
are corrupted by Gaussian noise with a 5m standard deviation, and
clutter, uniformly distributed on the region at an average of 460800
points. About 1 billion data points accumulate over 1000 frames,
each contains about 1 million objects on average, peaking to 1,217,531
at frame 700.}
\label{fig:large-scale}
\end{figure*}

\subsubsection{Gibbs Sampling\protect\label{sss:SGS-GLMB}}

An efficient way to generate significant GLMB components is to sample
from some probability distribution $\varpi$ on $\{-1\!:\!M\}^{P}$
\cite{VoVoH2017}. To ensure that mostly high-weight positive 1-1
\textit{P}-tuples are sampled, $\varpi$ is constructed so that only
positive 1-1 \textit{P}-tuples have positive probabilities, and those
with high weights are more likely to be chosen than those with low
weights. An obvious choice of $\varpi$ is (\ref{eq:thetajoint_dis}). 

Gibbs Sampling (GS) is an efficient special case of the Metropolis-Hasting
algorithm for sampling from an unnormalized distribution $\varpi$
\cite{geman1984stochastic,Cassella92}. This algorithm constructs
a Markov chain that sequentially generates a new iterate $\gamma^{\prime}$
from the current iterate $\gamma$. When the chain runs for long
enough, i.e., past the burn-in stage, subsequent samples would be
distributed according to the target distribution $\varpi$. Unlike
sampling for posterior inference, in GLMB filtering it is not necessary
to discard burn-ins  because all distinct positive 1-1 \textit{P}-tuples
will reduce the $L_{1}$-approximation error. 

The GLMB truncation developed in \cite{VoVoH2017} employs the classical
Systematic-scan GS (SGS) strategy, wherein the next iterate $\gamma^{\prime}$
 is generated from the current iterate $\gamma$ by sampling each
coordinate $\gamma^{\prime}(\ell_{i})$ of $\gamma^{\prime}$, for
$i=$ $1\!:\!P$, from the so-called the $i$-th \textit{conditional}
$\varpi^{(i)}(\cdot|\gamma^{\prime}(\ell_{1:i\texttt{-}1}),\gamma(\ell_{i\texttt{+}1:P}))$
on $\{-1\!:\!M\}$ \cite{geman1984stochastic,Cassella92}. It was
shown in \cite{VoVoH2017} that to arrive at the target distribution
(\ref{eq:thetajoint_dis}), the required unnormalized $i$-th conditional
is given by the $\!i$-th row of the association score matrix after
zeroing the entry of every positive $j$ that has been paired with
a label other than $\ell_{i}$, see Fig.~\ref{fig:matrix} for illustration.
More concisely, $\varpi^{(i)}(j|\gamma^{\prime}(\ell_{1:i\texttt{-}1}),\gamma(\ell_{i\texttt{+}1:P}))\propto\eta^{(i)}(j)$,
except at each positive $j$ in the set $\{\gamma^{\prime}(\ell_{1:i\texttt{-}1}),\gamma(\ell_{i\texttt{+}1:P})\}$
of values associated with labels other than $\ell_{i}$, wherein $\varpi^{(i)}(j|\gamma^{\prime}(\ell_{1:i\texttt{-}1}),\gamma(\ell_{i\texttt{+}1:P}))=0$.
This remarkably simple result provides an inexpensive way to generate
significant positive 1-1 tuples, and was exploited to implement SGS
GLMB truncation in \cite{VoVoH2017} with an $\mathcal{O}(TP^{2}M)$
complexity, but later reduced to $\mathcal{O}(TPM)$ in \cite{SVVOM2022Linear},
where $T$ is the number of iterates of the Gibbs sampler. Note that
starting with any positive 1-1 tuple, e.g., all 0's or all -1's, all
subsequent iterates are positive 1-1.

In \cite{VoVoBeard19}, SGS was extended to address the NP-hard multi-sensor
GLMB truncation problem with linear complexity in the total number
of detections across all sensors. This SGS multi-sensor GLMB filter
implementation also applies to approximations such as the LMB and
marginalized GLMB filters since these filters require full GLMB updates
\cite{Reuter2014}, \cite{Fantacci2016}. Such a multi-sensor LMB
filter implementation has been extended to address partially overlapping
fields of views in \cite{WXLBC2023Centralized}. 

Following the development of SGS GLMB truncation, other sampling-based
techniques have been proposed. In \cite{YangWW-18}, the positive
1-1 requirement was neglected to achieve linear complexity in $P$,
albeit at a degradation in performance. A herded GS implementation
of the LMB filter was proposed in \cite{wolf2020deterministic}, which
turned out to be slower than the SGS implementation. Spatial search
and approximation by a product of smaller GLMBs (see Subsection \ref{subsec:LRFS-approx})
that run in parallel were proposed in \cite{Beard18-largescale} to
reduce computation times. While this approach was  demonstrated to
track in excess of one million objects from approximately 1.3 million
detections at a time, see Fig. \ref{fig:large-scale}, the complexity
has not been reduced. Cross-entropy based solutions \cite{Nguyen2014}
were developed in \cite{YuSaucanetal17} and \cite{SV2018Distributed},
for multi-sensor GLMB filtering and its distributed version, but with
higher complexity than the SGS implementation in \cite{VoVoBeard19}.
The recently developed tempered GS (TGS) technique reduces the complexity
to $\mathcal{O}(T(P+M))$ \cite{SVVOM2022Linear}.

\subsection{Multi-Scan GLMB Filter/Smoother\protect\label{ss:Multi-object_Smoother}}

Similar to the (multi-object) Bayes filtering recursion, the posterior
recursion (\ref{eq:MO_posterior}) admits an analytic solution in
the form of a multi-scan GLMB. This subsection presents the multi-scan
GLMB recursion under the standard multi-object SSM model (with LMB
birth for simplicity). For convenience, when we write $\{P(\xi,\theta,I_{j:k})\!:\!(\xi,\theta,I_{j:k})\}$,
it is understood that the variables $\xi$, $\theta$, and $I_{j:k}$,
respectively, range over the spaces $\Xi$, $\Theta$, and $\cprod_{i=j}^{k}\mathcal{F}(\mathbb{L}_{i})$,
unless otherwise stated. 

The multi-scan GLMB is closed under the Bayes posterior recursion
(\ref{eq:MO_posterior}). Indeed, it is closed under: the prediction,
defined as $\boldsymbol{\pi}_{0:n}(\boldsymbol{X}_{0:n})\!=\boldsymbol{f}_{\!n}\!\left(\boldsymbol{X}_{n}|\boldsymbol{X}_{n\texttt{-}1}\right)\boldsymbol{\pi}_{0:n\texttt{-}1}(\boldsymbol{X}_{0:n\texttt{-}1})$;
and the Bayes update since it is a GLMB where the argument is a set
of labeled trajectories \cite{VoVomultiscan18}. Thus, starting with
an initial multi-scan GLMB prior, the multi-object prediction and
posterior at any time are multi-scan GLMBs.

The multi-scan GLMB recursion (or GLMB smoothing recursion) \cite[eq. (42)-(46)]{VoVomultiscan18}
propagates the multi-scan GLMB 
\begin{equation}
\boldsymbol{\pi}_{0:n\texttt{-}1}=\{(w_{\texttt{-}}^{(\xi)\!}(I_{0:n\texttt{-}1}),p_{\texttt{-}}^{(\xi)})\!:\!(\xi,I_{0:n\texttt{-}1})\},\label{eq:MSGLMB-prev}
\end{equation}
at time $n-1$ to the multi-scan GLMB
\begin{equation}
\boldsymbol{\pi}_{0:n}(\cdot|Z)\propto\{(w_{Z}^{(\xi,\theta)}(I_{0:n}),p_{Z}^{(\xi,\theta)})\!:\!(\xi,\theta,I_{0:n})\},\label{eq:MSGLMBBayesupdate}
\end{equation}
at time $n$, where
\begin{align}
w_{Z}^{(\xi,\theta)}(I_{0:n})= & \,w_{Z}^{\left(\xi,I_{n\texttt{-}1},\theta,I_{n}\right)}w_{\texttt{-}}^{(\xi)\!}(I_{0:n\texttt{-}1}),\label{eq:eq:MSGLMBBayesupdate-1}\\
p_{Z}^{(\xi,\theta)\!}(x_{s(\ell):t(\ell)},\ell)\propto\label{eq:eq:MSGLMBBayesupdate-2}\\
 & \!\!\!\!\!\!\!\!\!\!\!\!\!\!\!\!\!\!\!\!\!\!\!\!\!\!\!\!\!\!\!\!\!\!\!\left\{ \!\!\!\begin{array}{ll}
{\scriptstyle {\displaystyle p_{\texttt{-}}^{(\xi)\!}(x_{s(\ell):t(\ell)},\ell),}} & {\scriptstyle {\displaystyle \!\!\!\!t(\ell)<n\texttt{-}1}}\\
{\scriptstyle {\displaystyle (1\texttt{-}P_{S}(x_{t(\ell)},\ell))p_{\texttt{-}}^{(\xi)\!}(x_{s(\ell):t(\ell)},\ell),}} & {\scriptstyle {\displaystyle \!\!\!\!t(\ell)=n\texttt{-}1}}\\
{\scriptstyle {\displaystyle \Lambda_{S}^{\!(\theta(\ell))\!}(x_{n}|x_{n\texttt{-}1\!},\ell)p_{\texttt{-}}^{(\xi)\!}(x_{s(\ell):n\texttt{-}1},\ell),}} & {\scriptstyle {\displaystyle \!\!\!\!s(\ell)<t(\ell)=n}}\\
{\scriptstyle {\displaystyle \Lambda_{B}^{\!(\theta(\ell))\!}(x_{n},\ell),}} & {\scriptstyle {\displaystyle \!\!\!\!s(\ell)=t(\ell)=n}}
\end{array}\right.\!\!\!,\nonumber 
\end{align}
$w_{Z}^{\left(\xi,I_{n\texttt{-}1},\theta,I_{n}\right)}$ is the
weight increment given by (\ref{e:GLMB_JPU_3-1}), $s(\ell)$ and
$t(\ell)$ are, respectively, the earliest and latest times that $\ell$
exists on the time window $\{0\!:\!n\}$, $\Lambda_{S}^{\!(j)\!}(x_{n}|x_{n\texttt{-}1\!},\ell)$
and $\Lambda_{\!B}^{\!(j)\!}(x_{n},\ell)$ are, respectively, given
by (\ref{e:GLMB_JPU_6-1}) and (\ref{e:GLMB_JPU_7-1}). 

Observe the similarity between the GLMB filter (\ref{e:GLMB_JPU_0-1})-(\ref{e:GLMB_JPU_7-1})
and the multi-scan GLMB filter. Indeed, the latter is (algebraically)
simpler and more intuitive than the former since no marginalization
is needed. The new posterior GLMB weights are simply the old weights
scaled by their corresponding weight increments. Further, given the
old posterior density $p_{\texttt{-}}^{(\xi)\!}(\cdot,\ell)$ of trajectory
$\ell$ at time $n-1$: if $\ell$ died before time $n-1$, then its
new posterior is $p_{\texttt{-}}^{(\xi)\!}(\cdot,\ell)$; if $\ell$
died at time $n-1$, then its new posterior is $p_{\texttt{-}}^{(\xi)\!}(\cdot,\ell)$
times the death probability; if $\ell$ (born before time $n$) survives
to time $n$, then its new posterior is $p_{\texttt{-}}^{(\xi)\!}(\cdot,\ell)$
times the survival probability, transition density, and the SNR; if
$\ell$ is born at time $n$, then its new posterior is the probability
of birth times the attribute density of its new state, and the SNR. 

Similar to the GLMB filter, the number of components of the GLMB posterior
grows super-exponentially with time, and truncation by retaining a
prescribed number of components with highest weights minimizes the
$L_{1}$-norm approximation error \cite{VoVomultiscan18}. Unlike
the GLMB filter, the multi-scan GLMB truncation problem requires solving
large-scale \textit{multi-dimensional ranked assignment problems},
which is NP-hard for more than two dimensions.
\begin{figure*}
\begin{centering}
\resizebox{180mm}{!}{\includegraphics[width=18cm]{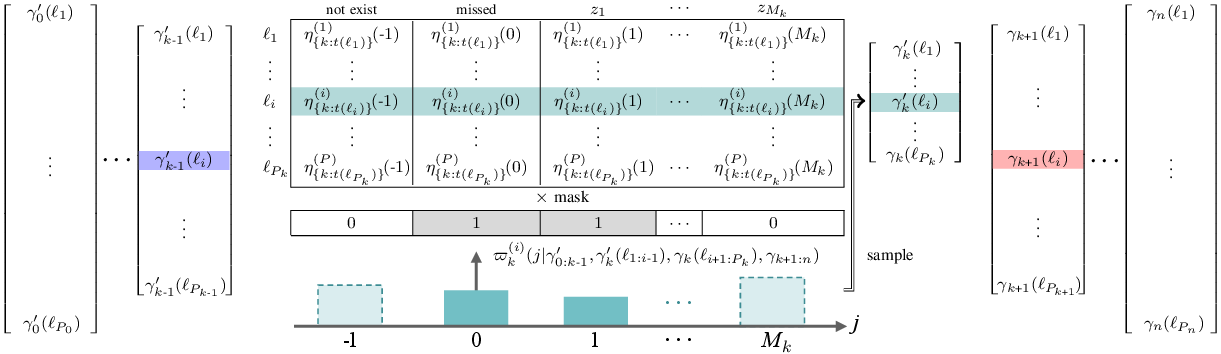}}\vspace{0.0cm}
\par\end{centering}
\caption{Multi-scan association score matrix and conditionals at time $k$.
The $i$-th unnormalized conditional $\varpi_{k}^{(i)}$ is determined
by taking the $i$-th row of the $k$-th association score matrix,
i.e., $\eta_{\{k:t(\ell_{i})\}}^{(i)}(j)$, $j=-1\!:\!M_{k}$, and
set $\eta_{\{k:t(\ell_{i})\}}^{(i)}(j)=0$ for: each positive $j$
associated with a label other than $\ell_{i}$; and any negative $j$,
if $\ell_{i}$ is a live label at time $k+1$.}
\label{fig:multi-scan-gibbs}
\end{figure*}

\subsection{Multi-Scan GLMB Filter Implementation\protect\label{ss:Multi-object_Smoother-1}}

An SGS multi-scan GLMB truncation technique has been developed for
the multi-scan GLMB filter implementation in \cite{VoVomultiscan18}.
Similar to the single-scan counter-part, we only need to consider
$(\theta,I_{k})$ represented by the (extended) association $\gamma_{k}\!:\!I_{k\texttt{-}1}\uplus\mathbb{B}_{k}\!\rightarrow\!\{-1\!:\!\left|Z_{k}\right|\}$
defined in (\ref{eq:gamma}). Note that the domain $\mathcal{D}(\theta)$
is given by $\mathcal{L}(\gamma_{k})\triangleq\{\ell\in I_{k\texttt{-}1}\uplus\mathbb{B}:\gamma_{k}(\ell)\geq0\}$,
called the \emph{live labels} of $\gamma_{k}$.  Assuming all observations
follow the standard model, $(\theta,I_{n})$ and $(\xi,I_{0:n\texttt{-}1})$
can be equivalently represented by $\gamma_{n}$ and $\gamma_{0:n\texttt{-}1}$.
Further, for each $k\in\{0\!:\!n\}$, enumerating $\mathcal{L}(\gamma_{k\texttt{-}1})\uplus\mathbb{B}_{k}=\{\ell_{1:P_{k}}\}$
and $Z_{k}=\{z_{1:M_{k}}\}$, we can represent $\gamma_{k}$ as a
\textit{P}\textsubscript{\textit{k}}-tuple in $\{-1\!:\!M_{k}\}^{P_{k}}$. 

Denoting the weight $w_{Z}^{(\xi,\theta)}(I_{0:n})$ by $\omega_{0:n}(\gamma_{0:n})$
to convey the dependence on $\gamma_{0:n}$, and applying the weight
propagation (\ref{eq:eq:MSGLMBBayesupdate-1}) iteratively to an initial
multi-scan GLMB $\boldsymbol{\pi}_{0}=\{(\omega_{0}(\gamma_{0}),p_{0}^{(\gamma_{0})})\!:\!\gamma_{0}\in\Gamma_{0}\}$
yields
\begin{align}
\!\omega_{0:n}(\gamma_{0:n}) & =\prod\limits_{k=1}^{n}\!\left[\boldsymbol{1}_{\Gamma_{k}}^{\!(\gamma_{k\texttt{-}1})\!}(\gamma_{k})\prod\limits_{i=1\!}^{P_{k}}\eta_{k}^{(i)\!}(\gamma_{k}(\ell_{i}))\right]\!\omega_{0}(\gamma_{0}),\label{eq:MSGLMBCannonical1-3}
\end{align}
where $\boldsymbol{1}_{\Gamma_{k}}^{\!(\gamma_{k\texttt{-}1})\!}(\gamma_{k})\!\triangleq\!\boldsymbol{1}_{\Gamma_{k}\!}(\gamma_{k})\boldsymbol{1}_{\mathcal{F}(\mathcal{L}(\gamma_{k\texttt{-}1})\mathbf{\uplus}\mathbb{B}_{k})\!}(\mathcal{L}(\gamma_{k})_{\!})$,
and $\eta_{k}^{(i)\!}(j)$ is defined in (\ref{eq:eta-1}) with the
time $k$ suppressed (keeping in mind its implicit dependence on $\xi=\gamma_{0:k\texttt{-}1}(\ell_{i})$).
Following the GS strategy, the objective is to generate significant
association histories by sampling from some discrete probability distribution
$\varpi$ on $\cprod_{k=0}^{n}\{-1\!:\!M_{k}\}^{P_{k}}$, preferably
$\varpi=\omega_{0:n}$, so that only valid association histories have
positive probabilities, and those with high weights are more likely
to be chosen. 

\subsubsection{Sequential Factor Sampling}

Decomposing $\varpi$ as
\begin{equation}
\varpi(\gamma_{0:n})=\prod\limits_{k=1}^{n}\varpi_{k}(\gamma_{k}|\gamma_{0:k\texttt{-}1})\varpi_{0}(\gamma_{0}),\label{eq:jointpi}
\end{equation}
and choosing the factors $\varpi_{0}=\omega_{0}$, 
\[
\varpi_{k}(\gamma_{k}|\gamma_{0:k\texttt{-}1})\propto\boldsymbol{1}_{\Gamma_{k}}^{(\gamma_{k\texttt{-}1})}(\gamma_{k})\prod\limits_{i=1\!}^{P_{k}}\eta_{k}^{(i)}(\gamma_{k}(\ell_{i})),
\]
for $k=1\!:\!n$, yields $\varpi=\omega_{0:n}$. Hence, a simple method
to sample from (\ref{eq:MSGLMBCannonical1-3}) is to sample $\gamma_{0}$
from $\varpi_{0}$, and then for $k=1\!:\!n$, sample $\gamma_{k}\!:\!\mathcal{L}(\gamma_{k\texttt{-}1})\uplus\mathbb{B}_{k}\!\rightarrow\!\{-1\!:\!\left|Z_{k}\right|\}$
from $\varpi_{k}(\cdot|\gamma_{0:k\texttt{-}1})$. This sequential
generation of $\gamma_{0:n}$ ensures $\boldsymbol{1}_{\mathcal{F}(\mathcal{L}(\gamma_{k\texttt{-}1})\mathbf{\uplus}\mathbb{B}_{k})\!}(\mathcal{L}(\gamma_{k})_{\!})_{\!}=1$,
and hence
\begin{equation}
\varpi_{k}(\gamma_{k}|\gamma_{0:k\texttt{-}1})\propto\boldsymbol{1}_{\Gamma_{k}}(\gamma_{k})\prod\limits_{i=1\!}^{P_{k}}\eta_{k}^{(i)}(\gamma_{k}(\ell_{i})),\label{eq:conditionals0-1}
\end{equation}
for each $k\in\{1\!:\!n\}$. Sampling from (\ref{eq:conditionals0-1})
can be accomplished using the SGS technique described in Subsection
\ref{sss:SGS-GLMB}. While any $\gamma_{0:n}$ generated by this method
is a valid association history \cite{VoVomultiscan18}, to ensure
that $\gamma_{0:n}$ is distributed according to (\ref{eq:jointpi}),
it is necessary to run the Gibbs sampler for each $k$ long enough
so that $\gamma_{k}$ is distributed according to (\ref{eq:conditionals0-1}).
Nonetheless, sequential factor sampling can be used to generate good
starting points for the Markov chains in full GS. 

\subsubsection{SGS}

Sampling from $\varpi$ via SGS involves constructing a Markov chain
where a new iterate $\gamma_{1:n}^{\prime}$ is generated from $\gamma_{1:n}$
by sampling the coordinates $\gamma_{k}^{\prime}(\ell_{i})$, $k=1\!:\!n$,
$i=1\!:\!P_{k}$ of $\gamma_{1:n}^{\prime}$, from the conditional
distributions $\varpi_{k}^{(i)}$ defined by
\begin{eqnarray}
 &  & \!\!\!\!\!\!\!\!\!\!\!\!\!\!\!\!\!\!\!\!\!\!\!\!\!\!\!\!\varpi_{k}^{(i)}(j|\overset{\text{past}}{\overbrace{\gamma_{0:k\texttt{-}1}^{\prime}}},\overset{\text{current (processed)}}{\overbrace{\gamma_{k}^{\prime}(\ell_{1:i\texttt{-}1})}},\overset{\text{current (unprocessed)}}{\overbrace{\gamma_{k}(\ell_{i\texttt{+}1:P_{k}})}},\overset{\text{future}}{\overbrace{\gamma_{k\texttt{+}1:n}\phantom{(}}})\nonumber \\
\text{ \ } & \!\!\!\!\propto & \!\!\!\varpi(\gamma_{0:k\texttt{-}1}^{\prime},\gamma_{k}^{\prime}(\ell_{1:i\texttt{-}1}),j,\gamma_{k}(\ell_{i\texttt{+}1:P_{k}}),\gamma_{k\texttt{+}1:n}).\label{eq:conditional-full-Gibbs}
\end{eqnarray}
Similar to the single-scan case, the conditional $\varpi_{k}^{(i)}$
is determined from the association score matrices at times $k$ and
$k+1$, see Fig.~\ref{fig:multi-scan-gibbs}, which can be pre-computed
as follows. 

Recall from (\ref{eq:eta-1}) that $\eta_{k}^{(i)\!}(j_{k})$ depends
on $\xi$, specifically, the indices $j_{s(\ell_{i}):k\texttt{-}1}$ of the
detections associated with $\ell_{i}$ up to time $k-1$ (since all
observations follow the standard model). To express this dependence
explicitly we write $\eta_{k}^{(i)\!}(j_{k})$ as $\eta_{k}^{(i)\!}(j_{s(\ell_{i}):k})$.
The $(i,j)$-th entry of the multi-scan association score matrix for
time $k$ is given by
\begin{align*}
\eta_{\{k:t(\ell_{i})\}}^{(i)}(j) & \triangleq\prod\limits_{m=k}^{t(\ell_{i})}\eta_{m}^{(i)}(\gamma_{0:k\texttt{-}1}^{\prime}(\ell_{i}),j,\gamma_{k\texttt{+}1:m}(\ell_{i})){\color{green}\mathpunct{\color{black},}}\label{eq:eta-multiscan}
\end{align*}
where $t(\ell_{i})$ is the latest time $\ell_{i}$ exists on $\{0\!:\!n\}$,
and by convention, $\eta_{k}^{(i)\!}(\gamma_{0:k\texttt{-}1}^{\prime}(\ell_{i}),j,\gamma_{k\texttt{+}1:k}(\ell_{i}))=\eta_{k}^{(i)\!}(\gamma_{0:k\texttt{-}1}^{\prime}(\ell_{i}),j)$.

It was shown in \cite{VoVomultiscan18} that to arrive at the target
distribution (\ref{eq:MSGLMBCannonical1-3}), the unnormalized conditional
$\varpi_{k}^{(i)}$ is given by the $\!i$-th row of the $k$-th association
score matrix after zeroing the entry of: every positive $j$ that
has been paired with a label other than $\ell_{i}$ at time $k$ (like
the single-scan case); and any negative $j<\gamma_{k\texttt{+}1}(\ell_{i})$
(a surviving label at time $k+1$ must be live at time $k,$ i.e.,
$\gamma_{k}^{\prime}(\ell_{i})>-1$, because the standard multi-object
dynamic model only admits unfragmented trajectories), see Fig.~\ref{fig:multi-scan-gibbs}
for illustration. This simple result admits a tractable SGS multi-scan
GLMB truncation algorithm for the GLMB smoother implementation in
\cite{VoVomultiscan18}.

This multi-scan GLMB truncation technique has been extended to multiple
sensors in \cite{MVVS2022Multi}, and demonstrated on a multi-object
smoothing problem with 100 scans and 4 sensors, which requires solving
400-dimensional ranked assignment problems with approximately 10 variables
in each dimension. Moving window-based implementations can achieve
a fixed $\mathcal{O}(LTVP^{2}M)$ complexity per time step using SGS
\cite{MVVS2022Multi}, or $\mathcal{O}(LTV(P+M))$ using TGS, where
$L$ is the length of the smoothing window, $T$ is the number of
iterates of the Gibbs sampler, $V$ is the number of sensors, $P$
is the number of hypothesized trajectories, and $M$ is the maximum
number of detections per sensor.

\subsection{Multi-Object Estimation with Non-Standard Models\protect\label{ss:Non-Standard-model}}

The standard multi-object dynamic model assumes conditionally independent
individual transitions, and is sufficient to cover sophisticated single-object
motion, e.g., motion with road constraints \cite{YCS2023Road}, or
multi-modal motion \cite{ReuterMMLMB15,Punchihewa16,YiJiangHoseinnzhad17,LH2019Labeled,WDZSC2019Fast,Cao-SP-21}.
However, it is not adequate to capture inter-object correlations,
which requires more general multi-object transition, e.g., interacting
objects \cite{Gostar-Interacting-19,IKBPH2022Interaction}, group
targets \cite{LC2019Resolvable,XHXML2021Resolvable,LLH2023Labeled},
and spawning/dividing objects\textit{ }\cite{Bryantetal18}, \cite{Nguyenetal-CellSpawning21}.
In the latter, we also seek the ancestries of the objects, which is
easily accomplished in the LRFS formulation by augmenting the parents'
labels into the object's labels, see Fig. \ref{fig:cell-div-model}.
Consequently, the multi-object filtering density is no longer a GLMB,
and the exact solution is intractable. M-GLMB approximations (see
Subsection \ref{subsec:LRFS-approx}) are used to capture inter-object
correlations and ancestries in the multi-object trajectory \cite{Bryantetal18},
\cite{Nguyenetal-CellSpawning21}. An approximate GLMB filter that
accommodates objects exiting the state space and reappearing later
(along with a non-standard observation model) was developed in \cite{VNVJ2024Track}.

While the standard multi-object observation likelihood is general
enough to accommodate a range of noise models, e.g., \cite{Dong-AES-19,DJLSW2018Student,ZYLL2018Heavy,Ong-TASLP-21,OVNVMS2022Audio},
it does not cover extended object observations \cite{Beardetal-extended16},
\cite{LLLG2020Improved}, merged/occluded observations \cite{Beardetal14},
\cite{Ong-TPAMI-20}, and image observations \cite{PapiKim-15,G2016Track,MahlerSuperPS18},
and \cite{SaucanPFSP-GLMB17}. 

Under the extended observation model that allows an object to generate
multiple detections, the GLMB filter and smoother are still exact
solutions \cite{Beardetal-extended16}, \cite{GBR2017Extended}, which
can be implemented with Gaussians (and Gaussian mixture) or particle
methods \cite{CJZ2023Multiple}. Extended objects are modeled by Gaussian-Gamma-Inverse-Wishart
distributions in \cite{Beardetal-extended16}, B-splines in \cite{Daniyanetal18},
and multiple ellipses in \cite{LLLG2020Improved}.

The standard observation model cannot accommodate multiple objects
sharing detections due to occlusions. In this case, the detection
probability of an object depends on the entire multi-object state,
and requires a non-standard observation model \cite{Beardetal14},
\cite{Ong-TPAMI-20}, \cite{VNVJ2024Track}. The resultant multi-object
filtering density is not a GLMB, and M-GLMB or LMB approximations
have been used in \cite{Beardetal14}, \cite{Saucan-MergeM-20}. In
\cite{Ong-TPAMI-20,VNVJ2024Track}, efficient GLMB approximations
based on the predicted states, were used to update the individual
detection probabilities to account for occlusions, while in \cite{Dai-TMech-20}
the dependence of the detection probability on other states are marginalized
out.
\begin{figure}[t]
\begin{centering}
\resizebox{88mm}{!}{\includegraphics[clip]{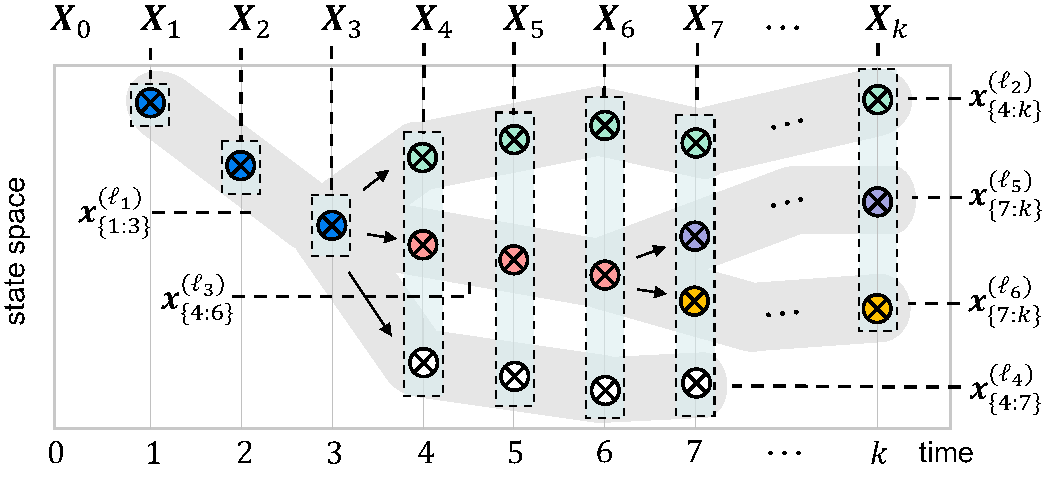}}
\par\end{centering}
\caption{Labels capture ancestries of object in spawnings. A daughter's label
$(\ell,k,\iota)$ consists of the parent label $\ell$, the time of
birth $k$, and an individual identity $\iota$ to distinguish it
amongst the siblings. At time $4$, $\ell_{1}$ spawns 3 daughters
with labels $\ell_{2}=(\ell_{1},4,1)$, $\ell_{3}=(\ell_{1},4,2)$,
and $\ell_{4}=(\ell_{1},4,3)$. At time $7$, $\ell_{3}$ spawns $\ell_{5}=(\ell_{3},7,1)$
and $\ell_{6}=(\ell_{3},7,2)$.}
\label{fig:cell-div-model}\vspace{-0cm}
\end{figure}

In applications where the multi-object state generates a single observation,
inter-object correlations are introduced through the update into the
multi-object filtering density which no longer takes on a GLMB form.
The particle implementations for a general observation model proposed
in \cite{PapiKim-15,G2016Track} are computationally demanding, and
the M-GLMB approximation has been developed in \cite{PapiVoVoetal15}.
Efficient approximations exploiting the additivity of superpositional
observations are also proposed in \cite{SaucanPFSP-GLMB17,MahlerSuperPS18},
\cite{DZSZ2023Modified}. This approach was combined with pseudo-smoothing
to address acoustic vector sensor observations in \cite{ZBYC2023DOA}.
A parallelizable LMB approximation is developed in \cite{LiYiHoseinnezhadetal18},\textit{
}while GLMB filtering and smoothing approximations are proposed in
\cite{CZ2022Generalized,CZ2022Multi}, \cite{YY2020Track} (the former
also covers extended object observations). An LMB filter for pre-clustered
laser range finder image data is proposed in \cite{Dai-VT-10}. A
salient advantage of LRFS in MOT over legacy approaches such as MHT,
is the seamless operation with different observation types such as
detections and images. Indeed, a GLMB filter has been applied to
an observation model that switches between detections and image observation
in \cite{KimVoVo19}. Non-standard models have also been developed
for applications with unknown model parameters. These are included
in the discussion on robust multi-object estimation.

\subsection{Robust Multi-Object Estimation\protect\label{ss:Robust-Est}}

Multi-object filtering solutions have also been developed to address
unknown multi-object model parameters. In \cite{Punchihewa-18},
the unknown detection probability is jointly estimated by augmenting
it to the attribute state, while the unknown clutter rate is addressed
by treating clutter as a different class of objects \cite{MahlerClutter18}.
This formulation results in a GLMB filter with jump Markov single-object
dynamic, which incurs extra computational complexity \cite{Punchihewa-18}.
A more efficient suboptimal alternative is the bootstrapping approach
that uses robust CPHD or multi-Bernoulli filters \cite{Mahler11,VVHM2013Robust}
to estimate the detection probability and clutter rate from the observations,
and feed them to the GLMB filter \cite{Do-Adaptive-20,Do-Robust-21},
\cite{Kim-IET-18}. In \cite{ZLYX2019Robust,HLTXD2021Robust,LWLGZ2020Robust,Scheel-ITS-19},
unknown observation noise parameters is addressed by augmenting the
object attributes with noise covariance matrices to be jointly estimated
by integrating Variational Bayes techniques to the multi-object filter. 

Unknown birth model parameters can be addressed by the measurement-driven
birth approach \cite{Reuter2014}, where current detections are used
to predict the birth parameters at the next time. GLMB filtering with
measurement-driven birth is popular for MOT with unknown birth models
\cite{LeGrand18,Lin-Speech-19,CGYZ2020Modeling}. In \cite{Zhu-DSP-21},
a measurement-driven birth model for interval observations was developed,
while in \cite{HBWR2020Detection}, the Rauch-Tung-Striebel smoother
was used to improve the multi-object birth density estimation. An
admissible region birth model for space objects was proposed in \cite{Gaebler-BirthR-20}.
In \cite{Trezza-SP22}, birth parameter estimation was developed for
multi-sensor GLMB filtering, where GS is used to select probable detection
combinations. A fully robust GLMB filter that combines unknown birth
and sensor parameters estimation techniques were developed in \cite{NDN2022Adaptive}.
Around the same time an interesting formulation using outer probability
measure, which allows GLMB filtering in the absence of model parameters
was developed in \cite{Cai-etal-22}. 

\subsection{Distributed Multi-Object Estimation\protect\label{ss:Distributed-Est} }

Encapsulating all information on the multi-object trajectories in
the filtering/posterior densities allows density fusion. In \cite{Fantacci18},
the Generalized Covariance Intersection (GCI) fusion rule \cite{Mahler2000},
was applied to fuse the local labeled multi-object densities $\boldsymbol{\pi}^{(s)}$,
$s\in\{1\!:\!S\}$ into 
\[
\bar{\boldsymbol{\pi}}_{GCI}=\arg\min_{\boldsymbol{\pi}}{\textstyle \sum_{s\in S}}\omega^{(s)}D_{KL}(\boldsymbol{\pi}||\boldsymbol{\pi}^{(s)}),
\]
where each $\omega^{(s)}$ is a user-defined fusion weight. It was
shown that the M-GLMBs or LMBs, are closed under GCI fusion \cite{Fantacci18}.
For LMB local densities, a technique for determining the weights $\omega^{(s)}$
of each of their component was proposed in \cite{Wang18}, using the
information gain at the local update step. The fusion weights can
also be optimized to fuse GLMBs from sensor nodes with different field
of view (FoV) \cite{JL2023Discrepant}. In \cite{SZDJL2022Consensus},
the event-trigger strategy was proposed to reduce the communication
burden in GCI fusion. 

GCI fusion is sensitive to mismatches in the labels due to uncertainty
in the object birth times \cite{LiYi18}. To avoid this mismatch,
GCI was used to fuse  multi-Bernoulli approximations of the unlabeled
versions of the local labeled multi-object densities, and then augment
the labels to the fused multi-object density \cite{LiYi18}. In \cite{Li19},
the label mismatch was quantified by the minimum KL-D between the
unlabeled versions of the local labeled densities. Further, it was
shown that for local LMB densities, minimizing the label mismatch
translates to solving a linear assignment problem. For sensors with
different FoVs, it was proposed in \cite{SDJL2021Consensus} to maintain
a list of label matches and extend the label space of each node to
include labels of objects that cannot be observed by the local sensor.
In \cite{LYXY2023Optimizing}, an additional group variable was introduced
to record matching information among sensors, and fusion was performed
separately on the surviving and newborn objects. 

Exchanging the arguments of the KL-D in GCI fusion results in Minimum
Information Loss (MIL) fusion \cite{Gao-TSP-20}:
\[
\bar{\boldsymbol{\pi}}_{MIL}=\arg\min_{\boldsymbol{\pi}}{\textstyle \sum_{s\in S}}\omega^{(s)}D_{KL}(\boldsymbol{\pi}^{(s)}||\boldsymbol{\pi}).
\]
Unlike GCI fusion, MIL fusion preserves both common and exclusive
information of the local densities. MIL fusion formulae for M-GLMBs
and LMBs that restrict the fused densities to the same family as the
local densities were derived in \cite{Gao-TSP-20}. The label matching
strategy in \cite{Li19} can be used to ensure all local densities
have the same label space.  For sensors with different FoVs, the
local densities can be decomposed into sub-densities with minimum
KL-D from the original ones, and then fused together using MIL \cite{GBC2022Fusion}.
GCI and MIL fusion were combined to exploit their respective advantages
in \cite{Gao-TAES-21,LBCGW2022Distributed}, and to enhance resilient
to cyber-attack in peer-to-peer sensor network in \cite{GBC2023Resilient}.
GCI and MIL fusion are considered as geometric averaging (of the multi-object
densities) approaches. Fusion via arithmetic averaging has also been
proposed as a versatile alternative to geometric averaging in \cite{L2023Arithmetic,LYDF2023Arithmetic}.

Apart from the KL-D based fusion discussed above, there are also other
fusion solutions. In a centralized setting, local densities can be
approximated by Poissons with matching first moments, which are then
fused together using CS-D based fusion \cite{Gostar-SP-20}. In \cite{Gostar-TSP-21},
the authors proposed to fuse the LMB components (from different nodes)
corresponding to a particular label into a single component. In \cite{Herrmann-TSP-21},
Bayes parallel combination rule was used to compute the global posterior
density from synchronous centralized sensors. In \cite{SV2018Distributed},
instead of fusing the local GLMBs, distributed cross-entropy and average
consensus were used to sample multi-sensor assignments with high scores.
An approach that directly fuses the local multi-object estimates was
proposed in \cite{NRVR2021Distributed}, which incurs far lower complexity
than density fusion and can be efficiently extended to multiple scans.
Further, since only the estimates are fused, this fusion algorithm
is not restricted to any tracking algorithms or approaches.

\subsection{Multi-Object Control\protect\label{ss:Control} }

In multi-object control, we seek control signals to drive the multi-object
state/trajectory by optimizing the cost/reward function, subject to
a set of constraints. Cost functions are usually based on information
divergence. Due to its closed-form for the GLMB filtering density,
the CS-D was used as the cost function for a sensor control problem
and void probability constraints in \cite{BVVA2015Sensor,BVVA2017Void}.
Moreover, a sensor management framework for the GLMB filter based
on the CS-D was developed in \cite{MahlerMTMSM19}. A passive sensor
management solution was also proposed for the GLMB filter in \cite{ZhuPOMDP22}.
While the CS-D can be computed analytically for GLMBs, it is still
an expensive task. A cheaper alternative was proposed in \cite{GHRWB2017Constrained}
by using the LMB filter with a cost function based on the CS-D between
the Poisson approximations of the LMBs, which can be computed efficiently
\cite{HVVM2015Cauchy} (see also Subsection \ref{subsec:Info-Div}).
In \cite{Zhang2sensorcontrol18}, R\'{e}nyi-D was also suggested as the
reward function with the GLMB filter. The simple closed-form information
divergence for LMBs in Subsection \ref{subsec:Info-Div} could be
used to accelerate computation of information-based reward functions. 

For multi-sensor tasking, a R\'{e}nyi-D reward was used in \cite{Cai-SSA-19},
while in \cite{WHGRXB2018Multi}, a task-driven reward function with
the LMB filter based on GCI fusion was proposed. In \cite{Panicker-SP-20},
maximizing the task-driven reward was replaced by minimizing the ratio
of non-existence to existence probability for a subset of objects
of interest in the LMB. A solution to this problem with CS-D as the
reward function was given in \cite{Ishtiaq-CSdiv-21}. Methods to
control drones based on R\'{e}nyi-D and CS-D were proposed in \cite{Nguyenetal-PP-19}.
Inspired by the multi-objective formulation for the multi-object localization
problem \cite{ZWL2019Multi}, in \cite{NVVRR2022Multi}, a multi-objective
reward function for the two conflicting tasks of tracking and discovery
with limited FoV sensors was proposed. Using mutual information and
differential entropy objectives, it was also established that the
multi-objective reward is monotone and sub-modular, which allows the
optimal control action to be efficiently computed via greedy algorithms.
Metric based reward functions have also been proposed. In \cite{GostarSensorMan17}
the dispersion of the labeled multi-object density about its statistical
mean using the OSPA distance was used to maximize the accuracy of
the estimates. Similarly, techniques for multi-sensor control/selection
that optimize a metric-based objective function were also proposed
in \cite{LHLH2018Constrained,LHWH2018Sensor}. Further, simultaneous
localization and mapping can also be performed with the LRFS filters
\cite{Deusch15,MoratuwageSLAM-Sensor-19}.

\section{Conclusions\protect\label{s:Conclusion}}

The notion of a State-Space Model (SSM) has been the cornerstone for
modern estimation and control theory, in which the celebrated Kalman
filter is often recognized as one of the seminal contributions. A
multi-object SSM generalizes the state vector to a finite set of vectors
for modeling multi-object systems, wherein both the number of objects
and their individual states are unknown and dynamically varying. Moreover,
to retain the fundamental premise that the system trajectory is given
by the history of the system states, it is necessary to use the labeled
multi-object representation of the system state. Analogous to how
information on a single trajectory is encapsulated in the filtering/posterior
density of the state vectors, a Labeled Random Finite Set (LRFS) enables
information on the multi-object trajectory to be encapsulated in the
multi-object filtering/posterior density.

This article has provided an overview of key developments in the
LRFS approach to multi-object SSM. LRFS resolves the theoretical
drawbacks of the unlabeled RFS formulation for multi-object SSMs,
and offers important advancements. In particular, the LRFS approach: 
\begin{itemize}
\item Enables principled multi-object trajectory estimation with bounded
complexity per time step (without resorting to post-processing heuristics),
resolving an important conceptual shortcoming of unlabeled RFS;
\item Provides statistical characterization of uncertainty for the underlying
multi-object trajectory ensemble, offering new conceptual tools and
applications;
\item Naturally covers trajectory crossings, fragmentations, and lineages
(in spawning objects), as well as seamless operations with different
types of observations simultaneously or at different times, thereby
resolving the theoretical problem concerning multiple objects occupying
the same attribute state and further broadening the application base;
\item Admits the concept of joint existence probability for the elements
of a multi-object state, and hence, a meaningful notion of the most
probable multi-object state(s), fundamental to multi-object estimation;
\item Enables principled approximations of multi-object densities (with
characterizable approximation errors) and reduces the standard multi-object
transition density from a combinatorial sum to a single term, offering
new analytical tools as well as facilitating efficient solutions,
which are scalable in the numbers of objects, detections, sensors
and scans.
\end{itemize}
Multi-object SSM is not only rich in theory, but also in mathematical
and computational tools. While the majority of research in the field
has focused on multi-object estimation, the related areas of system
identification and control remain largely unexplored. Within the estimation
problem there are nonetheless many open theoretical and computational
challenges, such as Cram\'{e}r-Rao like performance bounds for multi-object
estimation \cite{C2022Cramer}, more sophisticated multi-object models
and efficient accompanying solutions.

\bibliographystyle{mod_ieeetr}
\bibliography{reflib}
\cleardoublepage{}

\section*{Supplementary Materials\smallskip{}
}

This supplementary section presents the proofs and additional discussions
of the divergence results for Labeled Multi-Bernoullis (LMBs) in Subsection
\ref{subsec:Info-Div}. We recall the notation $\left\langle \boldsymbol{f}\right\rangle \!(\ell)\triangleq\int\boldsymbol{f}(x,\ell)dx$,
and for each LMB parameterized by $\{(r_{i}(\ell),p_{i}(\cdot,\ell))\!:\!\ell\in\mathcal{D}(\sigma_{i})\}$,
with $r_{i}<1$, $i=1,2$ and $\mathcal{D}(\sigma)$ denoting the
domain of $\sigma$. Moreover, we use the multi-object density of
the form $\boldsymbol{\pi}_{i}(\boldsymbol{X})=\Delta(\boldsymbol{X})K_{i}\boldsymbol{f}_{\!i}^{\boldsymbol{X}}$,
where $K_{i}=\tilde{r}_{i}^{\mathbb{L}}$, $\tilde{r}_{i}\triangleq1-r_{i}$,
and $\boldsymbol{f}_{\!i}\triangleq r_{i}p_{i}/\tilde{r}_{i}$. Note
that $\left\langle \boldsymbol{f}_{\!i}\right\rangle \!(\ell)=r_{i}(\ell)/\tilde{r}_{i}(\ell)$,
$\tilde{r}_{i}(\ell)=\left(1+\left\langle \boldsymbol{f}_{i}\right\rangle \!(\ell)\right)^{-1}$,
and $\tilde{r}_{i}^{\mathbb{L}}=\tilde{r}_{i}^{\mathcal{D}(\sigma_{i})}$,
because $r_{i}(\ell)=0$, i.e., $\tilde{r}_{i}(\ell)=1$ for $\ell\notin\mathcal{D}(\sigma_{i})$. 

\subsection{Preliminary Lemmas}

The following lemmas facilitate the arguments in the proofs.
\begin{lem}
For functions $g,h$ defined on a finite set $S$,
\begin{align*}
\sum_{L\subseteq S}g^{L} & h^{S-L}=\left(h+g\right){}^{S}.
\end{align*}
\end{lem}
\begin{proof}
Let $S=\{\ell_{1},...,\ell_{n}\}$, $f=g/h$, and recall Vieta's formula
for elementary symmetric functions $e_{i}\left(f(\ell_{1}),...,f(\ell_{n})\right)$:
\[
\sum_{i=0}^{n}e_{i}\left(f(\ell_{1}),...,f(\ell_{n})\right)\left(-1\right){}^{i}z^{n-i}=\prod_{j=1}^{n}\left(z-f(\ell_{j})\right),
\]
Substituting $z=-1$ gives
\begin{align*}
\left(-1\right){}^{n}\sum_{i=0}^{n}e_{i}\left(f(\ell_{1}),...,f(\ell_{n})\right) & =\prod_{j=1}^{n}\left(-1-f(\ell_{j})\right)\\
\sum_{i=0}^{n}e_{i}\left(f(\ell_{1}),...,f(\ell_{n})\right) & =\prod_{j=1}^{n}\left(1+f(\ell_{j})\right).
\end{align*}
Noting that 
\[
e_{i}\left(f(\ell_{1}),...,f(\ell_{n})\right)\triangleq\sum_{L\subseteq S,|L|=i}f{}^{L},
\]
we have
\begin{align*}
\sum_{L\subseteq S}f^{L} & =\sum_{i=0}^{n}e_{i}\left(f(\ell_{1}),...,f(\ell_{n})\right)\\
 & =\prod_{j=1}^{n}\left(1+f(\ell_{j})\right)=\left(1+f\right){}^{S}.
\end{align*}
Substituting $f=g/h$, and multiplying both sides by $h^{S}$ gives
\begin{align*}
\sum_{L\subseteq S}\left(\frac{g}{h}\right)^{L}h^{S} & =\left(1+\frac{g}{h}\right){}^{S}h^{S}.
\end{align*}
Hence, $\sum_{L\subseteq S}g^{L}h^{S-L}=\left(h+g\right){}^{S}.$\smallskip{}
\end{proof}
\begin{lem}
For $\boldsymbol{f}:\mathbb{X}\times\mathbb{L}\rightarrow\mathbb{R}$
integrable on $\mathbb{X}$,
\[
\int\Delta(\boldsymbol{X})\boldsymbol{f}^{\boldsymbol{X}}\delta\boldsymbol{X}=\sum_{L\subseteq\mathbb{L}}\left\langle \boldsymbol{f}\right\rangle ^{L}=\left(1+\left\langle \boldsymbol{f}\right\rangle \right){}^{\mathbb{L}}.
\]
\end{lem}
\begin{proof}
The result follows by applying Lemma 3 in \cite{VoConj13} to the
set integral, and then Lemma 1.
\end{proof}
\begin{lem}
For $f,g:\mathbb{X}\times\mathbb{L}\rightarrow\mathbb{R}$ integrable
on $\mathbb{X}$, with $g$ unitless
\[
\int\Delta(\boldsymbol{X})\boldsymbol{f}^{\boldsymbol{X}}\ln\boldsymbol{g}^{\boldsymbol{X}}\delta\boldsymbol{X}=\sum_{L\subseteq\mathbb{L}}\sum_{\ell\in L}\left\langle \boldsymbol{f}\right\rangle ^{L-\{\ell\}}\left\langle \boldsymbol{f}\ln\boldsymbol{g}\right\rangle (\ell).
\]
\end{lem}
The proof of this lemma is given in the Appendix of \cite{NVVRR2022Multi}.\medskip{}

\begin{lem}
For $f:\mathcal{F}(\mathbb{L})\rightarrow\mathbb{R}$, $g:\mathbb{L}\rightarrow\mathbb{R}$,
and $S\subseteq\mathbb{L}$,
\[
\sum_{L\subseteq S}f(L)\sum_{\ell\in L}g(\ell)=\sum_{\ell\in\mathbb{L}}g(\ell)\sum_{L\subseteq S-\{\ell\}}f(L\cup\{\ell\}).
\]
\end{lem}
The proof of this lemma is given in the Appendix of \cite{NVVRR2022Multi}.\medskip{}

\subsection{Proof of R\'{e}nyi Divergence (\ref{eq:Renyi-LMB}) \smallskip{}
}
\begin{proof}
The R\'{e}nyi Divergence (R\'{e}nyi-D) between the LMBs $\boldsymbol{\pi}_{1}$
and $\boldsymbol{\pi}_{2}$ is given by
\begin{align*}
 & \!\!\!\!\!\!D_{R}\!\left(\boldsymbol{\pi}_{1}||\boldsymbol{\pi}_{2}\right)\\
 & =\frac{1}{\alpha-1}\ln\int\Delta(\boldsymbol{X})\left(K_{1}\boldsymbol{f}_{1}^{\boldsymbol{X}}\right)^{\alpha}\left(K_{2}\boldsymbol{f}_{2}^{\boldsymbol{X}}\right)^{1-\alpha}\delta\boldsymbol{X}\\
 & =\frac{1}{\alpha-1}\ln\left(K_{1}^{\alpha}K_{2}^{1-\alpha}\int\!\Delta(\boldsymbol{X})\left(\boldsymbol{f}_{1}^{\alpha}\right)^{\boldsymbol{X}}\left(\boldsymbol{f}_{2}^{1-\alpha}\right)^{\boldsymbol{X}}\delta\boldsymbol{X}\right)\\
 & =\frac{1}{\alpha-1}\ln\left(\left(\tilde{r}_{1}^{\mathbb{L}}\right)^{\!\alpha}\!\left(\tilde{r}_{2}^{\mathbb{L}}\right)^{\!1-\alpha}\int\!\Delta(\boldsymbol{X})\left(\boldsymbol{f}_{1}^{\alpha}\boldsymbol{f}_{2}^{1-\alpha}\right)^{\boldsymbol{X}}\delta\boldsymbol{X}\right)\\
 & =\frac{1}{\alpha-1}\ln\left[\left(\tilde{r}_{1}^{\mathbb{\alpha}}\tilde{r}_{2}^{1-\alpha}\right)^{\mathbb{L}}\left(1+\left\langle \boldsymbol{f}_{1}^{\alpha}\boldsymbol{f}_{2}^{1-\alpha}\right\rangle \right)^{\mathbb{L}}\right]\\
 & =\sum_{\ell\in\mathbb{L}}\cfrac[l]{\ln\left[\left(\tilde{r}_{1}^{\mathbb{\alpha}}\tilde{r}_{2}^{1\!-\!\alpha}\right)\!(\ell)\left(1+\left\langle \boldsymbol{f}_{1}^{\alpha}\boldsymbol{f}_{2}^{1-\alpha}\right\rangle \!(\ell)\right)\right]}{\alpha-1}\\
 & =\sum_{\ell\in\mathbb{L}}\!\cfrac[l]{\ln\!\left[\left(\tilde{r}_{1}^{\mathbb{\alpha}}\tilde{r}_{2}^{1\!-\!\alpha}\right)\!(\ell)\right]+\ln\!\left[1\!+\!\left\langle \boldsymbol{f}_{1}^{\alpha}\boldsymbol{f}_{2}^{1\!-\!\alpha}\right\rangle \!(\ell)\right]}{\alpha-1},
\end{align*}
where the 3rd last line follows from Lemma 2. Note that $\tilde{r}_{1}^{\mathbb{\alpha}}\tilde{r}_{2}^{1-\alpha}>0$,
$\left\langle \boldsymbol{f}_{\!1}^{\alpha}\boldsymbol{f}_{\!2}^{1\!-\!\alpha}\right\rangle (\ell)\geq0$
and hence each term of the sum is well-defined. \smallskip{}
\end{proof}
\textit{Remark}: Since $\tilde{r}_{i}(\ell)=\left(1+\left\langle \boldsymbol{f}_{i}\right\rangle \!(\ell)\right)^{-1}$,
the above R\'{e}nyi-D can be written completely in terms of $\boldsymbol{f}_{1}$
and $\boldsymbol{f}_{2}$ as 
\begin{align*}
D_{R}\!\left(\boldsymbol{\pi}_{1}||\boldsymbol{\pi}_{2}\right) & =\sum_{\ell\in\mathbb{L}}\cfrac[l]{\ln\left[1\!+\!\left\langle \boldsymbol{f}_{1}^{\alpha}\boldsymbol{f}_{2}^{1\!-\!\alpha}\right\rangle \!(\ell)\right]}{\alpha-1}\\
 & -\!\sum_{\ell\in\mathbb{L}}\cfrac[l]{\alpha\ln\left[1\!+\!\left\langle \boldsymbol{f}_{\!1}\right\rangle \!(\ell)\right]\!+\!(1\!-\!\alpha)\ln\!\left[1\!+\!\left\langle \boldsymbol{f}_{\!2}\right\rangle \!(\ell)\right]}{\alpha-1}.
\end{align*}
Noting that $\boldsymbol{f}_{i}=r_{i}p_{i}/\tilde{r}_{i}$, and hence
\begin{align*}
\left\langle \boldsymbol{f}_{1}^{\alpha}\boldsymbol{f}_{2}^{1\!-\!\alpha}\right\rangle \!(\ell) & =\frac{r_{1}^{\alpha}\!(\ell)r_{2}^{1-\alpha}\!(\ell)}{\tilde{r}_{1}^{\alpha}\!(\ell)\tilde{r}_{2}^{1-\alpha}\!(\ell)}\left\langle p_{1}^{\alpha}p_{2}^{1-\alpha}\right\rangle \!(\ell),
\end{align*}
we can write the R\'{e}nyi-D in terms of the LMB parameters
\begin{align*}
D_{R}\!\left(\boldsymbol{\pi}_{1}||\boldsymbol{\pi}_{2}\right)= & \sum_{\ell\in\mathbb{L}}\cfrac[l]{\ln\left[\tilde{r}_{1}^{\alpha}\tilde{r}_{2}^{1-\alpha}(\ell)+r_{1}^{\alpha}r_{2}^{1-\alpha}(\ell)\!\left\langle p_{1}^{\alpha}p_{2}^{1-\alpha}\right\rangle \!(\ell)\right]}{\alpha-1}.
\end{align*}
\medskip{}
It also can be written in terms of the Probability Hypothesis Density
(PHD) $\boldsymbol{v}_{i}=r_{i}p_{i}=\tilde{r}_{i}\boldsymbol{f}_{i}$,
\begin{align*}
D_{R}\!\left(\boldsymbol{\pi}_{1}||\boldsymbol{\pi}_{2}\right) & =\sum_{\ell\in\mathbb{L}}\cfrac[l]{\ln\left[1\!+\!\frac{\left\langle \boldsymbol{v}_{1}^{\alpha}\boldsymbol{v}_{2}^{1-\alpha}\right\rangle \!(\ell)}{\left(1-\left\langle \boldsymbol{v}_{1}\right\rangle \!(\ell)\right)^{\alpha}\!\left(1-\left\langle \boldsymbol{v}_{2}\right\rangle \!(\ell)\right)^{1-\alpha}}\right]}{\alpha-1}\\
 & +\sum_{\ell\in\mathbb{L}}\frac{\alpha\ln\left[1\!-\!\left\langle \boldsymbol{v}_{1}\right\rangle \!(\ell)\right]\!+\!(1-\alpha)\ln\left[1\!-\!\left\langle \boldsymbol{v}_{2}\right\rangle \!(\ell)\right]}{\alpha-1},
\end{align*}
since $\tilde{r}_{i}(\ell)=1-\left\langle \boldsymbol{v}_{i}\right\rangle \!(\ell)$,
and $\boldsymbol{f}_{i}=\boldsymbol{v}_{i}/(1-\left\langle \boldsymbol{v}_{i}\right\rangle )$.

\subsection{Proof of Kullback-Leibler Divergence (\ref{eq:KL-LMB}) }
\begin{proof}
The Kullback-Leibler Divergence (KL-D) between the LMBs $\boldsymbol{\pi}_{1}$
and $\boldsymbol{\pi}_{2}$ is given by
\begin{align*}
D_{KL}\!\left(\boldsymbol{\pi}_{1}||\boldsymbol{\pi}_{2}\right)= & \:\int\Delta(\boldsymbol{X})K_{1}\boldsymbol{f}_{1}^{\boldsymbol{X}}\ln\frac{K_{1}\boldsymbol{f}_{1}^{\boldsymbol{X}}}{K_{2}\boldsymbol{f}_{2}^{\boldsymbol{X}}}\delta\boldsymbol{X}\\
= & \:\int\Delta(\boldsymbol{X})K_{1}\boldsymbol{f}_{1}^{\boldsymbol{X}}\left(\ln\frac{K_{1}}{K_{2}}+\ln\frac{\boldsymbol{f}_{1}^{\boldsymbol{X}}}{\boldsymbol{f}_{2}^{\boldsymbol{X}}}\right)\!\delta\boldsymbol{X}\\
= & \:\ln\frac{K_{1}}{K_{2}}\int\Delta(\boldsymbol{X})K_{1}\boldsymbol{f}_{1}^{\boldsymbol{X}}\delta\boldsymbol{X}\\
+ & \:K_{1}\int\Delta(\boldsymbol{X})\boldsymbol{f}_{1}^{\boldsymbol{X}}\ln\left(\frac{\boldsymbol{f}_{1}}{\boldsymbol{f}_{2}}\right)^{\!\boldsymbol{X}}\!\!\delta\boldsymbol{X}\\
= & \:\ln\frac{K_{1}}{K_{2}}+K_{1}\!\int\Delta(\boldsymbol{X})\boldsymbol{f}_{1}^{\boldsymbol{X}}\ln\left(\frac{\boldsymbol{f}_{1}}{\boldsymbol{f}_{2}}\right)^{\!\boldsymbol{X}}\!\!\delta\boldsymbol{X}\\
= & \:\ln\frac{K_{1}}{K_{2}}+K_{1}\!\sum_{L\subseteq\mathbb{L}}\sum_{\ell\in L}\frac{r_{1}^{L-\{\ell\}}}{\tilde{r}_{1}^{L-\{\ell\}}}\!\left\langle \!\boldsymbol{f}_{\!1}\!\ln\frac{\boldsymbol{f}_{\!1}}{\boldsymbol{f}_{\!2}}\!\right\rangle \!(\ell),
\end{align*}
where the last line follows from Lemma 3 and $\left\langle \boldsymbol{f}_{1}\right\rangle =r_{1}/\tilde{r}_{1}$.
Since $\boldsymbol{f}_{1}(\cdot,\ell)=r_{1}(\ell)=0$ for $\ell\notin\mathcal{D}(\sigma_{1})$,
and $0\ln0=0$ by convention, we can exchange $\mathbb{L}$ and $\mathcal{D}(\sigma_{1})$
in the sum. Further, applying Lemma 4 gives
\begin{align*}
 & \!\!\!\!\!D_{KL}\!\left(\boldsymbol{\pi}_{1}||\boldsymbol{\pi}_{2}\right)\\
 & =\ln\frac{\tilde{r}_{1}^{\mathbb{L}}}{\tilde{r}_{2}^{\mathbb{L}}}+\tilde{r}_{1}^{\mathbb{L}}\sum_{L\subseteq\mathcal{D}(\sigma_{1})}\frac{r_{1}^{L}}{\tilde{r}_{1}^{L}}\sum_{\ell\in L}\frac{\left\langle \boldsymbol{f}_{\!1}\!\ln\frac{\boldsymbol{f}_{1}}{\boldsymbol{f}_{2}}\right\rangle \!(\ell)}{r_{1}(\ell)/\tilde{r}_{1}(\ell)}\\
 & =\ln\frac{\tilde{r}_{1}^{\mathbb{L}}}{\tilde{r}_{2}^{\mathbb{L}}}+\tilde{r}_{1}^{\mathbb{L}}\sum_{\ell\in\mathcal{D}(\sigma_{1})}\!\!\frac{\left\langle \boldsymbol{f}_{1}\ln\frac{\boldsymbol{f}_{1}}{\boldsymbol{f}_{2}}\right\rangle \!(\ell)}{r_{1}(\ell)/\tilde{r}_{1}(\ell)}\!\!\!\sum_{L\subseteq\mathcal{D}(\sigma_{1})-\{\ell\}}\frac{r_{1}^{L\cup\{\ell\}}}{\tilde{r}_{1}^{L\cup\{\ell\}}}\\
 & =\ln\frac{\tilde{r}_{1}^{\mathbb{L}}}{\tilde{r}_{2}^{\mathbb{L}}}+\tilde{r}_{1}^{\mathbb{L}}\sum_{\ell\in\mathcal{D}(\sigma_{1})}\!\!\left\langle \boldsymbol{f}_{\!1}\!\ln\!\frac{\boldsymbol{f}_{\!1}}{\boldsymbol{f}_{\!2}}\right\rangle \!(\ell)\!\!\!\sum_{L\subseteq\mathcal{D}(\sigma_{1})-\{\ell\}}\frac{r_{1}^{L}}{\tilde{r}_{1}^{L}}\\
 & =\ln\frac{\tilde{r}_{1}^{\mathbb{L}}}{\tilde{r}_{2}^{\mathbb{L}}}+\!\sum_{\ell\in\mathbb{L}}\left\langle \!\boldsymbol{f}_{\!1}\!\ln\!\frac{\boldsymbol{f}_{\!1}}{\boldsymbol{f}_{\!2}}\!\right\rangle \!(\ell){\color{red}{\color{black}\tilde{r}_{1\!}(\ell}\mathclose{\color{black})}}\!\!\!\sum_{L\subseteq\mathbb{L}-\{\ell\}}\!r_{1}^{L}\tilde{r}_{1}^{(\mathbb{L}-\{\ell\})-L}\\
 & =\ln\frac{\tilde{r}_{1}^{\mathbb{L}}}{\tilde{r}_{2}^{\mathbb{L}}}+\!\sum_{\ell\in\mathbb{L}}{\color{red}{\color{black}\tilde{r}_{1}(\ell}\mathclose{\color{black})}\!}\left\langle \boldsymbol{f}_{\!1}\!\ln\!\frac{\boldsymbol{f}_{\!1}}{\boldsymbol{f}_{\!2}}\right\rangle \!(\ell)\\
 & =\sum_{\ell\in\mathbb{L}}\left[\ln\frac{\tilde{r}_{1}(\ell)}{\tilde{r}_{2}(\ell)}+{\color{red}{\color{black}\tilde{r}_{1}(\ell)\!}}\left\langle \boldsymbol{f}_{\!1}\!\ln\!\frac{\boldsymbol{f}_{\!1}}{\boldsymbol{f}_{\!2}}\right\rangle \!(\ell)\right],
\end{align*}
where the 2nd last equation follows from Lemma 1 and $r_{1}+\tilde{r}_{1}=1$.
\end{proof}
\textit{Remark}: Using $\tilde{r}_{i}(\ell)=\left(1+\left\langle \boldsymbol{f}_{i}\right\rangle \!(\ell)\right)^{-1}$,
the above KL-D can be written completely in terms of $\boldsymbol{f}_{1}$
and $\boldsymbol{f}_{2}$ as 
\begin{align*}
 & D_{KL}\!\left(\boldsymbol{\pi}_{1}||\boldsymbol{\pi}_{2}\right)=\sum_{\ell\in\mathbb{L}}\left[\ln\!\frac{1+\left\langle \boldsymbol{f}_{2}\right\rangle \!(\ell)}{1+\left\langle \boldsymbol{f}_{1}\right\rangle \!(\ell)}+\frac{\left\langle \boldsymbol{f}_{\!1}\!\ln\!\frac{\boldsymbol{f}_{\!1}}{\boldsymbol{f}_{\!2}}\right\rangle \!(\ell)}{1+\left\langle \boldsymbol{f}_{1}\right\rangle \!(\ell)}\right].
\end{align*}
Further, using $\boldsymbol{f}_{i}=r_{i}p_{i}/\tilde{r}_{i}$, we
can write
\begin{align*}
 & \!\!\!\!\!\!\left\langle \boldsymbol{f}_{1}\ln\frac{\boldsymbol{f}_{1}}{\boldsymbol{f}_{2}}\right\rangle \!(\ell)\\
 & =\int\frac{r_{1}(\ell)}{\tilde{r}_{1}(\ell)}p_{1}(x,\ell)\ln\left(\frac{r_{1}(\ell)\tilde{r}_{2}(\ell)p_{1}(x,\ell)}{\tilde{r}_{1}(\ell)r_{2}(\ell)p_{2}(x,\ell)}\right)dx\\
 & =\frac{r_{1}(\ell)}{\tilde{r}_{1}(\ell)}\int p_{1}(x,\ell)\left[\ln\frac{r_{1}(\ell)\tilde{r}_{2}(\ell)}{\tilde{r}_{1}(\ell)r_{2}(\ell)}+\ln\frac{p_{1}(x,\ell)}{p_{2}(x,\ell)}\right]dx\\
 & =\frac{r_{1}(\ell)}{\tilde{r}_{1}(\ell)}\left[\ln\frac{r_{1}(\ell)\tilde{r}_{2}(\ell)}{\tilde{r}_{1}(\ell)r_{2}(\ell)}+D_{KL}\!\left(p_{1}(\cdot,\ell)||p_{2}(\cdot,\ell)\right)\right].
\end{align*}
Hence, the KL-D can be written completely in terms of the LMB parameters
as follows
\begin{align*}
 & \!D_{KL}\!\left(\boldsymbol{\pi}_{1}||\boldsymbol{\pi}_{2}\right)=\\
 & \sum_{\ell\in\mathbb{L}}\ln\!\frac{\tilde{r}_{1\!}(\ell)}{\tilde{r}_{2}(\ell)}\!+\!{\color{red}{\color{black}r_{1\!}(\ell}\mathclose{\color{black})}\!}\left[\ln\frac{r_{1\!}(\ell)\tilde{r}_{2}(\ell)}{\tilde{r}_{1\!}(\ell)r_{2}(\ell)}+D_{KL}\!\left(p_{1}(\cdot,\ell)||p_{2}(\cdot,\ell)\right)\right]\!.
\end{align*}
Alternatively, using $\tilde{r}_{i}(\ell)=1\!-\!\left\langle \boldsymbol{v}_{i}\right\rangle \!(\ell)$
and $\boldsymbol{f}_{\!i}=\boldsymbol{v}_{i}/(1\!-\!\left\langle \boldsymbol{v}_{i}\right\rangle )$,
\begin{align*}
{\color{red}{\color{black}\tilde{r}_{1}(\ell}\mathclose{\color{black})}}\!\left\langle \!\boldsymbol{f}_{\!1}\!\ln\!\frac{\boldsymbol{f}_{\!1}}{\boldsymbol{f}_{\!2}}\!\right\rangle \!(\ell) & =\left\langle \!\boldsymbol{v}_{1}\!\ln\!\frac{\boldsymbol{v}_{1}(1-\left\langle \boldsymbol{v}_{2}\right\rangle )}{\boldsymbol{v}_{2}(1-\left\langle \boldsymbol{v}_{1}\right\rangle )}\!\right\rangle \!(\ell)\\
 & =\left\langle \!\boldsymbol{v}_{1}\!\ln\!\frac{\boldsymbol{v}_{1}}{\boldsymbol{v}_{2}}\!\right\rangle \!(\ell)-\left\langle \boldsymbol{v}_{1}\right\rangle \!(\ell)\ln\!\frac{1\!-\!\left\langle \boldsymbol{v}_{1}\right\rangle \!(\ell)}{1\!-\!\left\langle \boldsymbol{v}_{2}\right\rangle \!(\ell)},
\end{align*}
hence, the KL-D can be written in terms of the PHD as
\begin{align*}
 & \!\!\!\!\!\!\!\!D_{KL}\!\left(\boldsymbol{\pi}_{1}||\boldsymbol{\pi}_{2}\right)\\
 & =\sum_{\ell\in\mathbb{L}}\ln\frac{1\!-\!\left\langle \boldsymbol{v}_{1}\right\rangle \!(\ell)}{1\!-\!\left\langle \boldsymbol{v}_{2}\right\rangle \!(\ell)}-\sum_{\ell\in\mathbb{L}}\left\langle \boldsymbol{v}_{1}\right\rangle \!(\ell)\ln\!\frac{1\!-\!\left\langle \boldsymbol{v}_{1}\right\rangle \!(\ell)}{1\!-\!\left\langle \boldsymbol{v}_{2}\right\rangle \!(\ell)}\\
 & +\sum_{\ell\in\mathbb{L}}\left\langle \boldsymbol{v}_{1}\!\ln\!\frac{\boldsymbol{v}_{1}}{\boldsymbol{v}_{2}}\right\rangle \!(\ell)\\
 & =\sum_{\ell\in\mathbb{L}}\left[\ln\frac{1\!-\!\left\langle \boldsymbol{v}_{1}\right\rangle \!(\ell)}{1\!-\!\left\langle \boldsymbol{v}_{2}\right\rangle \!(\ell)}-\left\langle \boldsymbol{v}_{1}\right\rangle \!(\ell)\ln\!\frac{1\!-\!\left\langle \boldsymbol{v}_{1}\right\rangle \!(\ell)}{1\!-\!\left\langle \boldsymbol{v}_{2}\right\rangle \!(\ell)}\right]\\
 & +\sum_{\ell\in\mathbb{L}}\left\langle \boldsymbol{v}_{1}\!\ln\!\frac{\boldsymbol{v}_{1}}{\boldsymbol{v}_{2}}\right\rangle \!(\ell)\\
 & =\sum_{\ell\in\mathbb{L}}\left[(1\!-\!\left\langle \boldsymbol{v}_{1}\right\rangle )\!(\ell)\ln\frac{1\!-\!\left\langle \boldsymbol{v}_{1}\right\rangle \!(\ell)}{1\!-\!\left\langle \boldsymbol{v}_{2}\right\rangle \!(\ell)}+\!\left\langle \boldsymbol{v}_{1}\!\ln\!\frac{\boldsymbol{v}_{1}}{\boldsymbol{v}_{2}}\right\rangle \!(\ell)\right].
\end{align*}

\subsection{Proof of $\chi^{2}$ Divergence (\ref{eq:Chi-S-D})}
\begin{proof}
The $\chi^{2}$ Divergence ($\chi^{2}$-D) between the LMBs $\boldsymbol{\pi}_{1}$
and $\boldsymbol{\pi}_{2}$ is given by
\begin{align*}
D_{\chi^{2}}\!\left(\boldsymbol{\pi}_{1}||\boldsymbol{\pi}_{2}\right) & =\int\Delta(\boldsymbol{X})\frac{\left(K_{1}\boldsymbol{f}_{1}^{\boldsymbol{X}}\right)^{2}}{K_{2}\boldsymbol{f}_{2}^{\boldsymbol{X}}}\delta\boldsymbol{X}-1\\
 & =\frac{K_{1}^{2}}{K_{2}}\int\Delta(\boldsymbol{X})\frac{\left(\boldsymbol{f}_{1}^{2}\right)^{\boldsymbol{X}}}{\boldsymbol{f}_{2}^{\boldsymbol{X}}}\delta\boldsymbol{X}-1\\
 & =\frac{K_{1}^{2}}{K_{2}}\left(1+\left\langle \frac{\boldsymbol{f}_{1}^{2}}{\boldsymbol{f}_{2}}\right\rangle \right)^{\mathbb{L}}-1\\
 & =\left[\frac{\tilde{r}_{1}^{2}}{\tilde{r}_{2}}\left(1+\left\langle \frac{\boldsymbol{f}_{1}^{2}}{\boldsymbol{f}_{2}}\right\rangle \right)\right]^{\mathbb{L}}-1\\
 & =\prod_{\ell\in\mathbb{L}}\frac{\tilde{r}_{1}^{2}(\ell)}{\tilde{r}_{2}(\ell)}\left[1+\left\langle \frac{\boldsymbol{f}_{1}^{2}}{\boldsymbol{f}_{2}}\right\rangle \!(\ell)\right]-1,
\end{align*}
where the 3rd last equation follows from Lemma 2.
\end{proof}
\textit{Remark}: The above $\chi^{2}$-D can be written completely
in terms of $\boldsymbol{f}_{1}$, $\boldsymbol{f}_{2}$, or the LMB
parameters, or the PHDs, as follows 
\begin{align*}
D_{\chi^{2}}\!\left(\boldsymbol{\pi}_{1}||\boldsymbol{\pi}_{2}\right) & =\left[\frac{1+\left\langle \boldsymbol{f}_{2}\right\rangle }{(1+\left\langle \boldsymbol{f}_{1}\right\rangle )^{2}}\!\left(1+\left\langle \frac{\boldsymbol{f}_{1}^{2}}{\boldsymbol{f}_{2}}\right\rangle \right)\right]^{\mathbb{L}}-1\\
 & =\left[\frac{\tilde{r}_{1}^{2}}{\tilde{r}_{2}}\!\left(1+\frac{r_{1}^{2}\tilde{r}_{2}}{\tilde{r}_{1}^{2}r_{2}}\left\langle \frac{p_{1}^{2}}{p_{2}}\right\rangle \right)\right]^{\mathbb{L}}-1\\
 & =\left[\frac{(1\!-\!\left\langle \boldsymbol{v}_{1}\right\rangle )^{2}}{1\!-\!\left\langle \boldsymbol{v}_{2}\right\rangle }\!\left(\!1+\left\langle \!\frac{\boldsymbol{v}_{1}^{2}(1\!-\!\left\langle \boldsymbol{v}_{2}\right\rangle )}{\boldsymbol{v}_{2}(1\!-\!\left\langle \boldsymbol{v}_{1}\right\rangle )^{2}}\!\right\rangle \!\right)\right]^{\mathbb{L}}\!-\!1.
\end{align*}

\subsection{Proof of Cauchy-Schwarz Divergence (\ref{eq:CS-D-LMB}) }
\begin{proof}
For the multi-object exponentials $\boldsymbol{f}_{1}^{(\cdot)},\boldsymbol{f}_{2}^{(\cdot)}$,
let
\[
\left\langle \Delta\boldsymbol{f}_{1}^{(\cdot)},\Delta\boldsymbol{f}_{2}^{(\cdot)}\right\rangle _{U}=\int\Delta(\boldsymbol{X})U^{|\boldsymbol{X}|}\boldsymbol{f}_{1}^{\boldsymbol{X}}\boldsymbol{f}_{2}^{\boldsymbol{X}}\delta\boldsymbol{X},
\]
where $U$ is the unit of hyper-volume. Applying Lemma 2 gives
\begin{align*}
\left\langle \Delta\boldsymbol{f}_{1}^{(\cdot)},\Delta\boldsymbol{f}_{2}^{(\cdot)}\right\rangle _{U} & =\int\Delta(\boldsymbol{X})\left(U\boldsymbol{f}_{1}\boldsymbol{f}_{2}\right)^{\boldsymbol{X}}\delta\boldsymbol{X}\\
 & =(1+\left\langle U\boldsymbol{f}_{1}\boldsymbol{f}_{2}\right\rangle )^{\mathbb{L}}.
\end{align*}

Using the above result, the Cauchy-Schwarz Divergence (CS-D) between
two LMBs $\boldsymbol{\pi}_{1}$ and $\boldsymbol{\pi}_{2}$ can be
written as
\begin{align*}
 & \!\!\!\!\!\!\!D_{CS}\!\left(\boldsymbol{\pi}_{1},\boldsymbol{\pi}_{2}\right)\\
= & -\ln\frac{\left\langle K_{1}\Delta\boldsymbol{f}_{1}^{(\cdot)},K_{2}\Delta\boldsymbol{f}_{2}^{(\cdot)}\right\rangle _{U}}{\left\langle K_{1}\Delta\boldsymbol{f}_{1}^{(\cdot)},K_{1}\Delta\boldsymbol{f}_{1}^{(\cdot)}\right\rangle _{U}^{1/2}\left\langle K_{2}\Delta\boldsymbol{f}_{2}^{(\cdot)},K_{2}\Delta\boldsymbol{f}_{2}^{(\cdot)}\right\rangle _{U}^{1/2}}\\
= & -\ln\frac{\left\langle \Delta\boldsymbol{f}_{1}^{(\cdot)},\boldsymbol{f}_{2}^{(\cdot)}\right\rangle _{U}}{\left\langle \Delta\boldsymbol{f}_{1}^{(\cdot)},\boldsymbol{f}_{1}^{(\cdot)}\right\rangle _{U}^{1/2}\left\langle \Delta\boldsymbol{f}_{2}^{(\cdot)},\boldsymbol{f}_{2}^{(\cdot)}\right\rangle _{U}^{1/2}}.\\
= & -\ln\frac{(1+\left\langle U\boldsymbol{f}_{1}\boldsymbol{f}_{2}\right\rangle )^{\mathbb{L}}}{\left[(1+\left\langle U\boldsymbol{f}_{1}\boldsymbol{f}_{1}\right\rangle )^{\mathbb{L}}\right]^{1/2}\left[(1+\left\langle U\boldsymbol{f}_{2}\boldsymbol{f}_{2}\right\rangle )^{\mathbb{L}}\right]^{1/2}}\\
= & -\ln\frac{(1+\left\langle U\boldsymbol{f}_{1}\boldsymbol{f}_{2}\right\rangle )^{\mathbb{L}}}{\left(\sqrt{1+\left\langle U\boldsymbol{f}_{1}\boldsymbol{f}_{1}\right\rangle }\right)^{\mathbb{L}}\left(\sqrt{1+\left\langle U\boldsymbol{f}_{2}\boldsymbol{f}_{2}\right\rangle }\right)^{\mathbb{L}}}\\
= & -\ln\left(\frac{(1+\left\langle U\boldsymbol{f}_{1}\boldsymbol{f}_{2}\right\rangle )}{\sqrt{1+\left\langle U\boldsymbol{f}_{1}\boldsymbol{f}_{1}\right\rangle }\sqrt{1+\left\langle U\boldsymbol{f}_{2}\boldsymbol{f}_{2}\right\rangle }}\right)^{\mathbb{L}}\\
= & -\sum_{\ell\in\mathbb{L}}\ln\left(\frac{(1+\left\langle U\boldsymbol{f}_{1}\boldsymbol{f}_{2}\right\rangle (\ell))}{\sqrt{1+\left\langle U\boldsymbol{f}_{1}\boldsymbol{f}_{1}\right\rangle (\ell)}\sqrt{1+\left\langle U\boldsymbol{f}_{2}\boldsymbol{f}_{2}\right\rangle (\ell)}}\right)\\
= & -\sum_{\ell\in\mathbb{L}}\ln\!\frac{1+\left\langle U\!\boldsymbol{f}_{\!1}\boldsymbol{f}_{\!2}\right\rangle \!(\ell)}{\sqrt{1\!+\!\left\langle U\!\boldsymbol{f}_{\!1}^{2}\right\rangle \!(\ell)}\sqrt{1\!+\!\left\langle U\!\boldsymbol{f}_{\!2}^{2}\right\rangle \!(\ell)}}.
\end{align*}
Note that the logarithms in the above sum are well-defined for all
$\ell\in\mathbb{L}$.
\end{proof}
\textit{Remark}: While the CS-D depends on the unit $U$, the Bhattacharyya
distance $d_{B}\!\left(\boldsymbol{\pi}_{1},\boldsymbol{\pi}_{2}\right)=D_{CS}\!\left(\sqrt{\boldsymbol{\pi}_{1}},\sqrt{\boldsymbol{\pi}_{2}}\right)$
is invariant to it, because the unit of $\sqrt{\boldsymbol{\pi}_{1}(\boldsymbol{X})\boldsymbol{\pi}_{2}(\boldsymbol{X})}$
cancels out the unit of $\delta\boldsymbol{X}$. Noting that $\left\langle \sqrt{\boldsymbol{\pi}_{i}},\sqrt{\boldsymbol{\pi}_{i}}\right\rangle =\int\boldsymbol{\pi}_{i}\left(\boldsymbol{X}\right)\delta\boldsymbol{X}=1$,
the Bhattacharyya distance between two LMBs is
\begin{align*}
d_{B}\!\left(\boldsymbol{\pi}_{1},\boldsymbol{\pi}_{2}\right) & =-\ln\left\langle \!\sqrt{\boldsymbol{\pi}_{1}},\sqrt{\boldsymbol{\pi}_{2}}\right\rangle \\
 & =-\ln\left\langle \!\sqrt{K_{1}\Delta\boldsymbol{f}_{1}^{(\cdot)}},\sqrt{K_{2}\Delta\boldsymbol{f}_{2}^{(\cdot)}}\right\rangle \\
 & =-\ln\left(\!\sqrt{K_{1}K_{2}}\left\langle \sqrt{\Delta\boldsymbol{f}_{1}^{(\cdot)}},\sqrt{\Delta\boldsymbol{f}_{2}^{(\cdot)}}\right\rangle \right)\\
 & =-\ln\left(\!\sqrt{K_{1}K_{2}}\int\!\Delta(\boldsymbol{X})\sqrt{\boldsymbol{f}_{1}^{\boldsymbol{X}}\boldsymbol{f}_{2}^{\boldsymbol{X}}}\delta\boldsymbol{X}\right)\\
 & =-\ln\left(\left(\!\sqrt{\tilde{r}_{1}\tilde{r}_{2}}\right){}^{\mathbb{L}}\int\!\Delta(\boldsymbol{X})\left(\sqrt{\boldsymbol{f}_{1}\boldsymbol{f}_{2}}\right)^{\boldsymbol{X}}\!\delta\boldsymbol{X}\right)\\
 & =-\ln\left(\left(\!\sqrt{\tilde{r}_{1}\tilde{r}_{2}}\right){}^{\mathbb{L}}\left(1+\left\langle \sqrt{\boldsymbol{f}_{1}\boldsymbol{f}_{2}}\right\rangle \right){}^{\mathbb{L}}\right)\\
 & =-\ln\left[\!\sqrt{\tilde{r}_{1}\tilde{r}_{2}}\left(1+\left\langle \sqrt{\boldsymbol{f}_{1}\boldsymbol{f}_{2}}\right\rangle \right)\right]{}^{\mathbb{L}}\\
 & =-\sum_{\ell\in\mathbb{L}}\ln\!\sqrt{\tilde{r}_{1}(\ell)\tilde{r}_{2}(\ell)}\left[1+\left\langle \sqrt{\boldsymbol{f}_{1}\boldsymbol{f}_{2}}\right\rangle \!(\ell)\right],
\end{align*}
where 3rd last line follows from Lemma 2. This is the R\'{e}nyi-D with
$\alpha=0.5$.
\end{document}